\pdfoutput=1
\documentclass[12pt]{article}

\usepackage[margin=1in]{geometry}

\usepackage{amsmath, amssymb, mathtools}

\usepackage{graphicx}
\usepackage{booktabs}
\usepackage{multirow}
\usepackage{makecell}
\usepackage{tabularx}
\usepackage{longtable}
\usepackage{threeparttablex}
\usepackage{subcaption}   % replaces deprecated subfigure
\usepackage{siunitx}

\usepackage{microtype}
\usepackage{xcolor}
\usepackage{bm}
\usepackage{setspace}
\usepackage[flushmargin]{footmisc}
\usepackage{natbib}       % load BEFORE hyperref
\usepackage[hidelinks]{hyperref} % load AFTER natbib
\hypersetup{
    colorlinks=true,
    linkcolor=blue,
    citecolor=blue,
    urlcolor=blue,
}

\usepackage{algorithm}
\usepackage{algpseudocode}

\usepackage{listings}

\usepackage{pdflscape}

\makeatletter
\renewenvironment{abstract}{
  \vspace{0.5\baselineskip}
  \centerline{\normalfont\bfseries\large Abstract}
  \begin{list}{}{
    \setlength{\leftmargin}{0.25in}
    \setlength{\rightmargin}{0.25in}
  }\item\relax
  \fontsize{10pt}{11pt}\selectfont
}{\end{list}}
\makeatother

\usepackage{amsthm}
\numberwithin{equation}{section}
\newtheorem{theorem}{Theorem}[section]
\newtheorem{lemma}{Lemma}[section]
\newtheorem{proposition}{Proposition}[section]

\theoremstyle{definition}

\theoremstyle{remark}

\usepackage{authblk}
\title{\Large Testing composite null hypotheses with high-dimensional dependent data: a computationally scalable FDR-controlling procedure}

\author{\normalsize Pengfei Lyu$^1$, Xianyang Zhang$^2$ and Hongyuan Cao$^{3*}$}

\footnotetext[1]{\mbox{Department of Biostatistics \& Bioinformatics, Duke University, Durham, NC 27705, USA.}}
\footnotetext[2]{\mbox{Department of Statistics, Texas A\&M University, College Station, TX, 77843, USA.}}
\footnotetext[3]{\mbox{Department of Biostatistics and Medical Informatics, University of Wisconsin-Madison, Madison, WI 53706, USA.}}
\begingroup

\footnotetext{\mbox{Corresponding author: \href{mailto:hongyuancao@gmail.com}{hongyuancao@gmail.com}}}
\endgroup

\date{}

\begin{document}
\maketitle
\vspace*{-1.5em}
\pagestyle{plain}

\begin{abstract}
Testing composite null hypotheses is fundamental to many scientific applications, including mediation and replicability analyses, and becomes particularly challenging in high-throughput settings involving tens of thousands of features. Existing high-dimensional composite null hypotheses testing often ignores the dependence structure among features, leading to overly conservative or liberal results. To address this limitation, we develop a four-state hidden Markov model (HMM) for bivariate $p$-value sequences arising from two-study replicability analysis. This model captures local dependence among features and accommodates study-specific heterogeneity. Based on the HMM, we propose a multiple testing procedure that asymptotically controls the false discovery rate (FDR). Extending this framework to more than two studies is computationally intensive, with complexity growing exponentially in the number of studies $n.$ To address this scalability issue, we introduce a novel e-value framework that reduces computational complexity to quadratic in $n,$ while preserving asymptotic FDR control. Extensive simulations demonstrate that our method achieves higher power than existing approaches at the same FDR levels. When applied to genome-wide association studies (GWAS), the proposed approach identifies replicable SNP-level signals that are not detected at the same significance threshold by competing methods.
\end{abstract}

\begin{keywords}
  Composite null hypotheses, e-values, false discovery rate, hidden Markov model, high dimension, non-parametric maximum likelihood estimation
\end{keywords}
\newpage
\section{Introduction}

Composite null hypotheses frequently arise in modern statistical applications. One example is mediation analysis, where researchers aim to understand the mechanisms by which an exposure influences an outcome through intermediate variables, or mediators \citep{mackinnon2012introduction, sun2024testing}. The null hypothesis in this context is composite, comprising three distinct subspaces: (i) the exposure does not affect the mediator, and the mediator does not affect the outcome; (ii) the exposure influences the mediator, but the mediator has no effect on the outcome; and (iii) the mediator affects the outcome, but is not affected by the exposure. 
% {\color{red}
Crucially, this structural challenge also underpins replicability analysis, which aims to detect features showing consistent association across multiple independent studies. This paper addresses the broader challenge of high-dimensional composite null testing by focusing on replicability analysis in high-throughput experiments, introducing a statistically rigorous and computationally scalable framework.

A large body of literature has been developed for high-dimensional replicability analysis under the assumption of independent hypotheses. A simple \textit{ad hoc} approach applies the Benjamini--Hochberg (BH) procedure \citep{benjamini1995controlling} separately to each study and intersects the discoveries, but this generally fails to control the false discovery rate (FDR). To guarantee valid FDR control, \cite{benjamini2009selective} proposed using the maximum $p$-value across studies as a test statistic under the composite null, although this approach is often conservative. Subsequent works have improved power through empirical-Bayes modeling \citep{heller2014replicability}, cross-screening strategies \citep{bogomolov2018assessing}, nonparametric estimation \citep{zhao2020nonparametric}, refined null approximations \citep{lyu2023jump}, and conditionally symmetric Gaussian mixture models \citep{sun2024testing}. 
% Related developments for partial conjunction testing include \cite{benjamini2008screening}, \cite{wang2022detecting}, \cite{liang2022powerful}, and \cite{deng2024joint}. {\color{red}Partial conjunction testing asks whether a hypothesis is non-null in at least $u$ out of $n$ studies for some $1 \le u \le n$; in contrast, the full replicability target considered in this paper requires evidence of a signal in all $n$ studies ($u = n$), which corresponds to a strictly composite null and a different scientific question.} Comprehensive reviews are provided by \cite{bogomolov2023replicability}.
A related line of work focuses on partial conjunction testing \citep{benjamini2008screening, wang2022detecting, liang2022powerful, deng2024joint}, which assesses whether a feature is non-null in at least $u$ out of $n$ studies ($1 \le u \le n$). In contrast, the full replicability target considered in this paper requires evidence of a signal in all $n$ studies ($u = n$). This corresponds to a strictly composite null and addresses a distinct scientific question. Comprehensive reviews are provided by \cite{bogomolov2023replicability}. Several recent methods address related composite-hypothesis problems beyond classical replicability analysis. PLACO \citep{ray2020powerful} targets two-study testing under a composite null but relies on parametric modeling and does not account for data dependence. Primo \citep{gleason2020primo} and QCH \citep{mary2022querying,de2025large} provide more general frameworks for integrating multiple sets of $p$-values or testing composed hypotheses, but they do not explicitly model local dependence among nearby genetic variants. AMDP \citep{ding2023amdp} addresses high-dimensional mediation analysis, whose structure differs from the composite null hypotheses testing considered here.

% Despite these advances, most existing replicability methods assume independence among hypotheses or rely on restrictive dependence assumptions.

Despite these advances, most existing replicability methods relying on independence assumptions may fail to control the FDR or suffer from a loss of power when applied to dependent datasets. In high-throughput experiments, feature dependence is the norm rather than the exception. In genome-wide association studies (GWAS), nearby genetic variants are often correlated because of linkage disequilibrium (LD), a phenomenon in which variants close to each other on the genome tend to be inherited together \citep{visscher2012five}. This local correlation among single-nucleotide polymorphisms (SNPs) is a common source of feature dependence in GWAS.
% {\color{red} 
Ignoring such structural dependence can substantially reduce statistical power or invalidate inferential procedures. While some existing approaches accommodate composite null testing under dependence, they typically rely on strict Gaussian assumptions \citep{sun2024testing} or weak forms of positive dependence \citep{bogomolov2023testing}, both of which may be inadequate for complex, large-scale genomic applications.

Hidden Markov models (HMMs) provide a natural framework for capturing local dependence structures in large-scale inference problems. In genomic association data, HMMs have been widely used to capture LD-induced dependence \citep{li2003modeling,sun2009large,sesia2021false,abraham2022multiple}. This modeling framework is further supported by a rigorous body of parametric and nonparametric theory %theoretical properties of HMMs have been extensively studied 
\citep{leroux1992maximum,bickel1998asymptotic,alexandrovich2016nonparametric}. In particular, \cite{abraham2022multiple} developed empirical-Bayes multiple testing procedures under nonparametric two-state HMMs for single-sequence testing problems. Another dependence-aware approach, Cartesian HMM \citep{wang2019replicability}, was developed for two-study replicability analysis. However, this approach has three key limitations: the densities are modeled parametrically; replicability testing is not formulated explicitly as a composite null problem, and no consistency theory for the HMM estimators or asymptotic FDR control is provided; and the method is limited to two studies.

In this paper, we propose CoHiM (\underline{Co}mposite null hypotheses testing under \underline{Hi}dden \underline{M}arkov models), a dependence-aware framework for testing high-dimensional composite null hypotheses. We begin with the two-study setting and model the latent signal configurations using a four-state HMM corresponding to the four possible combinations of study-specific signal indicators. Unlike existing dependence-aware composite null hypotheses testing methods, such as Cartesian HMM, which imposes parametric assumptions on the non-null $p$-value distributions, CoHiM estimates the non-null densities nonparametrically under a mild monotonicity condition for each study. Based on the proposed HMM framework, we apply a nonparametric maximum likelihood estimation procedure integrating the forward-backward algorithm \citep{baum1970maximization}, the expectation-maximization algorithm \citep{dempster1977maximum}, and the pool-adjacent-violator algorithm \citep{robertson1988order}. The forward-backward, EM, and PAVA algorithm components are established algorithmic tools; CoHiM's methodological contribution lies in adapting them to a four-state composite-null HMM, together with establishing consistency theory and FDR-control guarantees. We further extend CoHiM to the multiple-study setting. Directly modeling all latent signal configurations across many studies requires exponentially growing state spaces and quickly becomes computationally infeasible. To overcome this challenge, we construct test statistics for all study pairs and transform them into e-values. A valid e-value has expectation no greater than one under the null hypothesis, with larger values providing stronger evidence against the null \citep{vovk2021values}. We then aggregate the resulting e-values and apply the e-BH procedure \citep{wang2022false} to obtain a scalable composite null hypotheses testing procedure whose computational complexity grows quadratically with the number of studies. 

Our work makes several contributions. First, we develop a dependence-aware framework for testing high-dimensional composite null hypotheses using multi-state HMMs, allowing both feature dependence and cross-study heterogeneity. Second, unlike existing HMM-based multiple testing methods that focus on simple null hypotheses under two-state latent models,
% , such as \citet{abraham2022multiple} who established empirical-Bayes FDR analysis for nonparametric two-state HMMs in single-sequence testing, 
our framework addresses composite nulls induced by four-state latent structures, where the null hypothesis comprises three distinct hidden states rather than a single null state, and inference targets the posterior probability of a union of hidden states. Third, we establish consistency of the estimated stationary probabilities, transition probabilities, and nonparametric density estimators under the proposed four-state HMM framework, and prove asymptotic FDR control of the resulting procedure. Finally, we develop a scalable extension to multiple studies through pairwise e-value aggregation and e-BH, avoiding the exponential computational burden of directly modeling all latent configurations simultaneously.
We summarize different methods in Table \ref{tab:method_comparison}.
% }
All simulation and data analysis results are fully reproducible, with code available at \url{https://github.com/hongyuan-cao/CoHiM}.

\begin{table*}[htbp]
\centering
\caption{Different methods for high-dimensional composite null hypotheses testing.
\textbf{Dep.}: whether the method explicitly models dependence among features.
\textbf{Nonp.}: whether the method is nonparametrically based.
% ; purely $p$-value- or rank-based procedures (e.g., \textit{ad hoc} BH, MaxP, AdaFilter) qualify because they require no parametric non-null modeling, 
% where ``Partial'' denotes copula- or mixture-based modeling that retains some parametric structure on the non-null signal.
\textbf{Multi-study}: whether the method applies to more than $2$ studies.
\textbf{FDR}: F-S\,=\,finite-sample guarantee; Asym.\,=\,asymptotic guarantee; No\,=\,no formal FDR guarantee.}
\label{tab:method_comparison}
\scriptsize
\setlength{\tabcolsep}{4pt}
\begin{tabularx}{\textwidth}{@{}
  p{2.1cm}
  >{\raggedright\arraybackslash}X
  c c c
  >{\centering\arraybackslash}p{1.4cm}
  >{\raggedright\arraybackslash}X
@{}}
\toprule
\textbf{Method} &
\textbf{Setting} &
\textbf{Dep.} &
\textbf{Nonp.} &
\textbf{Multi-study} &
\textbf{FDR} &
\textbf{Limitation} \\
\midrule
\textit{ad hoc} BH
  & $n/n$ replicability via BH intersection
  & $\times$ & \checkmark & \checkmark & No
  & No formal FDR guarantee \\[2pt]
MaxP
  & $n/n$ replicability, max $p$-value
  & $\times$ & \checkmark & \checkmark & F-S
  & Overly conservative\\[2pt]
MaRR
  & $n/n$ reproducibility, maximum rank
  & $\times$ & \checkmark & \checkmark & No
  & Rank-based; relies on a strong concordance structure; no dependence modeling \\[2pt]
radjust
  & $2/2$ replicability, cross-screening
  & $\times$ & $\times$ & $\times$ & F-S
  & Restricted to two studies\\[2pt]
JUMP
  & $2/2$ replicability
  & $\times$ & \checkmark & $\times$ & No
  & No dependence modeling \\[2pt]
STAREG
  & $2/2$ replicability, PAVA-based
  & $\times$ & \checkmark & $\times$ & Asym.
  & Primarily two studies \\[2pt]
AdaFilter
  & $u/n$ partial conjunction
  & $\times$ & \checkmark & \checkmark & F-S
  & No dependence modeling \\[2pt]
QCH
  & General composite null
  & $\times$ & Partial & \checkmark & Asym.
  & No dependence modeling \\[2pt]
AMDP
  & Mediation analysis composite null
  & $\times$ & Partial & $\times$ & Asym.
  & Not applicable to $n/n$ replicability \\[2pt]
Primo
  & General composite null
  & $\times$ & Partial & \checkmark & No
  & No formal FDR guarantee \\[2pt]
PLACO
  & $2/2$  composite null
  & $\times$ & $\times$ & $\times$ & Asym.
  & Two studies; parametric; no dependence modeling \\[2pt]
Cartesian HMM
  & $2/2$ composite null
  & \checkmark & $\times$ & $\times$ & No
  & Parametric assumptions; two studies \\
% \midrule
\textbf{CoHiM}
  & General composite null 
  & \checkmark & \checkmark & \checkmark & Asym.
  & Markovian assumption\\
\bottomrule
\end{tabularx}
\end{table*}

\section{Methodology}\label{sec_meth}
\subsection{Problem setup}\label{subsec_model}
We begin by considering the case of two studies. Throughout this section, we use GWAS as a motivating and illustrative example, although the proposed method CoHiM is broadly applicable to high-dimensional testing. Let $(y_{1j}, y_{2j})_{j=1}^m$ denote the paired $p$-values for $m$ hypotheses across two studies. For example, these may arise when testing marginal associations between SNPs and a phenotype, such as heart disease, in two different populations \citep{diabetes2012large}. Our goal is to identify SNPs that exhibit replicable association with the phenotype across both studies. Let $\theta_{ij}$ denote the latent state of the $j$th SNP in study $i$, where $\theta_{ij}=1$ indicates true association and $\theta_{ij}=0$ indicates no association. Define $s_j = 0, 1, 2, 3$ as the joint latent state corresponding to $(\theta_{1j}, \theta_{2j}) = (0, 0), (0, 1), (1, 0)$ and $(1, 1)$, respectively. The replicability null hypothesis is composite and can be specified as
\begin{equation}
    H_{0j}: s_j\in\{0,1,2\} \quad \text{for }j=1,\dots,m,\notag
\end{equation}
i.e., the SNP is not associated with the phenotype in at least one study. Rejecting $H_{0j}$ implies evidence for a replicable association, that is, $\theta_{1j} = \theta_{2j}=1.$

To capture local dependence across hypotheses, we assume that the latent state sequence $\boldsymbol s=(s_1, \dots,s_m)$ follows a stationary, irreducible, and aperiodic four-state Markov chain. The transition probabilities are defined as 
\begin{align}
    \label{eq_a_def}
    a_{k\ell} =  \mathbb P(s_{j+1}=\ell\mid s_j=k) \text{ for }j = 1,\ldots, m-1, \text{ and } k,\ell=0,1,2,3,
\end{align}
subject to the normalization condition $\sum_{\ell=0}^3a_{k\ell}=1$ for each $k$. The stationary probabilities of the Markov chain are given by $\pi = (\pi_0, \pi_1, \pi_2, \pi_3),$ where $\pi_k = \mathbb P(s_j=k)$ and $\sum_{k=0}^3\pi_k=1$. Let $A = (a_{k\ell})\in \mathbb R^{4\times 4}$ denote the transition probability matrix. By stationarity, the stationary distribution satisfies the equation $$\pi A = \pi.$$

Conditional on the hidden states, we assume a two-component mixture model for the $p$-values in each study:
\begin{equation}\label{eq_mixture}
    \begin{aligned}
        y_{1j}\mid \theta_{1j} \sim (1-\theta_{1j})f_0 + \theta_{1j}f_1,\\
        y_{2j}\mid \theta_{2j} \sim (1-\theta_{2j})f_0 + \theta_{2j}f_2,\\
    \end{aligned}
\end{equation}
where $f_0$ is the null density, and $f_1$ and $f_2$ are the non-null densities for study 1 and study 2, respectively. We assume $f_0$ is the standard uniform density on $[0, 1],$ and we impose the following monotone likelihood ratio condition \citep{sun2007oracle, cao2013optimal}:
\begin{equation}\label{monotone}
    f_1(x)/f_0(x) \text{ and } f_2(x)/f_0(x) \text{ are monotonically non-increasing in }x.
\end{equation}
This condition is natural, as smaller $p$-values provide stronger evidence against the null hypothesis. We assume that, conditional on the hidden states, the two studies are independent. Therefore, for $j = 1, \ldots, m$, the paired $p$-values $(y_{1j}, y_{2j})$ are conditionally independent given the joint latent states $s_j.$ 
Let $f^{(s_j)}$ denote the conditional density function of $(y_{1j}, y_{2j})$ given the latent state $s_j$. By the mixture model in (\ref{eq_mixture}), $f^{(s_j)}$ satisfies
\begin{align}
    \label{eq_f_sj}
    f^{(s_j)}(y_{1j}, y_{2j}) = \begin{cases}
        f_0(y_{1j})f_0(y_{2j}) \quad \text{ if } s_j = 0,\\
        f_0(y_{1j})f_2(y_{2j}) \quad \text{ if } s_j = 1,\\
        f_1(y_{1j})f_0(y_{2j}) \quad \text{ if } s_j = 2,\\
        f_1(y_{1j})f_2(y_{2j}) \quad \text{ if } s_j = 3.
    \end{cases}
\end{align}

\subsection{Estimation and testing procedure}
% {\color{red}
    Denote the true parameter as $\phi^* = (\pi^*, A^*, f_1^* , f_2^*)$. We obtain the maximum likelihood estimator $\widehat{\phi}_m = (\widehat{\pi}, \widehat{A}, \widehat{f}_1, \widehat{f}_2)$ via the Expectation-Maximization (EM) algorithm \citep{dempster1977maximum}, utilizing the forward-backward algorithm \citep{baum1970maximization} and the pool-adjacent-violators algorithm (PAVA) \citep{robertson1988order}. The PAVA update for monotone non-null p-value densities follows the estimation strategy in STAREG \citep{li2024stareg}, and it is adapted here to the four-state HMM. The details can be found in Section~\ref{subsec_est} in the Supplementary Materials.
% }

Define the forward probability $\alpha_j(s_j) = \mathbb P_{\phi^*}((y_{1t},y_{2t})_{t=1}^j, s_j)$ for $j = 1, \ldots, m$ and the backward probability $\beta_j(s_j) =\mathbb P_{\phi^*}((y_{1t},y_{2t})_{t=j+1}^m\mid s_j)$ for $j = 1, \ldots, m-1$, initialized by $\alpha_1(s_1) = \pi_{s_1}f^{(s_1)}(y_{11},y_{21})$ and $\beta_m(s_m) = 1.$ By the Markov property and the transition probability defined in \eqref{eq_a_def}, these quantities can be computed recursively: 
\begin{align*}
    \alpha_{j+1}(s_{j+1}) =& \sum_{s_j=0}^3\alpha_j(s_j)a_{s_js_{j+1}}f^{(s_{j+1})}(y_{1,j+1},y_{2,j+1}), \quad \text{ and }
    \\
    \beta_j(s_{j})=&\sum_{s_{j+1}=0}^3\beta_{j+1}(s_{j+1})f^{(s_{j+1})}(y_{1,j+1},y_{2,j+1}) a_{s_js_{j+1}}.
\end{align*}

With the estimator $\widehat{\phi}_m = (\widehat{\pi}, \widehat{A}, \widehat{f}_1, \widehat{f}_2)$, we compute the estimated forward and backward probabilities for $j = 1, \ldots, m-1$ as follows:
\begin{align}
    \widehat{\alpha}_1(s_1) =& \widehat{\pi}_{s_1}\widehat{f}^{(s_1)}(y_{11}, y_{21}),
    \quad \quad \widehat{\beta}_m(s_m) = 1,
    \label{eq_alpha_beta_est_initial}
    \\
    \widehat{\alpha}_{j+1}(s_{j+1}) =& \sum_{s_j = 0}^3 \widehat{\alpha}_j(s_j) \widehat{a}_{s_j, s_{j+1}} \widehat{f}^{(s_{j+1})}(y_{1,j+1}, y_{2,j+1})\quad \text{ and }
    \label{eq_alpha_est_update}
    \\
    \widehat{\beta}_j(s_j) =& \sum_{s_{j+1} = 0}^3 \widehat{\beta}_{j+1}(s_{j+1}) \widehat{a}_{s_j, s_{j+1}} \widehat{f}^{(s_{j+1})}(y_{1,j+1}, y_{2,j+1}).
    \label{eq_beta_est_update}
\end{align}
For $j = 1, \ldots, m$, define replicability Local Index of Significance (rLIS) as the posterior probability that the $j$th hypothesis is not replicable, i.e., $s_j$ belongs to the non-replicable configuration set $\{0,1,2\}$, given all observed $p$-value pairs:
\begin{align*}
    {\rm rLIS}_j =& \mathbb P_{\phi^*}\left(s_j\in\{0,1,2\}\mid (y_{1j'}, y_{2j'})_{j'=1}^m\right)
    \\
    =&\frac{\sum_{s_j=0}^2 \alpha_j(s_j)\beta_j(s_j)}{\sum_{s_j=0}^3 \alpha_j(s_j)\beta_j(s_j)}.
    %\quad \text{ for }j = 1, \ldots, m.
\end{align*}
Using the estimated forward and backward probabilities in \eqref{eq_alpha_beta_est_initial}-\eqref{eq_beta_est_update}, the estimated rLIS is 
\begin{align}
    \widehat{\mathrm{rLIS}}_j =& \mathbb P_{\widehat{\phi}_m} \left(s_j\in\{0, 1, 2\}\mid (y_{1j'}, y_{2j'})_{j'=1}^m\right)= \frac{\sum_{s_j=0}^2 \widehat{\alpha}_j(s_j)\widehat{\beta}_j(s_j)}{\sum_{s_j=0}^3 \widehat{\alpha}_j(s_j)\widehat{\beta}_j(s_j)}.\label{eq_test_stats}
\end{align}
To implement the data-driven step-up procedure, we first order the estimated replicability Local Index of Significance values $\widehat{\mathrm{rLIS}}_{(1)}\leq \cdots \leq \widehat{\mathrm{rLIS}}_{(m)}$ with the corresponding replicability null hypotheses denoted by $H_{0(1)}, \ldots, H_{0(m)}$. Given a target FDR level $q\in (0, 1)$, we have the step-up procedure
\begin{align}
    \begin{aligned}
        &\widehat{R} = \max\left\{r: \frac{1}{r}\sum_{j=1}^r \widehat{\mathrm{rLIS}}_{(j)} \leq q\right\},\\
        & \text{and reject } H_{0(j)} \quad \text{ for }j = 1,\ldots,\widehat R.
    \end{aligned}
    \label{eq_rej_procedure}
\end{align}

The full FDR-controlling procedure is summarized in Algorithm \ref{algo_2study}.
% \begin{algo}\label{algo1}
% CoHiM for the two-study case
% % \begin{algorithmic}[1]
% %     \label{algo_2study}
% \begin{tabbing}
%     \qquad\enspace\textbf{Input:} $p$-values from two studies $(y_{1j}, y_{2j})_{j=1}^m$, nominal FDR level $q$.
%     \\
%     \qquad\enspace Obtain the estimates $\widehat \phi = (\widehat \pi, \widehat A, \widehat f_1, \widehat f_2)$ from the EM algorithm in Section \ref{subsec_est}.
%     \\
%     \qquad\enspace Calculate the forward and backward probabilities $\widehat\alpha_j(s_j)$ and $\widehat\beta_j(s_j)$ for $j = 1,\ldots,m$ and \\\enspace $s_j = 0,1,2,3$ by (\ref{eq_alpha_beta_initial}), (\ref{eq_alpha_update}) and (\ref{eq_beta_update}).
%     \\
%     \qquad\enspace Calculate $\widehat{\rm rLIS}_j = \sum_{s_j=0}^2\widehat\alpha_j(s_j)\widehat\beta_j(s_j)/\sum_{s_j=0}^3\widehat\alpha_j(s_j)\widehat\beta_j(s_j)$.
%     \\
%     \qquad\enspace Order the rLIS values $\widehat{\rm rLIS}_{(1)}\leq \cdots \leq \widehat{\rm rLIS}_{(m)}$ with the corresponding null hypotheses 
%     \\\enspace denoted by $H_{0(1)}, \ldots, H_{0(m)}$.
%     \\
%     \qquad\enspace Find the rejection number $\widehat R = \max\{r: r^{-1}\sum_{j=1}^r \widehat{\rm rLIS}_{(j)}\leq q\}.$
%     \\
%     \qquad\enspace \textbf{Output:} Reject $H_{0j}$  if $\widehat{\rm rLIS}_j\leq \widehat{\rm rLIS}_{(\widehat R)}$ ($j = 1,\ldots,m$).
% \end{tabbing}
% % \end{algorithmic}
% \end{algo}
\begin{algorithm}[htbp]
\caption{CoHiM for the two-study case}\label{algo_2study}
\begin{algorithmic}[1]
\State \textbf{Input: }$p$-values from two studies $(y_{1j}, y_{2j})_{j=1}^m$, nominal FDR level $q$.
\State Estimate $\widehat \phi = (\widehat \pi, \widehat A, \widehat f_1, \widehat f_2)$ from the EM algorithm. %in Section \ref{subsec_est}.
\State Compute forward and backward probabilities $\widehat\alpha_j(s_j)$ and $\widehat\beta_j(s_j)$ for $j = 1,\ldots,m$ and $s_j = 0,1,2,3$ by (\ref{eq_alpha_beta_est_initial}), (\ref{eq_alpha_est_update}) and (\ref{eq_beta_est_update}).
\State Compute $\widehat{\rm rLIS}_j = \sum_{s_j=0}^2\widehat\alpha_j(s_j)\widehat\beta_j(s_j)/\sum_{s_j=0}^3\widehat\alpha_j(s_j)\widehat\beta_j(s_j)$.
\State Order the rLIS values $\widehat{\rm rLIS}_{(1)}\leq \cdots \leq \widehat{\rm rLIS}_{(m)}$ with the corresponding replicability null hypotheses denoted by $H_{0(1)}, \ldots, H_{0(m)}$.
\State Let $\widehat R = \max\{r: r^{-1}\sum_{j=1}^r \widehat{\rm rLIS}_{(j)}\leq q\}.$
\State \textbf{Output: }Reject $H_{0j}$  if $\widehat{\rm rLIS}_j\leq \widehat{\rm rLIS}_{(\widehat R)}$ for $j = 1,\ldots,m$.
\end{algorithmic}
\end{algorithm}

\subsection{Generalization to more than two studies}
We extend our procedure to enable replicability analysis across multiple studies. Suppose we have $p$-values from $n$ studies ($n > 2$) for $m$ features, denoted by $(y_{ij})_{n\times m}$. For $i = 1, \ldots, n$ and $j = 1,\ldots, m$, let $\theta_{ij}$ denote the hidden binary state of the $j$th feature in study $i$, where $\theta_{ij} = 1$ indicates an association of the $j$th feature with the phenotype in study $i$ and $\theta_{ij} = 0$ indicates no association. Our goal is to identify features that are associated with the phenotype across all $n$ studies. The replicability null and alternative hypotheses for the $j$th feature are 
$$H_{0j}: \prod_{i=1}^n \theta_{ij} = 0 \quad \text{and} \quad H_{1j}: \prod_{i=1}^n \theta_{ij} = 1.$$ 

Because the joint hidden states $(\theta_{1j}, \ldots, \theta_{nj})$ have $2^n$ possible configurations, the composite null hypothesis $H_{0j}$ comprises $2^n - 1$ distinct states. Modeling the dependence across all $n$ studies using a standard HMM would therefore require a $2^n$-dimensional stationary probability vector $\pi$ and a $2^n \times 2^n$ transition matrix $A$. This exponential growth introduces a severe computational bottleneck as $n$ increases.

To alleviate the computational burden for large $n$, we propose a pairwise testing strategy. Specifically, we conduct replicability analysis on all pairs of studies, convert the test results into pairwise e-values, and subsequently aggregate them. For each pair of studies $1 \leq k < \ell \leq n$, let $(\widehat{\rm rLIS}_j^{k\ell})_{j=1}^m$ denote the estimated pairwise replicability Local Index of Significance (rLIS) values, and let $\widehat R^{k\ell}(t)$ denote the number of rejections at the pairwise FDR level $t \in (0, q]$.  Let $\mathcal H_0 = \{j: \prod_{i=1}^n \theta_{ij} = 0\}$ and $\mathcal H_0^{k\ell} = \{j: \theta_{kj}\theta_{\ell j} = 0\}$ represent the multi-study and pairwise null indices, respectively. An e-value is a random variable whose expectation is bounded by $1$ under the null hypothesis, where larger values signify stronger evidence against the null. Accordingly, we define the pairwise e-values for $1 \leq j \leq m$ as:\begin{align}\label{eq_evalue}\widehat e_{j}^{k\ell}(t)=\frac{m I\left(\widehat{\rm rLIS}j^{k\ell} \le \widehat{\rm rLIS}{(\widehat R^{k\ell}(t))}^{k\ell} \right)}{\sum_{j'=1}^{m}I\left(\widehat{\rm rLIS}{j'}^{k\ell}  \le \widehat{\rm rLIS}{(\widehat R^{k\ell}(t))}^{k\ell} \right)\widehat{\rm rLIS}_{j'}^{k\ell} }.\end{align} This construction is partly motivated by \cite{li2025note}, who demonstrated that most existing multiple testing procedures are equivalent to e-BH procedures \citep{wang2022false} when applied to an appropriately defined set of e-values.

For full $n/n$ replicability, a feature must demonstrate evidence of association across all $n$ studies, with the strength of this evidence quantified by e-values. Under this framework, if even a single pair of studies provides weak pairwise evidence, the global evidence must also be weak. Conversely, a large aggregated e-value is achieved only when every relevant pair exhibits sufficiently strong pairwise replicability. This logic motivates aggregating the pairwise e-values using a minimum operation.

The e-BH procedure for a list of e-values $(e_1, \ldots, e_m)$ with FDR nominal level $t$ proceeds by ordering the e-values from largest to smallest as $e_{(1)} \geq \cdots \geq e_{(m)}$ and rejecting the hypotheses corresponding to the top $R$ e-values, where $R = \max\{r: e_{(r)} \geq m/(rt)\}$. 
% {\color{red}
The key requirement for applying e-BH is that the constructed e-values have average expectation at most one over the null hypotheses. In our setting, we show asymptotically that
\begin{align*}
    \frac{1}{m}\sum_{j\in \mathcal H_0} \mathbb E(\widehat e_j) \le 1,
\end{align*}
which permits the use of the e-BH argument of \citet{wang2022false}. The proof is given in the Appendix.
% }
Following the e-BH procedure, we sort the pairwise e-values in (\ref{eq_evalue}) in descending order as $\widehat e_{(1)}^{k\ell}(t)  \geq \cdots \geq\widehat e_{(m)}^{k\ell}(t)$. Note that for studies $k$ and $\ell$, all the e-values share the same denominator, and the numerators take either the value of $0$ or $m$. Furthermore, the non-decreasing sequence $\widehat{\rm rLIS}_{(1)}^{k\ell}, \ldots, \widehat{\rm rLIS}_{(m)}^{k\ell}$ corresponds to the non-increasing e-value sequence $\widehat e_{(1)}^{k\ell}(t), \ldots, \widehat e_{(m)}^{k\ell}(t)$. By the construction of e-values, the largest $r$ such that $r^{-1}\sum_{j=1}^r \widehat{\rm rLIS}_{(j)}^{k\ell} \leq t$ is the same as the largest $r$ such that $\widehat e_{(r)}^{k\ell}(t)\geq m/(rt)$. Therefore, we obtain the following result, which shows the equivalence between the testing procedure (\ref{eq_rej_procedure}) based on the replicability Local Index of Significance values and the e-BH procedure based on $(\widehat e^{k\ell}_j(t))_{j=1}^m$.
\begin{proposition}
    For any pair $1\le k<\ell \le n$, applying the e-BH procedure \citep{wang2022false} based on the e-values defined in (\ref{eq_evalue}) yields an equivalent set of discoveries to the testing procedure in (\ref{eq_rej_procedure}) with FDR level $t$.
\end{proposition}

For any pair $1\le k<\ell \le n$, define $\pi_0^{k\ell} = \mathbb P(\theta_{kj} = 0, \theta_{\ell j} = 0)$, $\pi_1^{k\ell} = \mathbb P(\theta_{kj} = 0, \theta_{\ell j} = 1)$ and $\pi_2^{k\ell} = \mathbb P(\theta_{kj} = 1, \theta_{\ell j} = 0)$, and denote $\widehat\pi_0^{k\ell}, \widehat\pi_1^{k\ell}$ and $ \widehat\pi_2^{k\ell}$ as the corresponding estimators. To ensure replicability across all studies, we aggregate the pairwise e-values by defining
\begin{align}
    \label{eq_e_val_min}
    \widehat e_j(t) = \min_{k< \ell}\{(\widehat\pi_0^{k\ell}+\widehat\pi_1^{k\ell}+\widehat\pi_2^{k\ell})\cdot\widehat e_j^{k\ell} (t)\} \quad\text{ for }j = 1,\ldots,m.
\end{align}

The factor $(\widehat\pi_0^{k\ell}+\widehat\pi_1^{k\ell}+\widehat\pi_2^{k\ell})$ estimates the composite null proportion for pair $(k,\ell)$. We multiply the pairwise e-value by this factor to have valid e-values. This is the key step for verifying the asymptotic average e-value condition under the null required by the eBH procedure of \citet{wang2022false}.
We then apply the e-BH procedure to the aggregated e-values $(\widehat e_1(t), \ldots, \widehat e_m(t))$. 
To avoid the degenerate case where no hypothesis is rejected, we set a constant lower bound $q _-$ for the choice of the pairwise FDR level $t$. In practice, we suggest setting $q_- = q/\{n(n-1)\}$. In the Supplementary Materials, we show that for any $t\in[q_-, q]$, the eBH procedure controls the FDR at level $q$. 

By the construction in \eqref{eq_evalue}, for any fixed pair $1\le k<\ell\le n$ and fixed $t$, each pairwise e-value $\widehat e_j^{k\ell}(t)$ across $j = 1, \ldots, m$ is either $0$ or the same positive value,
$$U^{k\ell}(t) = \frac{m}{\sum_{j'=1}^{m}I\left(\widehat{\rm rLIS}_{j'}^{k\ell}  \le \widehat{\rm rLIS}_{(\widehat R^{k\ell}(t))}^{k\ell} \right)\widehat{\rm rLIS}_{j'}^{k\ell}}.$$
Consequently, the aggregated e-value in \eqref{eq_e_val_min} is also either $0$ or the positive value
$$U(t) = \min_{k<\ell}\{(\widehat\pi_0^{k\ell}+\widehat\pi_1^{k\ell}+\widehat\pi_2^{k\ell})\cdot U^{k\ell} (t)\}.$$
As $t$ increases, the pairwise rejection sets become larger, so more pairwise e-values in \eqref{eq_evalue} are nonzero and more aggregated e-values in \eqref{eq_e_val_min} are nonzero. At the same time, a larger $t$ increases the denominator in $U^{k\ell}(t)$, leading to smaller values of $U^{k\ell}(t)$ and hence a smaller $U(t)$. Thus, choosing $t$ too large may make the nonzero aggregated e-values too small to pass the eBH threshold at fixed FDR level $q$, resulting in an empty discovery set.
Therefore, to ensure the maximum number of discoveries, we select the largest $t\in[q_-, q]$, such that
\begin{align*}
    \min_{k<\ell}\left\{\frac{(\widehat \pi^{k\ell}_0 + \widehat \pi^{k\ell}_1 + \widehat \pi^{k\ell}_2)\widehat R(t)}{\widehat R^{k\ell}(t)}\right\} \geq \frac{t}{q}.
\end{align*}
This is a technical requirement. The rationale and details are provided in Step 1 of Section \ref{subsec_eBH_proof} in the Supplementary Materials.
We denote the selected pairwise FDR level as $\widehat q_m,$ which is used to form the final rejection set for replicability analysis across $n$ studies. 

Algorithm \ref{algo_n_study} summarizes the complete 
CoHiM procedure for the multiple study setting using pairwise e-value aggregation. To reduce computational cost and avoid redundant estimation, we recommend estimating the non-null densities $f_i$ once per study using a set of disjoint study pairs (e.g., $(1,2)$ and $(3,4)$ for $n=4$). The resulting estimates $(\widehat{f}_i)_{i=1}^n$ are reused across other pairs, while the stationary probabilities and transition matrices are learned separately for each pair via the EM algorithm. This approach replaces the full $2^n$-state HMM with $n(n-1)/2$ pairwise $4$-state HMMs, greatly reducing the computational burden while preserving theoretical guarantees on FDR control.   

\begin{algorithm}[htbp]
\caption{CoHiM for the $n$-study case}\label{algo_n_study}
\begin{algorithmic}[1]
\State \textbf{Input: }$p$-values from $n$ studies $(y_{ij})_{n\times m}$, nominal FDR level $q$ and pairwise FDR lower bound $q_- < q$.
\For{$1\leq k < \ell \leq n$}
    \State Compute $\widehat{\rm rLIS}_j^{k\ell}$ as in Algorithm \ref{algo_2study}.
    \State Order the rLIS values $\widehat{\rm rLIS}_{(1)}^{k\ell}\leq \cdots \leq \widehat{\rm rLIS}_{(m)}^{k\ell}$.
    \State Find the pairwise rejection number $\widehat R^{k\ell}(t) = \max\{r: r^{-1}\sum_{j=1}^r \widehat{\rm rLIS}_{(j)}^{k\ell} \leq t\}$ as a function of $t\in[q_-, q]$. 
    \State Compute 
        $\widehat e_j^{k\ell}(t) = mI(\widehat{\rm rLIS}_j^{k\ell} \le \widehat{\rm rLIS}_{(\widehat R_{k\ell}(t))}^{k\ell} )/\{\sum_{j=1}^{\widehat R_{k\ell}(t)}\widehat{\rm rLIS}_{(j)}^{k\ell}\}.$
\EndFor
\State Obtain $\widehat e_j(t) = \min_{k< \ell}\{(\widehat\pi_0^{k\ell} + \widehat\pi_1^{k\ell} + \widehat\pi_2^{k\ell})\widehat e_j^{k\ell}(t)\}$ for $j=1,\ldots,m$.
\State Order the e-values as $\widehat e_{(1)}(t) \geq \cdots \geq \widehat e_{(m)}(t)$.
\State Find the rejection number $\widehat R(t) = \max\{j: \widehat e_{(j)}(t)\geq m/(jq)\}$.
\State Find $\widehat q_m$ as the largest value in $[q_-, q]$ satisfying $\min\{(\widehat\pi_0^{k\ell} + \widehat\pi_1^{k\ell} + \widehat\pi_2^{k\ell})\widehat R(\widehat q_m)/\widehat R^{k\ell}(\widehat q_m)\}\geq \widehat q_m/q$. If no such value exists, let $\widehat q_m = q$.
\State \textbf{Output: } Reject $H_{0j}$ if $\widehat e_j(\widehat q_m) \geq \widehat e_{(\widehat R(\widehat q_m))}(\widehat q_m)$ for $j = 1,\ldots,m$.
\end{algorithmic}
\end{algorithm}

\section{Theory}\label{sec_theory}
\subsection{Notations}
We first consider the two-study case. Recall that $s_j$ takes values in $0,1,2,3$, corresponding to $(\theta_{1j}, \theta_{2j}) = (0,0), (0,1), (1,0), (1,1)$, respectively. Let $\pi = (\pi_0, \pi_1, \pi_2, \pi_3)$ denote the stationary probability of the underlying Markov chain, and let $A = (a_{k\ell})_{k,\ell=0}^3$ be its transition probability matrix. We denote by $f_1$ the probability density function of $y_{1j}$ conditional on $\theta_{1j} = 1$ and $f_2$ the probability density function of $y_{2j}$ conditional on $\theta_{2j} = 1$. Since the HMM is assumed to be stationary, we have $\pi A = \pi$, i.e., $\pi$ is the left eigenvector of $A$ with eigenvalue $1$. When ${\rm rank}(A-I_4) = 3$, the stationary probability $\pi$ is uniquely determined by $A$ under the constraint that $\sum_{k=0}^3\pi_k = 1$. Let $\Phi$ denote the parameter space of $\phi = (\pi, A, f_1, f_2)$, defined as
\begin{align*}
    \Phi = \bigg\{\phi = (\pi, A, f_1, f_2):&~\pi_k \in (0, 1), \sum_{k=0}^3\pi_k = 1; a_{k\ell} \in (0, 1), \\
    &\sum_{\ell=0}^3 a_{k\ell} = 1, \text{ for } k=0,1,2,3; \pi A = \pi; f_1, f_2\in \mathcal H\bigg\},
\end{align*}
where $\mathcal H$ is the space of non-increasing probability density functions supported on $[0, 1]$ satisfying $\lim_{\delta\to0^+}\sup_{f\in\mathcal H}\int_0^\delta f(y){\rm d}y = 0$. 

To measure the distance between two parameters $\phi^{(1)} =  (\pi^{(1)}, A^{(1)}, f_1^{(1)}, f_2^{(1)})$ and  $\phi^{(2)} =  (\pi^{(2)}, A^{(2)}, f_1^{(2)}, f_2^{(2)})$, we define the following metric:
\begin{equation}
    \label{eq_distance_of_parameters}
    d(\phi^{(1)}, \phi^{(2)}) = \|\pi^{(1)} - \pi^{(2)}\|_2 + \|A^{(1)} - A^{(2)}\|_F + d_H(f_1^{(1)}, f_1^{(2)}) + d_H(f_2^{(1)}, f_2^{(2)}),
\end{equation}
where $\|\cdot\|_2$ denotes the $L_2$ norm for vectors, $\|\cdot\|_F$ denotes the Frobenius norm for matrices, and $d_H(\cdot, \cdot)$ denotes the Hellinger distance between two density functions defined as $d_H(g_1, g_2)^2 = 0.5\cdot\int_0^1 \left\{g_1(y)^{1/2} - g_2(y)^{1/2}\right\}^2{\rm d}y.$
Under the distance metric (\ref{eq_distance_of_parameters}), we establish the compactness of the parameter space $\Phi$.
\begin{proposition}\label{prop_compact}
    The parameter space $\Phi$ is compact with respect to the distance $d(\cdot, \cdot)$ defined in (\ref{eq_distance_of_parameters}).
\end{proposition}
The proof of Proposition \ref{prop_compact} is provided in the Supplementary Materials. Compactness is a crucial requirement for establishing the consistency of the maximum likelihood estimator. The transition matrix $A$ is assumed to have strictly positive entries, which implies that the Markov chain is irreducible (i.e., every state can be reached from any other state in finite steps), aperiodic (i.e., the chain does not exhibit periodic behavior) and ergodic (i.e., the chain converges to its unique stationary probability $\pi$ regardless of the initial state) by \cite{walters2000introduction}. Additionally, the non-null density functions $f_1$ and $f_2$ (corresponding to $\theta_{1j}=1$ and $\theta_{2j}=1,$ respectively) are non-uniform and non-increasing. Consequently, assuming conditional independence given the latent states, the four state-dependent density functions $f^{(s_j)}$ in Equation (\ref{eq_f_sj}) for $s_j=0,1,2,3$ are mutually distinct. By Theorem 1 of \cite{alexandrovich2016nonparametric}, the model parameter $\phi = (\pi, A, f_1, f_2)$ is fully identifiable up to label-switching, provided that $A$ is of full rank.

Our theoretical analysis is related in spirit to the empirical-Bayes FDR analysis for two-state nonparametric HMMs in \citet{abraham2022multiple}. However, the present setting requires additional treatment because the two-study replicability null is a union of three hidden states in a four-state HMM, and the procedure further involves study-specific non-null densities and a multi-study e-value aggregation step.

\subsection{Consistency of the maximum likelihood estimation}\label{subsec_consistency}
We impose the following conditions to establish the consistency of the maximum likelihood estimator $\widehat{\phi}_m$ defined in (\ref{eq_MLE_definition}) and to guarantee asymptotic FDR control.

(C1)
The true parameter  $\phi^*$ lies in the interior of the parameter space $\Phi$.

(C2) There exist constants $\delta_0 > 0$ and $0<\varepsilon_0 \leq 1/4$ such that for any $\phi$ satisfying $d(\phi, \phi^*) < \delta_0$, we have $\pi_k(\phi) \geq \varepsilon_0$ and $a_{k\ell}(\phi) \geq \varepsilon_0$ for all $k, \ell = 0, 1, 2, 3$. Furthermore, $\pi_3(\phi) < 1-q$, where $q$ is the target FDR level. 
Additionally, define 
\begin{align*}
    c = c(\varepsilon_0, q) = \frac{1-2\varepsilon_0 + \{(1-2\varepsilon_0)^2 + 4 (1-3\varepsilon_0)\varepsilon_0^3 q/(2-q)\}^{1/2}}{2\varepsilon_0^3q/(2-q)}.
\end{align*}
Since both $\varepsilon_0$ and $q$ are small, it follows that $c > 1$.
We require
\begin{align*}
    \lim_{y\rightarrow 0} f_1(y) > c, \quad \lim_{y\rightarrow 0} f_2(y) > c.
\end{align*}

(C3) There exists a constant $C_1 >0$ such that $\mathbb E_{\phi^*}\left[\left|\log f^{(k)}\left(Y_{11}, Y_{21}; \phi^*\right)\right|\right]<C_1$ for $k = 0,1,2,3$.

(C4) There exist constants $C_2>0$ and $\delta_1>0$ such that $$\mathbb E_{\phi^*}\left[\sup _{d(\phi^*, \phi)<\delta_1}\left\{\log f^{(k)}\left(Y_{11}, Y_{21}; \phi\right)\right\}^{+}\right]<C_2$$ for $k = 0, 1, 2, 3$, where $x^{+}=\max \{x, 0\}$. 

(C5) There exists $\delta_2 > 0$ such that for each $k = 0, 1, 2, 3$, 
$$\mathbb P_{\phi^*}(\rho_0(Y_{11}, Y_{21}) < \infty\mid s_1 = k) > 0,$$ where
\begin{align*}
    \rho_0(y_1, y_2) = \sup_{d(\phi, \phi^*) < \delta_2} \max_{0\leq k,k'\leq 3}\left\{\frac{f^{(k)}(y_1,y_2;\phi)}{f^{(k')}(y_1,y_2;\phi)}\right\}.
\end{align*}

Condition (C1) ensures that $\phi^*$ lies in the interior of a compact parameter space $\Phi,$ facilitating the consistency of the MLE. Condition (C2) prevents degeneracy in both stationary and transition probabilities near $\phi^*$ and enforces that small $p$-values are more probable under the non-null, aligning with the monotonicity assumption in (\ref{monotone}). The value of $c(\varepsilon_0, q)$ is the positive root of a quadratic equation arising in the proof of Theorem~\ref{thm_data_driven_FDR_control}; it ensures a sufficient lower bound on the signal strengths. A simple sufficient interpretation is that, in a neighborhood of zero, both non-null densities are bounded below by the constant $c(\epsilon_0,q)$, and this lower bound depends only on the minimum stationary and transition probability bound $\epsilon_0$ and the target FDR level $q$.

Condition (C3) is a standard regularity assumption from \cite{leroux1992maximum}, ensuring integrability of the log-likelihood. Condition (C4) guarantees the existence of the generalized Kullback–Leibler divergence between densities indexed by $\phi^*$ and $\phi$, where $\phi$ lies in a small neighborhood of $\phi^*$, also per \cite{leroux1992maximum}. Condition (C5) rules out pathological cases by ensuring that the likelihood ratios between state-dependent distributions remain finite with positive probability.

\begin{theorem}
    \label{thm_consistency}
    Under Conditions (C1)-(C5), the maximum likelihood estimator $\widehat\phi_m$ in (\ref{eq_MLE_definition}) is consistent; that is, $d(\widehat{\phi}_m,\phi^*)\to 0$ in probability as $m\to\infty$.
\end{theorem}

\subsection{Oracle FDR control}
%Before establishing the asymptotic FDR guarantee, we first present a widely used oracle FDR control result and its proof to provide intuition for our procedure. 
In the oracle case, we assume that $\phi^* = (\pi^*, A^*, f_1^*, f_2^*)$ is known. The following theorem shows that FDR can be controlled under the oracle case.
\begin{theorem}
    \label{thm_oracle_FDR_control}
    Under the oracle case where $\phi^*$ is known, denote ${\rm rLIS}_j = \mathbb P_{\phi^*}(s_j \in \{0, 1, 2\} \mid (y_{1j}, y_{2j})_{j=1}^m)$ for $j = 1,\ldots, m$. Order the test statistics ${\rm rLIS}_{(1)} \leq \cdots \leq {\rm rLIS}_{(m)}$ with the corresponding null hypotheses $H_{0(1)}, \ldots, H_{0(m)}$. For a pre-specified FDR level $q$, we have the following procedure
    \begin{align*}
        R =& \max\left\{r: \frac{1}{r}\sum_{j=1}^r {\rm rLIS}_{(j)} \leq q\right\}, \\
        & \text{and reject } H_{0(j)} \quad \text{ for }j = 1,\ldots, R.
    \end{align*}
    Then this procedure can control the FDR at level $q$.
\end{theorem}
\begin{proof}
    Denote $R$ as the number of total rejections and $V$ as the number of false rejections. If we reject the replicability null hypothesis $H_{0j}$ if ${\rm rLIS}_j \leq \lambda$ for some threshold $\lambda$, then $\lambda$ satisfies
    \begin{align*}
        {\rm rLIS}_{(R)} \leq \lambda < {\rm rLIS}_{(R+1)}.
    \end{align*}
    Let $\lambda = {\rm rLIS}_{(R)}$ for simplicity.

    Therefore,
    \begin{align*}
        {\rm FDR} = & \mathbb E\left\{\frac{V}{R\vee 1}\right\}\\
        =&\mathbb E\left\{\mathbb E\left(\frac{V}{R\vee 1}\bigg| (y_{1j}, y_{2j})_{j=1}^m\right)\right\}\\
        =& \mathbb E\left\{\frac{1}{R\vee 1}\mathbb E\left(V\mid (y_{1j}, y_{2j})_{j=1}^m\right)\right\}.
    \end{align*}
    The last equality holds because $R$ is a function of $(y_{1j}, y_{2j})_{j=1}^m$.
    Since $V = \sum_{j=1}^m I({\rm rLIS}_j \leq \lambda, s_j \in \{0, 1, 2\}) = \sum_{j=1}^R I(s_{(j)} \in \{0, 1, 2\})$ and $(y_{1j}, y_{2j})_{j=1}^m$ does not contain all information for the hidden states $s_j$ for $j = 1, \ldots, m$, 
    we have
    \begin{align*}
        \mathbb E\left(V\mid (y_{1j}, y_{2j})_{j=1}^m\right) =& \sum_{j=1}^R\mathbb P(s_{(j)}\in \{0, 1, 2\}\mid (y_{1j}, y_{2j})_{j=1}^m)
        = \sum_{j=1}^R {\rm rLIS}_{(j)}.
    \end{align*}
    Consequently,
    \begin{align*}
        {\rm FDR} = \mathbb E\left\{\frac{1}{R\vee 1}\sum_{j=1}^R {\rm rLIS}_{(j)}\right\} \leq q.
    \end{align*}
\end{proof}

\subsection{Asymptotic FDR control}\label{subsec_asymp_fdr}
We now establish the asymptotic FDR control of Algorithm \ref{algo_2study}. %, which applies to the two-study setting.
\begin{theorem}
    \label{thm_data_driven_FDR_control}
   Suppose the paired $p$-values $(y_{1j}, y_{2j})_{j=1}^m$ follow a four-state HMM with true parameter $ \phi^* = (\pi^*, A^*, f_1^*, f_2^*)$. Conditional on the hidden states, the paired $p$-values $(y_{1j}, y_{2j})_{j=1}^m$ follow the mixture model specified in (\ref{eq_mixture}), and the null $p$-values are uniformly distributed. 
   If the monotonicity condition (\ref{monotone}) and Conditions (C1)-(C5) hold, then Algorithm \ref{algo_2study} asymptotically controls the FDR at level $q$.
\end{theorem}
Theorem \ref{thm_data_driven_FDR_control} builds on Theorem \ref{thm_consistency}, which guarantees consistent estimation of the stationary probabilities, transition probabilities, and non-null densities. Together, these results provide a self-contained framework for testing composite null hypotheses in high-dimensional settings while accounting for dependence. 

We now establish asymptotic FDR control under the multiple-study setting.
\begin{theorem}
    \label{thm_eBH_procedure}
    Let $(y_{ij})_{n\times m}$ denote the $p$-value matrix for $m$ features across $n$ studies. Suppose $(y_{ij})_{n\times m}$ follow a $2^n$-state HMM, and that the hidden states satisfy the following pairwise Markov property:
     \begin{align}
        \label{eq_markov_property}
         \mathbb P(\theta_{k,j+1}, \theta_{\ell,j+1}\mid \theta_{1j}, \ldots, \theta_{nj}) = \mathbb P(\theta_{k,j+1}, \theta_{\ell,j+1}\mid \theta_{kj}, \theta_{\ell j})
     \end{align} 
     for all $1 \le k< \ell \le n$ and $j = 1,\ldots,m-1$. Suppose the assumptions of Theorem \ref{thm_data_driven_FDR_control} hold for each study pair $(k,\ell)$, with Condition (C2) satisfied for some $q_- < q$. Then, Algorithm \ref{algo_n_study} controls the FDR asymptotically at level $q$.
\end{theorem}
Note that in Condition (C2), the threshold $c(\varepsilon_0, q)$ is a decreasing function of $q$. Therefore, if (C2) holds for any $q_- < q$, it also holds for $q$. For $n>2$, Theorem \ref{thm_eBH_procedure} establishes asymptotic FDR control for the e-BH procedure based on the e-values defined in Equations (\ref{eq_evalue}) and (\ref{eq_e_val_min}). Our procedure assumes that, for each pair $k< \ell$, the joint hidden states $(\theta_{kj}, \theta_{\ell j})_{j=1}^m$ follow a $4$-state Markov chain. Importantly, this assumption is strictly weaker than the following assumptions: 
\begin{enumerate}
\item[(i)] for each $i=1,\ldots,n$, the marginal sequence $(y_{ij})_{j=1}^m$ follows a $2$-state HMM. 
\item[(ii)] the $n$ studies are mutually independent. 
\end{enumerate}
Under Assumptions (i) and (ii), the full $p$-value array $(y_{ij})_{n \times m}$ necessarily follows a \(2^n\)-state HMM and satisfies the pairwise Markov property in (\ref{eq_markov_property}). However, the converse is not true in general. 

We illustrate this distinction by constructing a model that satisfies (\ref{eq_markov_property}) but violates mutual independence. Let $(\theta_{1j}, \theta_{2j}, \theta_{3j})_{j = 1}^m$ be a Markov chain on $\{0,1\}^3$, evolving as follows. At each time step $j$, draw a random vector $(A_{1j}, A_{2j}, A_{3j}) \in \{0,1\}^3$, and update:
\begin{align*}
    \theta_{i,j+1} = \theta_{i,j} \oplus A_{ij}, \quad i = 1, 2, 3,
\end{align*}
where $\oplus$ denotes addition modulo 2: $a \oplus b = I(a \ne b)$ for $a, b \in \{0, 1\}$. This construction allows each of the $2^3=8$ states to be visited and ensures that each pair $(\theta_{kj}, \theta_{\ell j})$ forms a valid 4-state Markov chain.

To verify this, consider the pair $(\theta_{1j}, \theta_{2j})$. At time $j+1$, we have $(\theta_{1,j+1}, \theta_{2,j+1}) = (\theta_{1j} \oplus A_{1j}, \theta_{2j} \oplus A_{2j})$. Hence,
\begin{align*}
    &\mathbb{P}(\theta_{1,j+1} = x', \theta_{2,j+1} = y' \mid \theta_{1j} = x, \theta_{2j} = y, \theta_{3j} = z, (\theta_{1j'}, \theta_{2j'}, \theta_{3j'})_{j'=1}^{j-1}) 
    \\
    =&\mathbb{P}(A_{1j} = x \oplus x', A_{2j} = y \oplus y'),
\end{align*}
which depends only on the current states $(x, y)$ and not on $\theta_{3j}$ or the full past. The same reasoning applies to other pairs $(\theta_{2j}, \theta_{3j})$ and $(\theta_{1j}, \theta_{3j})$. This satisfies the Markovian assumption in \eqref{eq_markov_property}. However, if the components of $A_{1j}, A_{2j}, A_{3j}$ are not mutually independent, for example, if $A_{3j} = A_{1j} A_{2j}$, then the marginal sequences $(\theta_{ij})$ are no longer mutually independent. This violates assumption (ii).%, even though the pairwise Markov property (\ref{eq_markov_property}) still holds. 

\section{Simulations}\label{sec_simu}
\subsection{Two studies}\label{subsec_two_study}
We conduct simulation studies to evaluate the finite-sample performance of the proposed method, CoHiM, in terms of FDR control and statistical power. We fix the number of hypotheses at $m=10,000$ and generate dependence among features through a Markov chain $(s_j)_{j=1}^m$ with stationary probability $\pi=(\pi_0, \pi_1, \pi_2, \pi_3)$. 
We consider two settings for the stationary probabilities:
\begin{itemize}
\item Setting 1: $\pi = (0.75, 0.1, 0.1, 0.05).$
\item Setting 2: $\pi = (0.65, 0.15, 0.15, 0.05).$
\end{itemize}

Because the transition matrix $A$ is linked to the stationary distribution via the stationarity condition $\pi A=\pi$, modifying $\pi$ alters both the signal composition and the local dependence encoded in $A$. The resulting transition matrices are given by
\begin{align*}
    A = (a_{k\ell})_{k,\ell=0}^3 = \begin{pmatrix}
        0.956 & 0.015 & 0.015 & 0.015\\
        0.111 & 0.667 & 0.111 & 0.111\\
        0.111 & 0.111 & 0.667 & 0.111\\
        0.222 & 0.222 & 0.222 & 0.333
    \end{pmatrix}
    \text{ and }
    \begin{pmatrix}
        0.949 & 0.017 & 0.017 & 0.017\\
        0.074 & 0.778 & 0.074 & 0.074\\
        0.074 & 0.074 & 0.778 & 0.074\\
        0.222 & 0.222 & 0.222 & 0.333\\
    \end{pmatrix}.
\end{align*}
The initial state $s_1$ is drawn from the stationary probability, i.e., $\mathbb P(s_1 = k) = \pi_k$ for $k = 0, 1, 2, 3$. The subsequent states are sampled via the Markov transition rule: $\mathbb P(s_j = \ell\mid s_{j-1}=k) = a_{k\ell}$ for $j = 2,\ldots, m$.
For $j = 1, \ldots, m$, each hidden state $s_j$ indicates the pair of binary latent states $(\theta_{1j}, \theta_{2j})$. Conditional on $\theta_{ij}$, we simulate z-scores as $$X_{ij}\mid\theta_{ij}\sim (1-\theta_{ij})\mathcal N(0, 1) + \theta_{ij}\mathcal N(\mu_i, 1),$$ where $\mu_i$ controls the signal strength in study $i$ for $i = 1, 2$. The one-sided $p$-values are then computed as $y_{ij} = \mathbb P(Z \geq X_{ij})$ for $Z \sim \mathcal N(0, 1)$. 
\begin{figure}[htbp]
    \centering{\includegraphics[width=\textwidth]{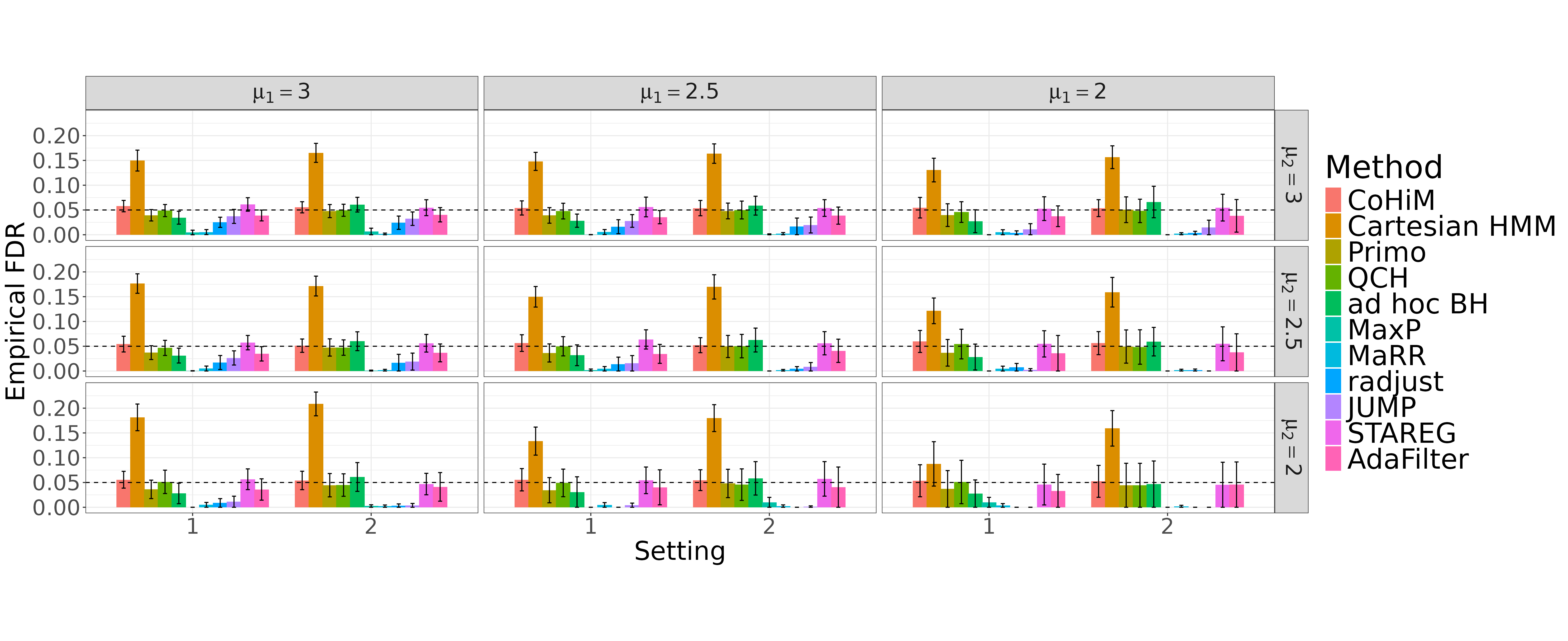}
    \includegraphics[width=\textwidth]{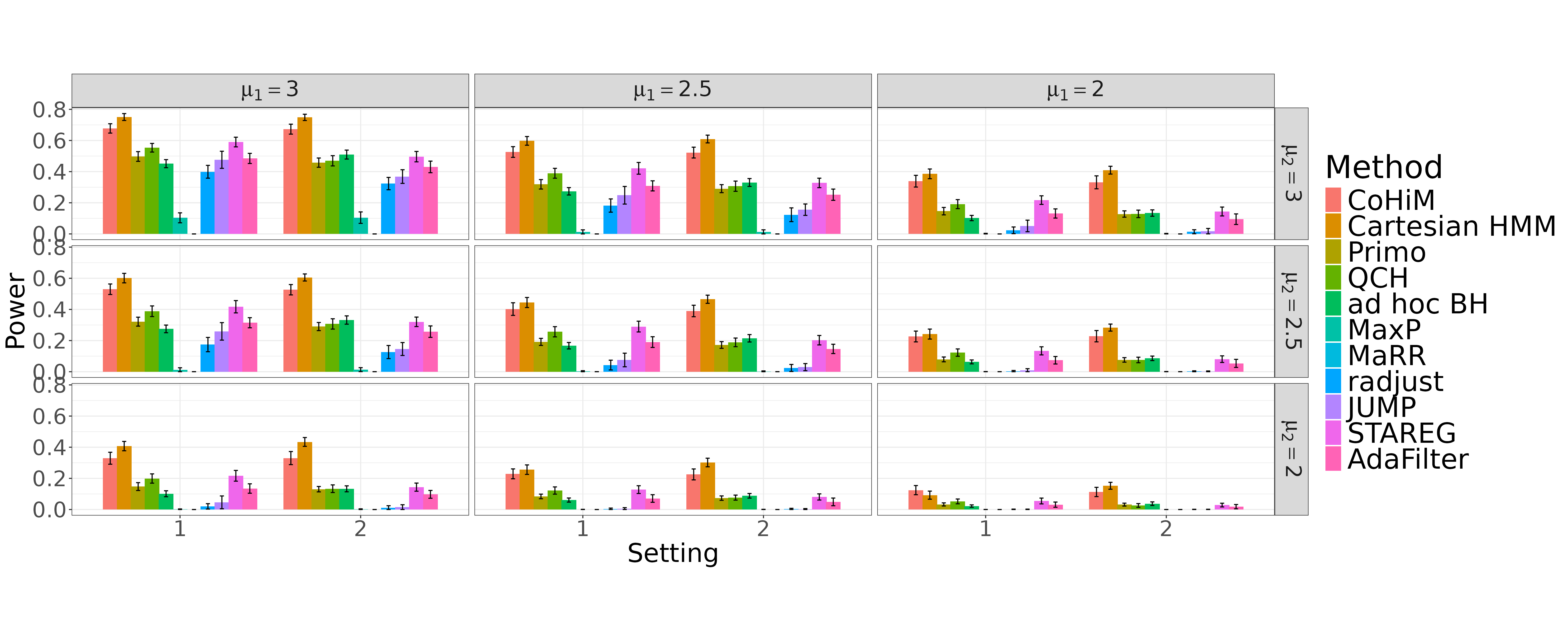}}
    \caption{Empirical FDR and power of different methods under different settings. The dashed horizontal line in the left panel indicates the nominal FDR level of 0.05. Error bars represent the mean $\pm$ one standard deviation over 100 replicates.}
    \label{fig_simu2_fdr_power}
\end{figure}

We compare CoHiM with the following methods:
\begin{itemize}
    \item Cartesian HMM \citep{wang2019replicability}.
    \item PLACO \citep{ray2020powerful}.
    \item Primo \citep{gleason2020primo}.
    \item QCH \citep{mary2022querying, de2025large}.
    \item {\it ad hoc} BH \citep{benjamini1995controlling},
    \item MaxP \citep{benjamini2009selective}, 
    \item MaRR \citep{philtron2018maximum}, 
    \item radjust \citep{bogomolov2018assessing},
    \item JUMP \citep{lyu2023jump}, 
    \item STAREG \citep{li2024stareg},
    \item AdaFilter \citep{wang2022detecting}. 
\end{itemize}
Detailed descriptions of these methods are provided in the Supplementary Materials. For each setting, we conduct $100$ replicates and compute the mean and standard deviation of the empirical FDR and statistical power, using a nominal FDR level of $q = 0.05$. Unless otherwise stated, all competing methods in the simulation studies are run using their default software settings, with the nominal FDR level set to the target level \(q\).

To explore performance under varying levels of signal strength, we vary  $\mu_1$ and $\mu_2$. The resulting empirical FDR and power are summarized in Figure \ref{fig_simu2_fdr_power}. Most methods achieve valid FDR control, although Cartesian HMM, {\it ad hoc} BH and STAREG occasionally exceed the target level or show large variability. In addition, PLACO, MaxP, radjust, and JUMP are overly conservative and consequently exhibit low power. CoHiM consistently achieves higher power than all competitors except Cartesian HMM, which fails to control FDR, especially under weak-signal scenarios. Power 
increases with stronger signals across all methods. 

We also assess FDR control and power across a range of nominal FDR levels from $0.001$ to $0.2$. In this setting, we fix $m = 10,000$, $\pi = (0.7, 0.05, 0.05, 0.2)$, $\mu_1 = \mu_2 = 2$, and use the following transition matrix 
\begin{align*}
    A = \begin{pmatrix}
        0.952 & 0.016 & 0.016 & 0.016\\
        0.222 & 0.333 & 0.222 & 0.222\\
        0.222 & 0.222 & 0.333 & 0.222\\
        0.222 & 0.222 & 0.222 & 0.333
    \end{pmatrix}.
\end{align*}
Figure \ref{fig_FDR_nominal} shows that CoHiM maintains valid FDR control across all thresholds, comparable to STAREG and AdaFilter. Meanwhile, MaxP, radjust, and JUMP remain conservative, and {\it ad hoc} BH exhibits inflated FDR at some thresholds. Notably, CoHiM achieves the highest power across all nominal FDR levels among the eleven methods evaluated. 
% {\color{red}This robustness analysis uses the original seven competing methods ({\it ad hoc} BH, MaxP, MaRR, radjust, JUMP, STAREG, AdaFilter); PLACO, Primo, and QCH are included in the FDR/power comparisons in Figures~\ref{fig_simu2_fdr_power} and \ref{fig_simu3_fdr_power} but are not shown here as the figure was designed to assess calibration behavior across nominal levels rather than to benchmark all competitors.}

\begin{figure}[htbp]
    \centering
    \includegraphics[width = 0.9\linewidth]{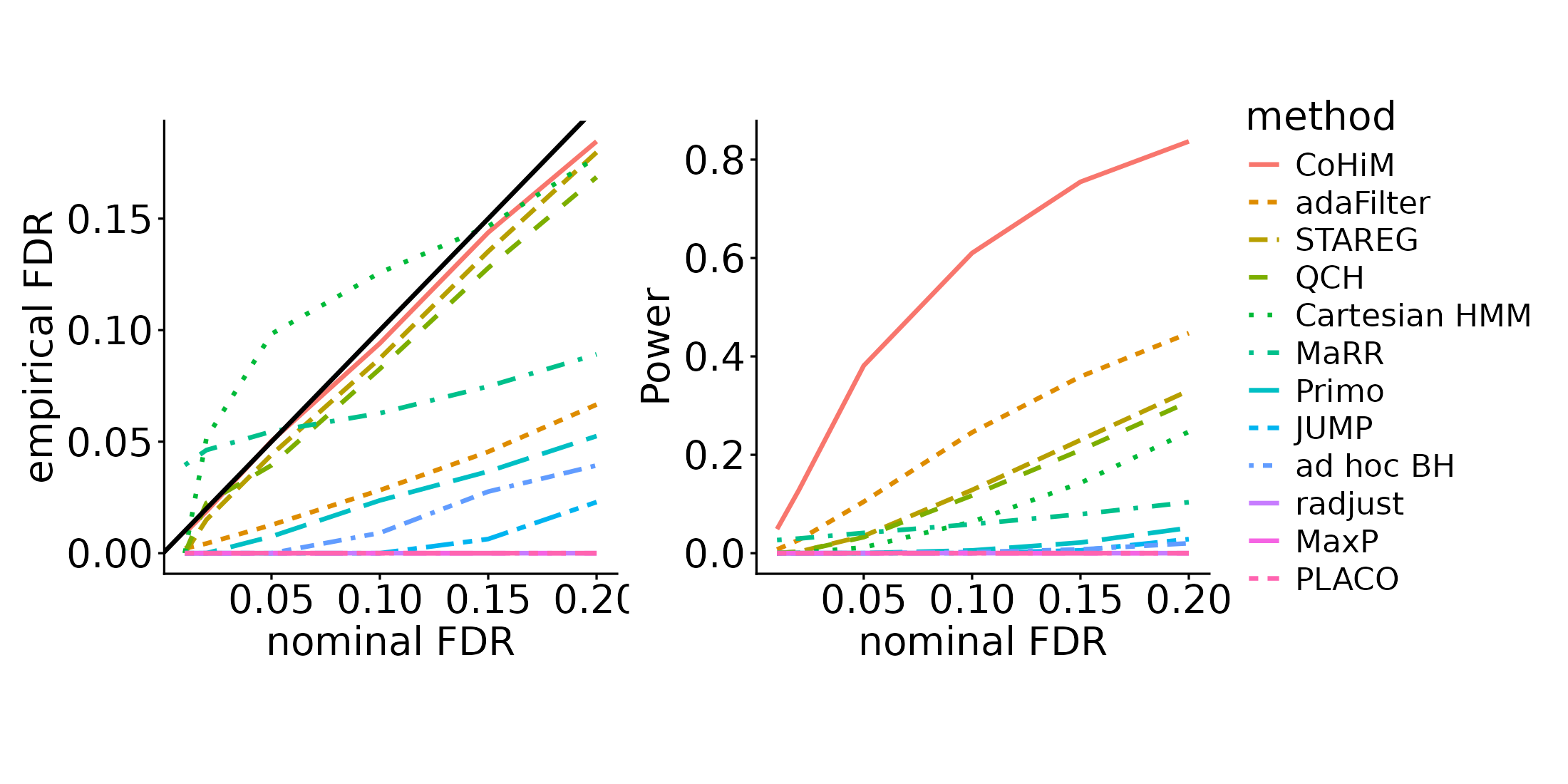}
    \caption{Empirical FDR (left) and power (right) of different methods across a range of nominal FDR levels. 
    The diagonal line (slope 1) in the left panel represents perfect FDR calibration.}
    \label{fig_FDR_nominal}
\end{figure}

\subsection{Three studies}
We next evaluate the performance of CoHiM in the setting of three studies. Let $s_j \in \{0, 1, \ldots, 7\}$ denote the joint signal configuration for feature $j$, corresponding to the eight possible combinations of binary states $(\theta_{1j}, \theta_{2j}, \theta_{3j}) \in \{0, 1\}^3.$ We define the replicability null hypothesis for feature $j$ as 
$$H_{0j}: s_j \in \{0, 1, \ldots, 6\} \quad \mbox{ for } \quad j = 1,\dots,m,$$ i.e., the signal is not consistently present in all three studies. Only $s_j = 7$ corresponds to a feature with replicable signals across all studies. We set the total number of hypotheses to $m = 10,000$. The hidden states $(s_j)_{j=1}^m$ are generated from a Markov chain with stationary probability $\pi$ and transition probability matrix $A$, satisfying the stationary condition $\pi A = \pi.$ We fix $\pi_7 = 0.1$ and set $\pi_1 = \cdots = \pi_6,$ considering values of  $\pi_1 = 0.01$ and $0.015$ to control the degree of sparsity. For instance, with $\pi_1 = 0.01$, the stationary probabilities and transition probabilities are 
\begin{align*}
    \pi =& \begin{pmatrix}
        0.840&0.010&0.010&0.010& 0.010&0.010&0.010&0.100
    \end{pmatrix},\quad \mbox{and}\\ 
    A =& \begin{pmatrix}
        0.972 & 0.004& 0.004& 0.004& 0.004& 0.004& 0.004& 0.004\\
        0.095&0.333&0.095&0.095&0.095&0.095&0.095&0.095\\
        0.095&0.095&0.333&0.095&0.095&0.095&0.095&0.095\\
        0.095&0.095&0.095&0.333&0.095&0.095&0.095&0.095\\
        0.095&0.095&0.095&0.095&0.333&0.095&0.095&0.095\\
        0.095&0.095&0.095&0.095&0.095&0.333&0.095&0.095\\
        0.095&0.095&0.095&0.095&0.095&0.095&0.333&0.095\\
        0.012&0.012&0.012&0.012&0.012&0.012&0.012&0.916
    \end{pmatrix}.
\end{align*}
\begin{figure}[htbp]
    \centering
    \includegraphics[width=\linewidth]{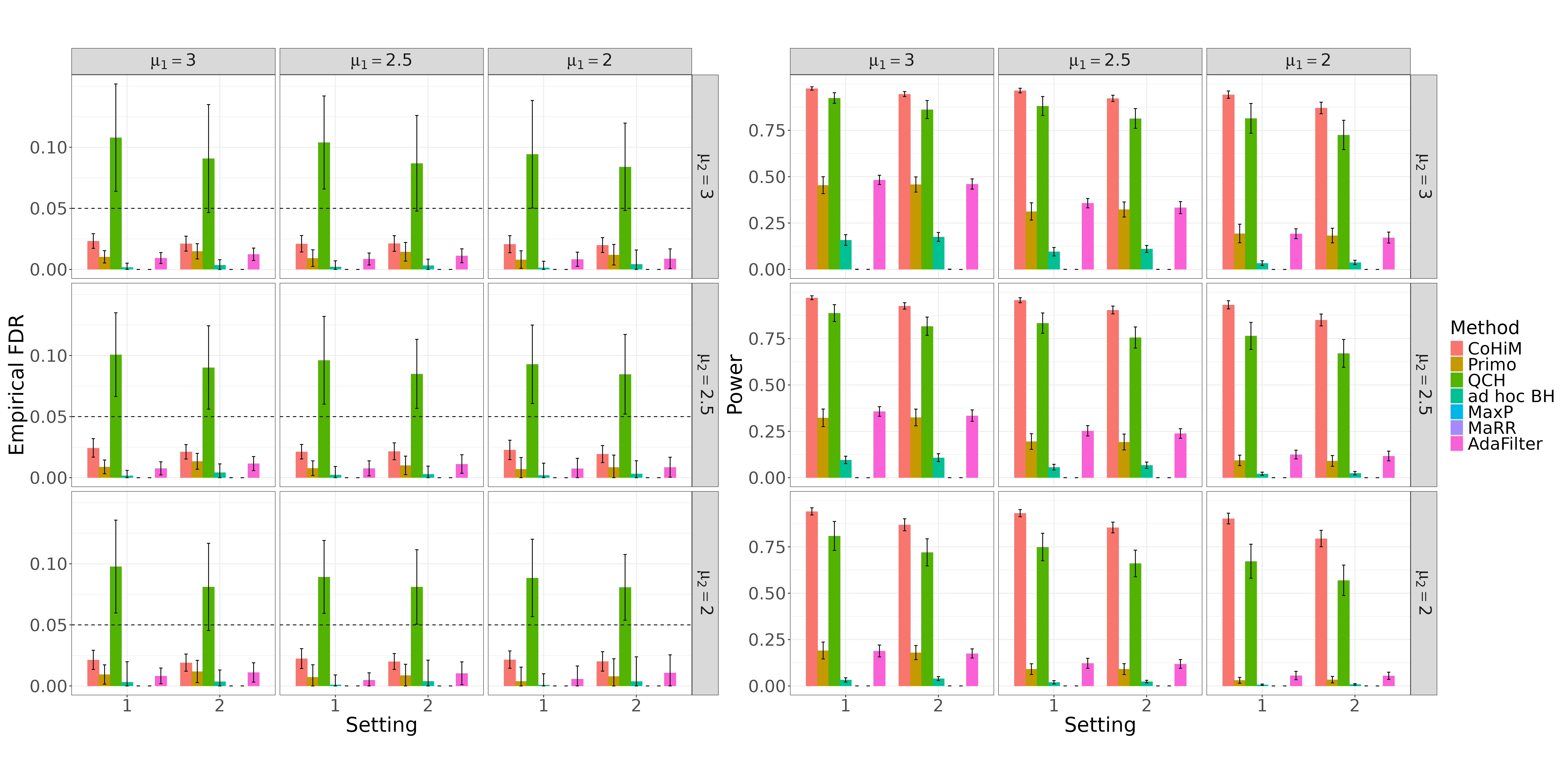}
    \caption{Empirical FDR and power of different methods for three studies under different settings. The dashed horizontal line in the left panel indicates the nominal level of 0.05. Error bars represent the mean $\pm$ one standard deviation over 100 replicates.
    }
    \label{fig_simu3_fdr_power}
\end{figure}

For each study $i$ and hypothesis $j$, we simulate z-scores from a two-component mixture distribution: 
$$X_{ij} \mid \theta_{ij} \sim (1-\theta_{ij})\mathcal N(0, 1) + \theta_{ij} \mathcal N(\mu_i, 1),$$ where $\mu_i$ denotes the signal strength in study $i$. The one-sided $p$-values are computed as $y_{ij} = \mathbb P(Z \geq X_{ij})$ for $Z \sim \mathcal N(0, 1)$. To introduce heterogeneity across studies, we fix $\mu_3 = 2.5$ and vary $\mu_1$ and $\mu_2$ across simulation settings.

We apply CoHiM using Algorithm \ref{algo_n_study}, which aggregates pairwise e-values across the three study pairs $\{ 1,2\}, \{1,3 \}$ and $\{2,3 \},$ and determines the final rejection set via the e-BH procedure. We compare CoHiM to six existing methods:
\begin{itemize}
    \item Primo \citep{gleason2020primo}.
    \item QCH \citep{mary2022querying, de2025large}.
\item {\it ad hoc} BH \citep{benjamini1995controlling}, 
\item MaxP \citep{benjamini2009selective}, 
\item MaRR \citep{philtron2018maximum}, 
% \item JUMP \citep{lyu2023jump}, 
\item AdaFilter \citep{wang2022detecting}. 
\end{itemize}
The results are summarized in Figure \ref{fig_simu3_fdr_power}. CoHiM effectively controls the FDR across all tested settings. In contrast, QCH fails to control the FDR and the other competing methods are generally conservative, with lower empirical FDR and limited power. MaxP and MaRR, in particular, yield very few discoveries across all signal strengths. AdaFilter achieves moderate power only under strong signal settings. Notably, the Markov chain $(s_j)_{j=1}^m$ used here does not satisfy the pairwise Markov property in Equation (\ref{eq_markov_property}), meaning the theoretical guarantees of Theorem \ref{thm_eBH_procedure} do not strictly apply. Nonetheless, CoHiM demonstrates strong empirical performance and robustness under this form of model misspecification, highlighting its practical utility in complex multiple-study scenarios. 

\subsection{Five studies}
\label{subsec:five-study}

We next evaluate CoHiM in the five-study setting. Since a single joint HMM over five studies would require $2^5 = 32$ latent states, rendering joint estimation computationally prohibitive, we adopt a block structure: studies~1 and~2 are governed by one four-state Markov chain, and studies~3, 4, and~5 by an independent eight-state Markov chain. Let $s_j^{(12)} \in \{0,1,2,3\}$ denote the joint signal configuration for the first pair, corresponding to $(\theta_{1j}, \theta_{2j}) \in \{0,1\}^2$, and let $s_j^{(345)} \in \{0,\ldots,7\}$ denote the configuration for the remaining triple, corresponding to $(\theta_{3j}, \theta_{4j}, \theta_{5j}) \in \{0,1\}^3$. The composite null hypothesis for feature $j$ is $H_{0j}\colon \theta_{1j}\theta_{2j}\theta_{3j}\theta_{4j}\theta_{5j} = 0$, so that $j$ is replicable only when all five study-specific indicators are simultaneously nonzero.

We set $m = 10{,}000$. The stationary distribution for Block~A (studies~1--2) is $\pi^{(12)} = (0.50,\,0.05,\,0.05,\,0.40)^\top$, where $\pi_3^{(12)} = 0.40$ is the probability of being non-null in both studies. The stationary distribution for Block~B (studies~3--5) is
$$\pi^{(345)} = (0.55,\,\underbrace{0.05/6,\ldots,0.05/6}_{6 \text{ times}},\,0.40)^\top,$$
where $\pi_7^{(345)} = 0.40$ is the probability of being non-null in all three studies and the six partial non-null states each carry probability $0.05/6$. The expected composite signal rate is $\pi_3^{(12)} \times \pi_7^{(345)} = 0.16$, yielding approximately $1{,}600$ true replicable features. For each block, the transition matrix is constructed to satisfy the stationarity condition and to capture local dependence among features, following the same approach as in the preceding simulations. For each study $i \in \{1,\ldots,5\}$ and hypothesis $j$, z-scores and $p$-values are generated as in the two-study case: $X_{ij} \mid \theta_{ij} \sim (1-\theta_{ij})\mathcal{N}(0,1) + \theta_{ij}\mathcal{N}(\mu_i,1)$ with $y_{ij} = \mathbb{P}(Z \geq X_{ij})$.

We consider three signal configurations with varying effect sizes; see Table~\ref{tab:lowsignal-cohim-win-settings} for the specific values. In all settings, the signal strengths satisfy $\mu_1 = \mu_2 \geq \mu_3 = \mu_4 = \mu_5$, reflecting stronger signals in studies~1--2 than in studies~3--5. We set the nominal FDR level to $q = 0.05$ and conduct $100$ independent replications per setting.

\begin{table}[htbp]
\centering
\caption{Five-study simulation settings ($m=10{,}000$, $100$ replications, $q=0.05$). In all three settings, $\pi^{(12)}=(0.50,\,0.05,\,0.05,\,0.40)^\top$ and $\pi^{(345)}=(0.55,\,0.05/6,\ldots,0.05/6,\,0.40)^\top$.}
\label{tab:lowsignal-cohim-win-settings}
\begin{tabular}{c c c c}
\hline
Setting & $\mu_1=\mu_2$ & $\mu_3=\mu_4=\mu_5$ & Non-null proportion \\
\hline
Setting~1 & 4.5 & 4.0 & 16\% \\
Setting~2 & 4.0 & 3.5 & 16\% \\
Setting~3 & 4.0 & 3.0 & 16\% \\
\hline
\end{tabular}
\end{table}

We compare CoHiM with six existing methods:
\begin{itemize}
    \item Primo \citep{gleason2020primo},
    \item QCH \citep{mary2022querying,de2025large},
    \item {\it ad hoc} BH \citep{benjamini1995controlling},
    \item MaRR \citep{philtron2018maximum},
    \item MaxP \citep{benjamini2009selective},
    \item AdaFilter \citep{wang2022detecting}.
\end{itemize}
Detailed descriptions of these methods are provided in the Supplementary Materials.

The results are summarized in Figure~\ref{fig_simu5_fdr_power}. Although QCH attains competitive power in some settings, it fails to control the FDR at the nominal level of $0.05$ across all three settings.
% Because a method that violates the error-rate constraint provides no statistical guarantee, its power figures are not a meaningful basis for comparison; we therefore focus on FDR-controlling competitors. 
Among the remaining approaches, CoHiM effectively controls FDR and achieves the highest power in all three settings---$0.962$, $0.877$, and $0.754$ in Settings~1--3, respectively. When signal strengths decrease, the power decreases accordingly, and the competing methods exhibit larger power losses. These results demonstrate that CoHiM maintains valid FDR control and achieves competitive power as the number of studies increases to five. The Markov chain used here does not satisfy the pairwise Markov property required by Theorem~\ref{thm_eBH_procedure}. Nonetheless, CoHiM exhibits robust empirical performance under this form of model misspecification.%, consistent with the three-study findings above.

\begin{figure}[htbp]
    \centering
    \includegraphics[width=\linewidth]{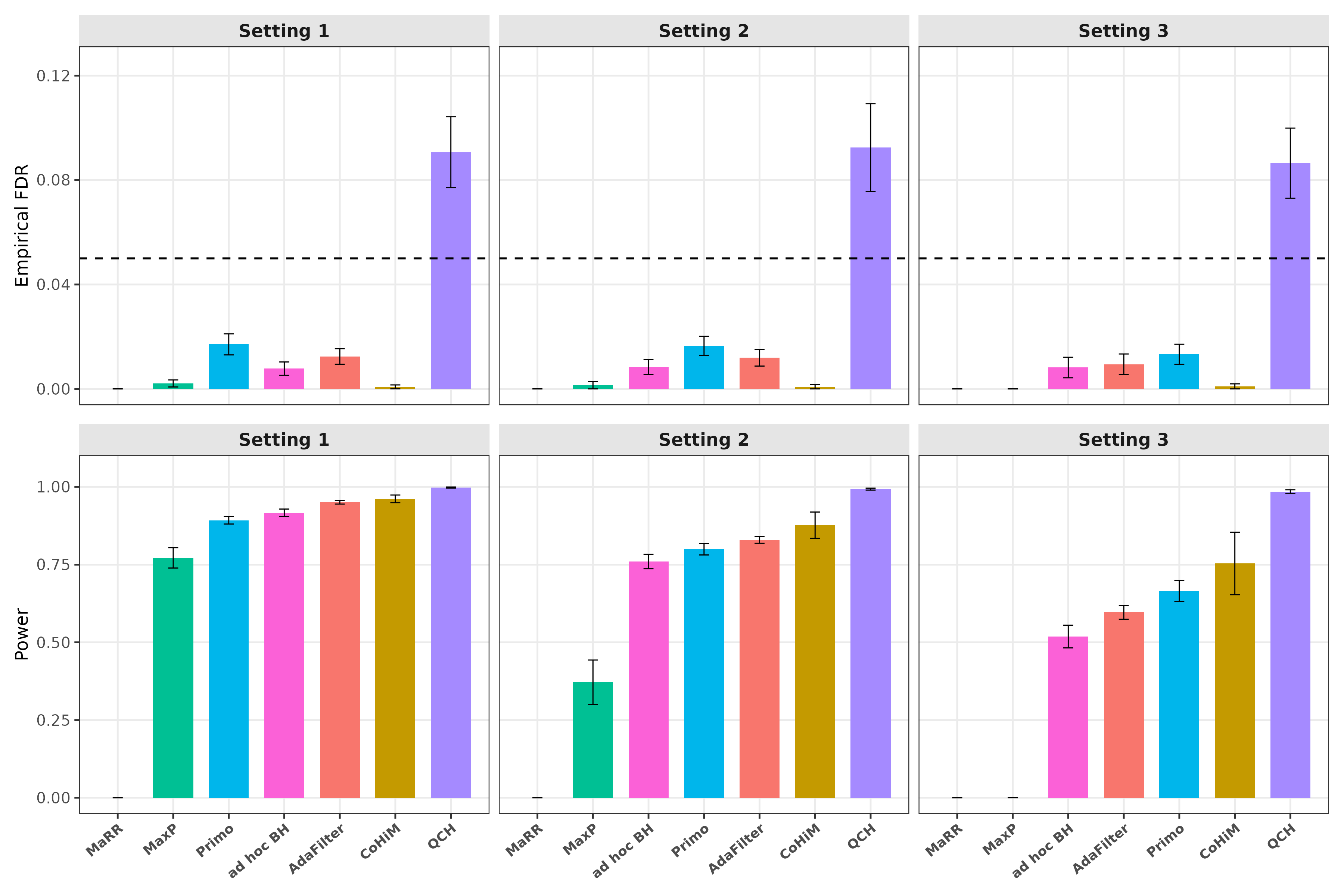}
    \caption{Empirical FDR and power of different methods for five studies across three settings. The dashed horizontal line in the FDR panel indicates the nominal level of $0.05$.  Error bars represent the mean $\pm$ one standard deviation over 100 replicates.}
    \label{fig_simu5_fdr_power}
\end{figure}

\subsection{Computational time comparison}
\begin{figure}[htbp]
    \centering
    \includegraphics[width=0.85\linewidth]{figures/runtime_comparison_with_joint.pdf}
    \caption{Computational time comparison between CoHiM and Overall HMM as the number of studies $n$ increases.}
    \label{fig:running time}
\end{figure}
We conducted an additional experiment to directly compare the computational time of the proposed pairwise-aggregation implementation (CoHiM) in Algorithm~\ref{algo_n_study} against a full joint HMM implementation (overall HMM) whose latent-state size is of order $2^n$.  We vary the number of studies $n\in\{2,3,5,8,10\}$, set the number of hypotheses as $m=10,000$, and target FDR level $q=0.05$.% To ensure feasibility of the full joint benchmark, we used the cap $K_{\max}=1025$ (i.e., up to $2^{10}=1024$ latent states).

The results show that the methods are similar for small $n$, but the overall HMM method becomes increasingly slower as $n$ grows. At $n=10$, overall HMM takes more than ten thousand seconds versus $100$ seconds for CoHiM. This effectively shows that our method is computationally scalable for large-scale biomedical data analysis.

\section{Data analysis}\label{sec_data}
\subsection{Type 2 diabetes: sex-stratified replicability}\label{subsec_t2d}
We illustrate the utility of CoHiM by analyzing two sex-stratified GWAS datasets from \cite{diabetes2012large}, which investigate associations between SNPs and type 2 diabetes. In such datasets, significant SNPs often exhibit clustering due to LD, making HMMs particularly suitable for modeling the local dependence structure. Type 2 diabetes is a metabolic disorder characterized by elevated blood glucose levels, affecting approximately 329 million individuals globally in 2015 \citep{lipton2016gbd}. Identifying replicable genetic associations is critical for advancing our understanding of the disease's biological mechanisms and guiding therapeutic development. 

The male dataset comprises $20,219$ cases and $54,604$ controls, while the female dataset includes $14,621$ cases and $60,377$ controls. Summary statistics were obtained from the DIAbetes Genetics Replication and Meta-analysis (DIAGRAM) Consortium (\url{https://www.diagram-consortium.org/downloads.html}). The male group includes summary statistics for
$123,535$ SNPs, while the female group includes summary statistics for $118,399$ SNPs. After matching the SNPs across datasets, we analyze $m=118,364$ SNPs common to both sexes, with \(y_{1j}\) and $y_{2j}$ for $j = 1,\ldots,m$ denoting the $p$-values in males and females, respectively.

We applied CoHiM via Algorithm \ref{algo_2study} to estimate the HMM parameters and identify replicable associations. The estimated transition matrix is 
\begin{align*}
    \widehat{A} =\begin{pmatrix}
        0.9840 & 0.0066 & 0.0040 & 0.0055\\
        0.0657 & 0.9271 & 0.0004 & 0.0069\\
        0.0546 & 0.0010 & 0.9379 & 0.0066\\
        0.0501 & 0.0045 & 0.0050 & 0.9403
    \end{pmatrix},
\end{align*} 
with corresponding stationary probability
\begin{align*}
    \widehat{\pi} = (0.779, 0.077, 0.057, 0.087).
\end{align*}
Figure \ref{fig_alternative_density} displays the estimated probability density functions of the non-null $p$-values, $\widehat f_1$ and $\widehat f_2$ for males and females, respectively, highlighting substantial heterogeneity between the two studies. 
\begin{figure}[htbp]
    \centering
    \includegraphics[width = 0.8\linewidth]{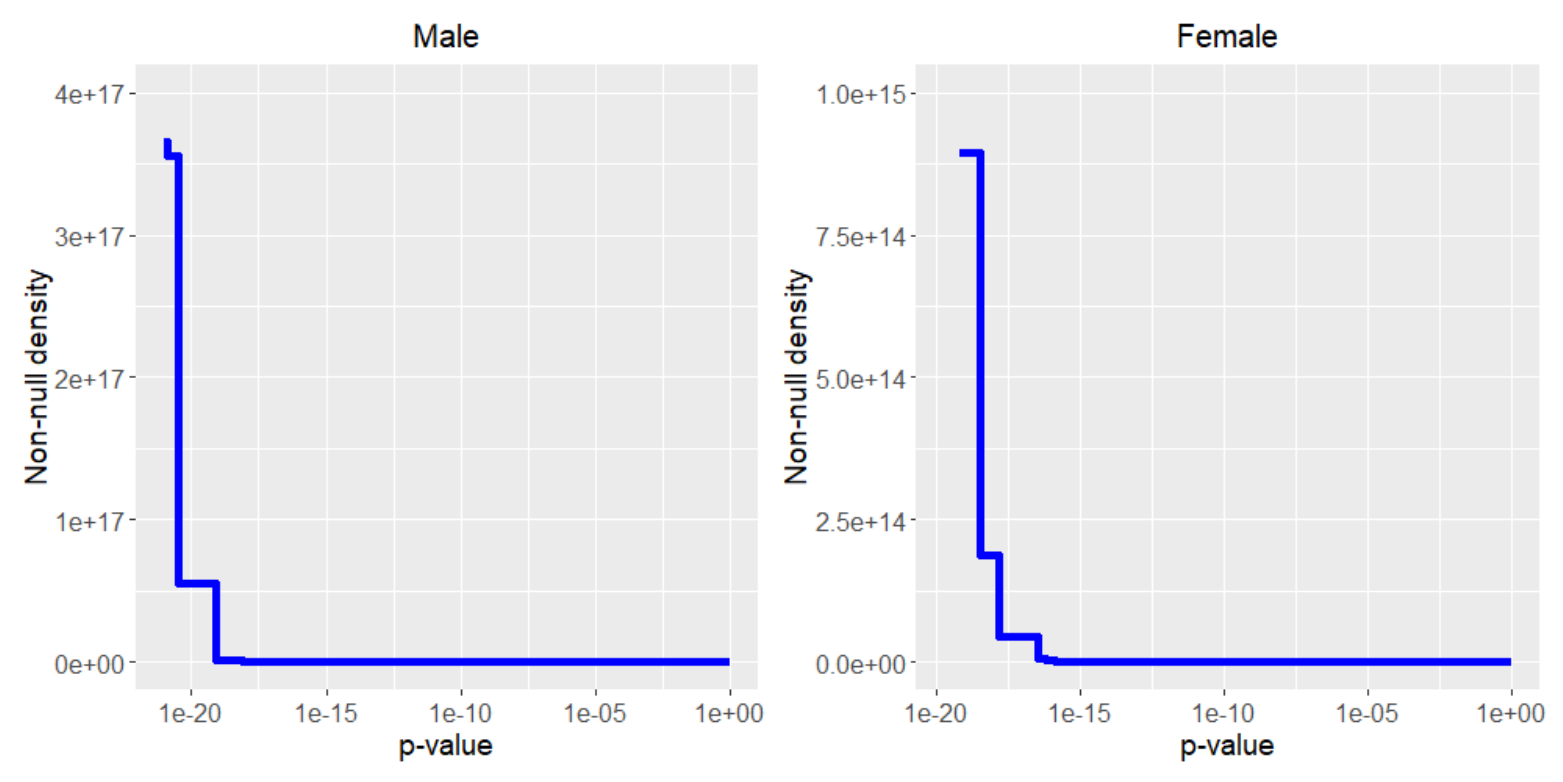}
    \caption{Estimated non-null $p$-value densities for the male and female type 2 diabetes studies.}
    \label{fig_alternative_density}
\end{figure}

We benchmark CoHiM against existing replicability analysis methods at a nominal FDR level of $q = 10^{-5}$. GWAS data contain many correlated markers due to linkage disequilibrium among nearby SNPs \citep{visscher2012five,li2012evaluating}, and stringent thresholding is standard in large-scale GWAS multiple testing to prioritize high-confidence signals \citep{dudbridge2006note,pe2008estimation}. Although $q=10^{-5}$ is more conservative than conventional FDR levels, our goal here is to illustrate that CoHiM can identify high-confidence replicable SNP-level candidates under a conservative screening threshold, not to maximize the number of discoveries. Figure \ref{fig_Discoveries_of_different_methods} summarizes the number of SNPs discovered by each method. MaxP is the most conservative, yielding $176$ findings, all of which are also identified by the other methods. In contrast, CoHiM identifies $1,604$ SNPs, including $646$ uniquely detected by our approach. 
\begin{figure}[htbp]
    \centering
    \includegraphics[width = 0.8\linewidth]{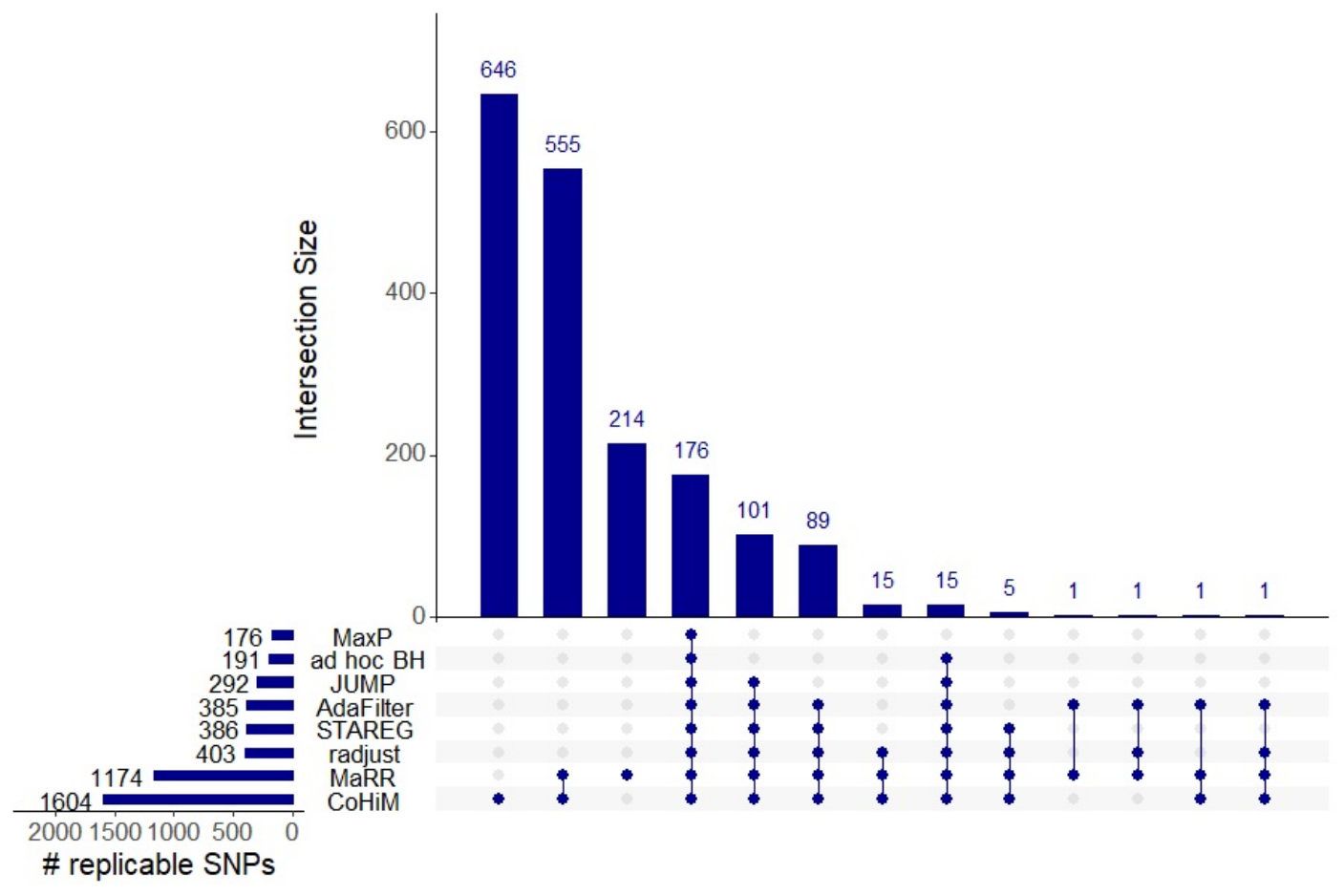}
    \caption{Number of SNPs identified as replicable by different methods in the type 2 diabetes data analysis.} 
    \label{fig_Discoveries_of_different_methods}
\end{figure}
Among the $646$ SNPs uniquely discovered by CoHiM, 30 are listed in the NHGRI-EBI GWAS Catalog (\url{https://www.ebi.ac.uk/gwas/}) as significantly associated with type 2 diabetes. To further validate the remaining SNPs, we mapped them to genes using the R package \texttt{snpGeneSets} \citep{mei2016snpgenesets}, resulting in $616$ SNPs mapped to $77$ genes. Many of these genes have been previously implicated in type 2 diabetes. For instance, genes such as {\it JAZF1}, {\it CDC123}, {\it THADA}, {\it ADAMTS9-AS2}, and {\it NOTCH2} have been reported to be associated with type 2 diabetes \citep{zeggini2008meta}. In particular, we highlight three genes with strong relevance to type 2 diabetes. In addition, SNPs in these gene regions can only be detected by other methods at less stringent FDR nominal levels, suggesting that CoHiM is able to identify weaker SNP-level signals under the same FDR level.
\begin{itemize}
\item {\it JAZF1}: a transcriptional regulator involved in ribosome biogenesis, protein synthesis, and insulin translation, with established links to diabetes risk \citep{kobiita2020diabetes}. CoHiM identifies $33$ unique SNPs mapped to this gene, including rs10245867 (rLIS: $2.64\times 10^{-6}$; male $p$-value: $1.03\times 10^{-8}$; female $p$-value: $6.64\times 10^{-5}$). The gene {\it JAZF1} is also detectable by STAREG at a less stringent FDR level with tagging SNPs. 
\item {\it ADAMTS9}: known to impair insulin sensitivity and increase diabetes risk \citep{graae2019adamts9}. CoHiM detects $25$ associated SNPs, including rs11914351 (rLIS: $3.71\times 10^{-5}$; male $p$-value: $8.53\times 10^{-4}$; female $p$-value: $5.70\times 10^{-2}$). Several tagging SNPs within the {\it ADAMTS9} gene region are detectable by STAREG at FDR levels $10^{-3}$ and $10^{-4}$, indicating that this gene locus is not missed at the gene level, but is identified through different SNPs and at less stringent FDR levels.
\item {\it NOTCH2}: implicated in poor glycemic control via elevated expression levels \citep{ghanem2020expression}. CoHiM links $9$ SNPs to this gene, including rs10127888 (rLIS: $4.44\times 10^{-5}$; male $p$-value: $2.82\times 10^{-2}$; female $p$-value: $1.52\times 10^{-2}$). The individual $p$-values for this SNP are modest, and the replicability signal should be interpreted with caution. CoHiM flags it based on evidence across both studies, but the evidence is weaker than that for the {\it JAZF1} or {\it ADAMTS9} loci.
\end{itemize}

Figure \ref{fig_manhattan} presents Manhattan plots for MaxP, STAREG, and CoHiM. The vertical axes display the $-\log_{10}$ transformations of each method's test statistics: $p_{\rm max}$ for MaxP, ${\rm Lfdr}$ for STAREG, and ${\rm rLIS}$ for CoHiM. Although the global patterns are similar, the methods differ substantially in thresholding behavior and interpretability. MaxP yields the fewest discoveries and lacks separation between rejected and non-rejected SNPs. STAREG detects more signals but with many borderline rejections. In contrast, CoHiM exhibits a sharper separation, with clear distinctions between high-confidence discoveries and nulls, enhancing the interpretability and reliability of the results.
\begin{figure}
    \centering
    \includegraphics[width = \linewidth]{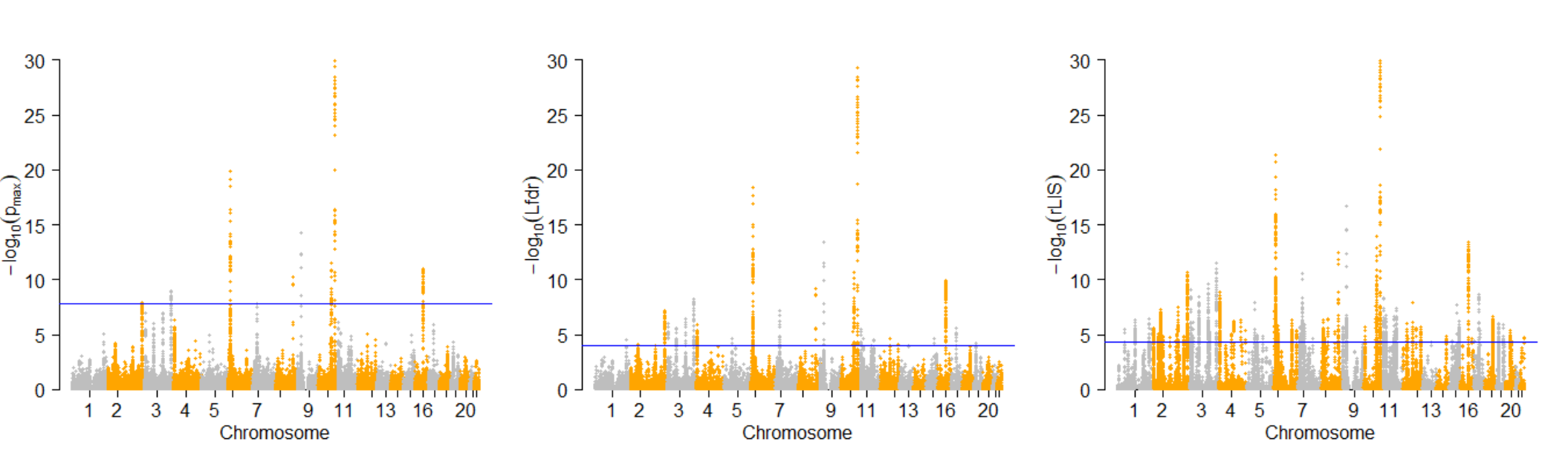}
    \caption{Manhattan plots of the type 2 diabetes GWAS data. The vertical axes show $-\log_{10}$ of each method's test statistic: $p_{\rm max}$ (the maximum $p$-value used by MaxP), ${\rm Lfdr}$ (the local false discovery rate statistic used by STAREG), and ${\rm rLIS}$ (replicability local index of significance used by CoHiM). Horizontal lines indicate the FDR threshold of $10^{-5}$.}
    \label{fig_manhattan}
\end{figure}

\subsection{Large-scale replicability analysis: type 2 diabetes and prostate cancer}\label{subsec_placo}

To further assess the scalability of CoHiM on large-scale data, we apply it to two publicly available GWAS datasets analyzed by \citet{ray2020powerful}: the type~2 diabetes (T2D) summary statistics from the DIAGRAM consortium (\url{https://cnsgenomics.com/data/t2d/}) and the prostate cancer (PCa) summary statistics (accession GCST006085) from the EBI GWAS Catalog (\url{https://ftp.ebi.ac.uk/pub/databases/gwas/summary_statistics/GCST006001-GCST007000/GCST006085/}).

\paragraph{Type 2 diabetes.}
After matching SNPs present in both T2D studies, the analysis retains $m = 118{,}364$ SNPs. Figure~\ref{fig_T2D_rejection_counts} reports the number of replicable SNPs identified by each method at the nominal FDR level $q=10^{-5}$. In this analysis, CoHiM detects approximately $1,600$ replicable SNPs, substantially more than any competing method. PLACO is the closest competitor, followed by QCH, radjust, STAREG, JUMP, and Cartesian HMM. Primo, MaxP, and \textit{ad~hoc}~BH are more conservative, whereas MaRR and AdaFilter report essentially no discoveries. The T2D dataset exhibits strong LD-induced local dependence and relatively large effect sizes. By explicitly modeling local dependence through an HMM, CoHiM can borrow information across neighboring SNPs, which helps explain its substantially higher power relative to methods that do not model the local dependence structure.

\begin{figure}[htbp]
\centering
\includegraphics[width=0.9\linewidth]{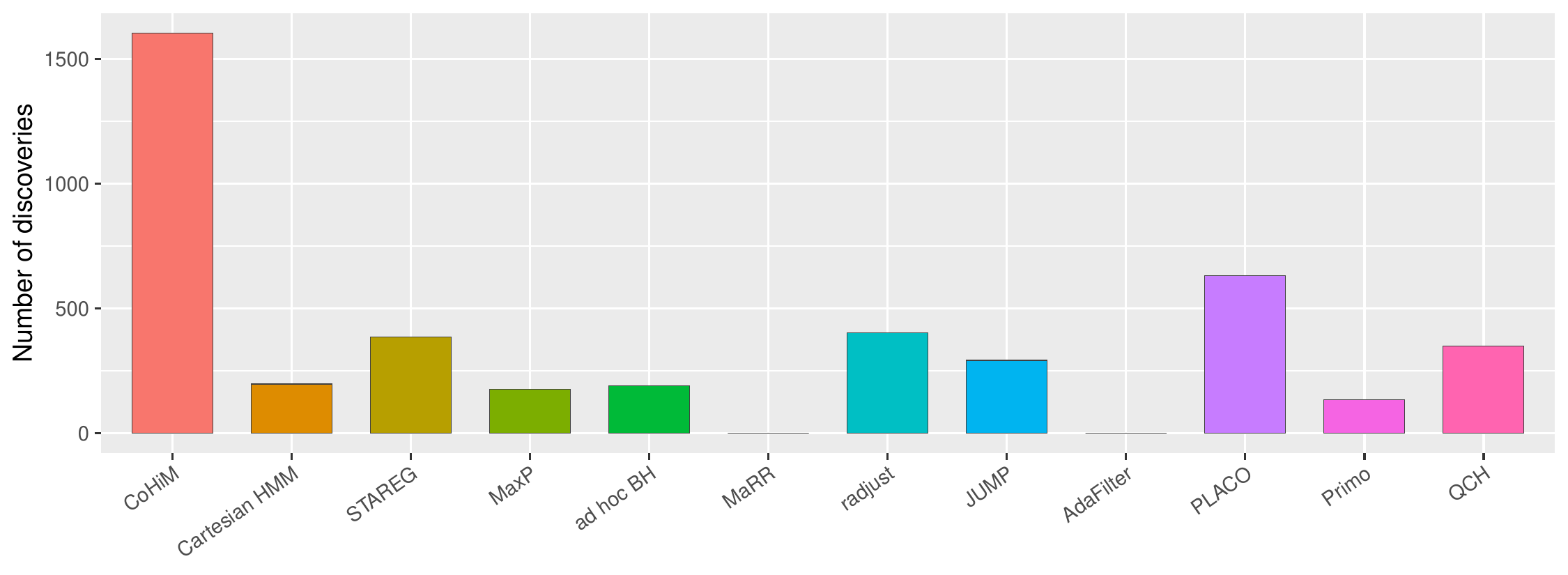}
\caption{Number of replicable SNPs identified by each method in the type~2 diabetes analysis at the nominal FDR level $q = 10^{-5}$.}
\label{fig_T2D_rejection_counts}
\end{figure}

% \begin{figure}[htbp]
% \centering
% \includegraphics[width=\linewidth]{figures/runtime_subsets_to_full.png}
% \caption{Runtime, in seconds, of each method in the prostate cancer analysis at the nominal FDR level $q = 10^{-5}$, across different subsample sizes and the full dataset. Both runtime and the number of SNPs are plotted on logarithmic scales.}
% \label{fig_PCa_runtime}
% \end{figure}

\begin{table}[htbp]
  \centering
  \caption{Computational time in seconds across different numbers of SNPs. Methods are listed in decreasing order of computational time on the one-million-SNP subset.}
  \label{tab:runtime_by_snp}
  \begin{tabular}{lrrrrr}
  % \toprule
            & 50K & 100K & 500K & 1M & 13.5M \\
  \midrule
  Primo     & 39.848   & 148.963  & 3,734.826  & 14,693.449 & -- \\
  Cartesian HMM & 306.752 & 677.083 & 4,577.498 & 10,823.967 & -- \\
  PLACO     & 200.917  & 381.496  & 2,171.799  & 3,833.786  & -- \\
  QCH       & 18.904   & 34.437   & 232.460   & 2,004.864  & -- \\
  AdaFilter & 4.116    & 8.539    & 47.528    & 87.492    & 793.828 \\
  CoHiM     & 0.523    & 1.982    & 6.586     & 10.994    & 390.596 \\
  JUMP      & 0.419    & 0.490    & 2.963     & 6.445     & 54.161 \\
  STAREG    & 0.260    & 0.238    & 3.556     & 2.993     & 46.766 \\
  \textit{ad hoc} BH & 0.012    & 0.058    & 0.310     & 2.049     & 9.869 \\
  MaxP      & 0.006    & 0.023    & 0.164     & 0.279     & 4.917 \\
  radjust   & 0.113    & 0.019    & 0.096     & 0.165     & 2.239 \\
  MaRR      & 0.032    & 0.001    & 0.005     & 0.007     & 0.071 \\
  % \bottomrule
  \end{tabular}
  \end{table}

\paragraph{Prostate cancer.}
The full PCa dataset contains $13{,}513{,}098$ overlapping SNPs. To enable a tractable comparison across all competing methods, we analyze the subsets by randomly subsampling $m = 50{,}000$, $100{,}000$, $500{,}000$ and $1{,}000{,}000$ SNPs, and also analyze the full dataset. Even the subsampled analyses are substantially larger than the simulation setting with $m=10{,}000$, thereby providing practical evidence of scalability. For each analyzed set, we apply CoHiM and the competing methods at the nominal FDR level $q=10^{-5}$, and record both the number of discoveries and the computational time.

Table~\ref{tab:runtime_by_snp} compares the computational time of each method across different numbers of SNPs. On the one-million-SNP subset, Primo, Cartesian HMM, PLACO, and QCH require $4.08$ hours ($14{,}693$ seconds), $3.01$ hours ($10,823$ seconds),  $1.06$ hours ($3{,}834$ seconds), and $0.56$ hours ($2{,}004$ seconds), respectively. Their computational times are substantially longer than those of AdaFilter ($87$ seconds) and CoHiM ($11$ seconds). To reduce computation and memory cost, we therefore omit Primo, Cartesian HMM, PLACO, and QCH from the full-data analysis. On the full dataset with $13.5$ million SNPs, CoHiM completes the analysis in $390$ seconds, whereas AdaFilter takes $793$ seconds. These results show that CoHiM is computationally efficient and scalable to large-scale GWAS datasets, even with tens of millions of hypotheses.

% \begin{figure}[htbp]
% \centering
% \includegraphics[width=\linewidth]{figures/rejections_subsets_to_full.png}
% \caption{Number of replicable SNPs identified by each method in the prostate cancer analysis at the nominal FDR level $q = 10^{-5}$, across different subsample sizes and the full dataset. The number of discoveries is plotted on a logarithmic scale.}
% \label{fig_PCa_rejection_counts}
% \end{figure}

\begin{table}[ht]
  \centering
  \caption{Number of discoveries across different numbers of SNPs. Methods are listed in decreasing order according to the number of discoveries on the one-million-SNP subset.}
  \label{tab:prostate_discoveries_by_size}
  \begin{tabular}{lrrrrr}
  % \toprule
  Method & 50K & 100K & 500K & 1M & 13.5M \\
  \midrule
  JUMP        & 203 & 401 & 1,856 & 3,669 & 50,006 \\
  PLACO       & 148 & 300 & 1,386 & 2,744 & -- \\
  QCH         & 138 & 280 & 1,296 & 2,239 & -- \\
  CoHiM       & 106 & 220 & 1,004 & 2,007 & 28,300 \\
  STAREG      & 105 & 218 & 995  & 1,973 & 27,386 \\
  radjust     & 101 & 211 & 982  & 1,957 & 27,276 \\
  AdaFilter   & 101 & 211 & 980  & 1,952 & 27,157 \\
  Cartesian HMM & 55 & 132 & 895& 1,654 & --      \\
  \textit{ad hoc} BH   & 39  & 88  & 395  & 797  & 11,024 \\
  MaxP        & 39  & 88  & 395  & 796  & 11,007 \\
  Primo       & 31  & 70  & 316  & 648  & -- \\
  MaRR        & 0   & 0   & 0    & 0    & 0 \\
  % \bottomrule
  \end{tabular}
  \end{table}

Table~\ref{tab:prostate_discoveries_by_size} displays the number of replicable SNPs identified by each method across different numbers of SNPs. Primo, Cartesian HMM, PLACO, and QCH are omitted from the full-data analysis due to their substantially higher computational cost. Among the remaining methods, JUMP reports the largest number of SNP-level discoveries. This result should be interpreted with caution. In the simulations, JUMP is generally conservative and has low power, whereas in the prostate cancer analysis it produces substantially more discoveries than the other methods. The larger SNP-level discovery count may reflect LD-induced correlations among nearby SNPs that are not explicitly modeled by JUMP, so that multiple correlated SNPs tagging the same underlying association signal are counted as separate discoveries.

In contrast, CoHiM, STAREG, radjust, and AdaFilter produce similar numbers of discoveries, suggesting broadly stable behavior among these methods. Cartesian HMM, although explicitly modeling feature dependence, yields fewer discoveries than CoHiM on the subsampled datasets, possibly reflecting the limitations of its parametric density assumptions. PLACO and QCH produce relatively large numbers of discoveries on the subsets, but their computational costs make full-data analysis less practical in this setting. The \textit{ad hoc} BH procedure, MaxP, and Primo are substantially more conservative, yielding discoveries on a smaller scale, while MaRR reports no discoveries, consistent with its conservative behavior in the simulations.

%Overall, we view the prostate cancer analysis primarily as evidence of computational scalability and empirical robustness, rather than as a definitive power comparison among methods. The results show that CoHiM remains computationally tractable on a dataset with more than $13.5$ million SNPs and produces discovery counts comparable to several well-calibrated competing procedures.

Taken together, these two analyses demonstrate that CoHiM remains computationally tractable for large-scale GWAS applications. CoHiM delivers substantially higher power in the T2D analysis, where the local dependence structure is pronounced, while producing discovery counts comparable to other well-calibrated competitors in the PCa analysis. These results support the scalability and practical utility of CoHiM beyond the simulation scale.

\section{Concluding remarks}\label{sec_conclusion} 
We have introduced CoHiM, a flexible and powerful framework for testing high-dimensional composite null hypotheses under dependence. Operating directly on $p$-values rather than raw data, CoHiM is practically advantageous in contexts where individual-level data are unavailable due to privacy or resource constraints. By modeling the local dependence structure through hidden Markov models and accommodating cross-study heterogeneity, CoHiM enables rigorous and scalable replicability analysis.

Our framework is first developed for the two-study case and then generalized to multiple studies via a novel e-value construction and aggregation strategy. Theoretically, we establish consistency of the maximum likelihood estimators and asymptotic FDR control. Empirically, CoHiM demonstrates favorable performance in simulations, maintaining valid FDR while achieving higher power than existing approaches. In our application to type 2 diabetes GWAS data, CoHiM uncovers novel SNP associations, including variants mapped to well-established diabetes-related genes not selected at the same FDR level by the competing procedures. Applications to larger-scale datasets demonstrate that CoHiM is computationally scalable and efficient.

Several important directions remain open. First, although $p$-values are readily available in most studies, they do not convey the direction of effects. Incorporating directional information into CoHiM could further improve replicability detection. Second, while we establish consistency of the MLE under the HMM framework, characterizing its convergence rates remains an open theoretical challenge. Lastly, HMMs capture local dependence effectively but may not adequately represent long-range or complex dependency structures such as spatial or network-based correlations. Extending CoHiM to such contexts while retaining scalability is a compelling avenue for future research.

\section*{Acknowledgement}
We thank Yan Li for her help with the simulation studies. This research is partially supported by NSF 2311249 and 2553817.

\bibliography{ref}

@article{li2025note,
  title={A note on e-values and multiple testing},
  author={Li, Guanxun and Zhang, Xianyang},
  journal={Biometrika},
  volume={112},
  number={1},
  pages={asae050},
  year={2025},
  publisher={Oxford University Press}
}

@article{deng2024joint,
  title={Joint mirror procedure: Controlling false discovery rate for identifying simultaneous signals},
  author={Deng, Linsui and He, Kejun and Zhang, Xianyang},
  journal={Biometrics},
  volume={80},
  number={4},
 pages={ujae142},
  year={2024},
  publisher={Oxford University Press}
}

@article{cao2022optimal,
  title={Optimal false discovery rate control for large scale multiple testing with auxiliary information},
  author={Cao, Hongyuan and Chen, Jun and Zhang, Xianyang},
  journal={The Annals of Statistics},
  volume={50},
  number={2},
  pages={807--857},
  year={2022},
  publisher={NIH Public Access}
}

@article{leroux1992maximum,
  title={Maximum-likelihood estimation for hidden {Markov} models},
  author={Leroux, Brian G},
  journal={Stochastic Processes and Their Applications},
  volume={40},
  number={1},
  pages={127--143},
  year={1992},
  publisher={Elsevier}
}

@article{riesz1928convergence,
  title={Sur la convergence en moyenne},
  author={Riesz, Fr{\'e}d{\'e}ric},
  journal={Acta Sci. Math},
  volume={4},
  number={1},
  pages={58--64},
  year={1928}
}

@article{bickel1998asymptotic,
  title={Asymptotic normality of the maximum-likelihood estimator for general hidden {Markov} models},
  author={Bickel, Peter J and Ritov, Ya’acov and Ryden, Tobias},
  journal={The Annals of Statistics},
  volume={26},
  number={4},
  pages={1614--1635},
  year={1998},
  publisher={Institute of Mathematical Statistics}
}

@article{alexandrovich2016nonparametric,
  title={Nonparametric identification and maximum likelihood estimation for hidden {Markov} models},
  author={Alexandrovich, Grigory and Holzmann, Hajo and Leister, Anna},
  journal={Biometrika},
  volume={103},
  number={2},
  pages={423--434},
  year={2016},
  publisher={Oxford University Press}
}

@article{sun2007oracle,
  title={Oracle and adaptive compound decision rules for false discovery rate control},
  author={Sun, Wenguang and Cai, T Tony},
  journal={Journal of the American Statistical Association},
  volume={102},
  number={479},
  pages={901--912},
  year={2007},
  publisher={Taylor \& Francis}
}

@article{cao2013optimal,
  title={The optimal power puzzle: scrutiny of the monotone likelihood ratio assumption in multiple testing},
  author={Cao, Hongyuan and Sun, Wenguang and Kosorok, Michael R},
  journal={Biometrika},
  volume={100},
  number={2},
  pages={495--502},
  year={2013},
  publisher={Oxford University Press}
}

@article{baum1970maximization,
  title={A maximization technique occurring in the statistical analysis of probabilistic functions of {Markov} chains},
  author={Baum, Leonard E and Petrie, Ted and Soules, George and Weiss, Norman},
  journal={Annals of Mathematical Statistics},
  volume={41},
  number={1},
  pages={164--171},
  year={1970},
  publisher={JSTOR}
}

@article{sun2009large,
  title={Large-scale multiple testing under dependence},
  author={Sun, Wenguang and Cai, Tony},
  journal={Journal of the Royal Statistical Society: Series B: Statistical Methodology},
  volume={71},
  number={2},
  pages={393--424},
  year={2009},
  publisher={Wiley Online Library}
}

@book{durrett2019probability,
  title={Probability: theory and examples},
  author={Durrett, Rick},
  year={2019},
  publisher={Cambridge University Press}
}

@article{dempster1977maximum,
  title={Maximum likelihood from incomplete data via the {EM} algorithm},
  author={Dempster, Arthur P and Laird, Nan M and Rubin, Donald B},
  journal={Journal of the Royal Statistical Society: Series B (Methodological)},
  volume={39},
  number={1},
  pages={1--22},
  year={1977},
  publisher={Wiley Online Library}
}

@inproceedings{robertson1988order,
  title={Order restricted statistical inference},
  author={Tim Robertson and Richard L. Dykstra and F. T. Wright},
  booktitle={Wiley Series in Probability and Mathematical Statistics},
  publisher={John Wiley and Sons},
  year={1988}
}

@article{fekete1923verteilung,
  title={{\"U}ber die Verteilung der Wurzeln bei gewissen algebraischen Gleichungen mit ganzzahligen Koeffizienten},
  author={Fekete, Michael},
  journal={Mathematische Zeitschrift},
  volume={17},
  number={1},
  pages={228--249},
  year={1923},
  publisher={Springer}
}

@article{birkhoff1931proof,
  title={Proof of the ergodic theorem},
  author={Birkhoff, George D},
  journal={Proceedings of the National Academy of Sciences},
  volume={17},
  number={12},
  pages={656--660},
  year={1931},
  publisher={National Acad Sciences}
}

@article{benjamini2009selective,
  title={Selective inference in complex research},
  author={Benjamini, Yoav and Heller, Ruth and Yekutieli, Daniel},
  journal={Philosophical Transactions of the Royal Society A: Mathematical, Physical and Engineering Sciences},
  volume={367},
  number={1906},
  pages={4255--4271},
  year={2009},
  publisher={The Royal Society Publishing}
}

@article{li2024stareg,
  title={STAREG: Statistical replicability analysis of high throughput experiments with applications to spatial transcriptomic studies},
  author={Li, Yan and Zhou, Xiang and Chen, Rui and Zhang, Xianyang and Cao, Hongyuan},
  journal={Plos Genetics},
  volume={20},
  number={10},
  pages={e1011423},
  year={2024},
  publisher={Public Library of Science San Francisco, CA USA}
}

@article{philtron2018maximum,
  title={Maximum rank reproducibility: a nonparametric approach to assessing reproducibility in replicate experiments},
  author={Philtron, Daisy and Lyu, Yafei and Li, Qunhua and Ghosh, Debashis},
  journal={Journal of the American Statistical Association},
  volume={113},
  number={523},
  pages={1028--1039},
  year={2018},
  publisher={Taylor \& Francis}
}

@article{bogomolov2018assessing,
  title={Assessing replicability of findings across two studies of multiple features},
  author={Bogomolov, Marina and Heller, Ruth},
  journal={Biometrika},
  volume={105},
  number={3},
  pages={505--516},
  year={2018},
  publisher={Oxford University Press}
}

@article{bogomolov2023replicability,
  title={Replicability across multiple studies},
  author={Bogomolov, Marina and Heller, Ruth},
  journal={Statistical Science},
  volume={38},
  number={4},
  pages={602--620},
  year={2023},
  publisher={Institute of Mathematical Statistics}
}

@article{heller2014replicability,
  title={Replicability analysis for genome-wide association studies},
  author={Heller, Ruth and Yekutieli, Daniel},
  journal={Annals of Applied Statistics},
  volume={8},
  number={1},
  pages={481--498},
  year={2014},
  publisher={Institute of Mathematical Statistics}
}

@article{benjamini1995controlling,
  title={Controlling the false discovery rate: a practical and powerful approach to multiple testing},
  author={Benjamini, Yoav and Hochberg, Yosef},
  journal={Journal of the Royal Statistical Society: Series B (Methodological)},
  volume={57},
  number={1},
  pages={289--300},
  year={1995},
  publisher={Wiley Online Library}
}

@article{lyu2023jump,
  title={{JUMP}: replicability analysis of high-throughput experiments with applications to spatial transcriptomic studies},
  author={Lyu, Pengfei and Li, Yan and Wen, Xiaoquan and Cao, Hongyuan},
  journal={Bioinformatics},
  volume={39},
  number={6},
  pages={btad366},
  year={2023},
  publisher={Oxford University Press}
}

@article{diabetes2012large,
  title={Large-scale association analysis provides insights into the genetic architecture and pathophysiology of type 2 diabetes},
  author={Morris, Andrew and Voight, Benjamin and Teslovich, Tanya and Ferreira, Teresa and Segr{\'e}, Ayellet and others},
  journal={Nature Genetics},
  volume={44},
  number={9},
  pages={981--990},
  year={2012},
  publisher={Nature Publishing Group US New York}
}

@article{zeggini2008meta,
  title={Meta-analysis of genome-wide association data and large-scale replication identifies additional susceptibility loci for type 2 diabetes},
  author={Zeggini, Eleftheria and Scott, Laura J and Saxena, Richa and Voight, Benjamin F and Marchini, Jonathan L and Hu, Tianle and de Bakker, Paul IW and Abecasis, Gon{\c{c}}alo R and Almgren, Peter and Andersen, Gitte and others},
  journal={Nature Genetics},
  volume={40},
  number={5},
  pages={638--645},
  year={2008},
  publisher={Nature Publishing Group US New York}
}

@article{mei2016snpgenesets,
  title={{snpGeneSets}: an r package for genome-wide study annotation},
  author={Mei, Hao and Li, Lianna and Jiang, Fan and Simino, Jeannette and Griswold, Michael and Mosley, Thomas and Liu, Shijian},
  journal={G3: Genes, Genomes, Genetics},
  volume={6},
  number={12},
  pages={4087--4095},
  year={2016},
  publisher={Oxford University Press}
}

@article{graae2019adamts9,
  title={{ADAMTS9} regulates skeletal muscle insulin sensitivity through extracellular matrix alterations},
  author={Graae, Anne-Sofie and Grarup, Niels and Ribel-Madsen, Rasmus and Lystbaek, Sara H and Boesgaard, Trine and Staiger, Harald and Fritsche, Andreas and Wellner, Niels and Sulek, Karolina and Kjolby, Mads and others},
  journal={Diabetes},
  volume={68},
  number={3},
  pages={502--514},
  year={2019},
  publisher={Am Diabetes Assoc}
}

@article{ghanem2020expression,
  title={Expression of {Notch} 2 and {ABCC8} genes in patients with type 2 diabetes mellitus and their association with diabetic kidney disease},
  author={Ghanem, Yehia and Ismail, Azza and Elsharkawy, Rania and Fathalla, Reem and El Feky, Amr},
  journal={Clinical Diabetology},
  volume={9},
  number={5},
  pages={306--312},
  year={2020}
}

@article{kobiita2020diabetes,
  title={The diabetes gene {JAZF1} is essential for the homeostatic control of ribosome biogenesis and function in metabolic stress},
  author={Kobiita, Ahmad and Godbersen, Svenja and Araldi, Elisa and Ghoshdastider, Umesh and Schmid, Marc W and Spinas, Giatgen and Moch, Holger and Stoffel, Markus},
  journal={Cell Reports},
  volume={32},
  number={1},
  page = {107846},
  year={2020},
  publisher={Elsevier}
}

@article{sesia2021false,
  title={False discovery rate control in genome-wide association studies with population structure},
  author={Sesia, Matteo and Bates, Stephen and Cand{\`e}s, Emmanuel and Marchini, Jonathan and Sabatti, Chiara},
  journal={Proceedings of the National Academy of Sciences},
  volume={118},
  number={40},
  pages={e2105841118},
  year={2021},
  publisher={National Acad Sciences}
}

@article{li2003modeling,
  title={Modeling linkage disequilibrium and identifying recombination hotspots using single-nucleotide polymorphism data},
  author={Li, Na and Stephens, Matthew},
  journal={Genetics},
  volume={165},
  number={4},
  pages={2213--2233},
  year={2003},
  publisher={Oxford University Press}
}

@article{abraham2022multiple,
  title={Multiple testing in nonparametric hidden {Markov} models: An empirical {Bayes} approach},
  author={Abraham, Kweku and Castillo, Isma\"el and Gassiat, Elisabeth},
  journal={Journal of Machine Learning Research},
  volume={23},
  number={94},
  pages={1--57},
  year={2022}
}

@article{barlow1972isotonic,
  title={The isotonic regression problem and its dual},
  author={Barlow, Richard and Brunk, Hugh},
  journal={Journal of the American Statistical Association},
  volume={67},
  number={337},
  pages={140--147},
  year={1972},
  publisher={Taylor \& Francis}
}

@article{lipton2016gbd,
  title={Global, regional, and national incidence, prevalence, and years lived with disability for 310 diseases and injuries, 1990-2015: a systematic analysis for the {Global Burden of Disease Study} 2015},
  author={Lipton, RB and Schwedt, TJ and Friedman, BW and others},
  journal={Lancet},
  volume={388},
  number={10053},
  pages={1545--1602},
  year={2016}
}

@book{williams1991probability,
  title={Probability with martingales},
  author={Williams, David},
  year={1991},
  publisher={Cambridge university press}
}

@article{zhao2020nonparametric,
  title={Nonparametric false discovery rate control for identifying simultaneous signals},
  author={Zhao, Sihai Dave and Nguyen, Yet Tien},  
  journal={Electronic Journal of Statistics},
  volume={14},
  number={1},
  pages={110-142},
  year={2020}
}

@article{wang2022detecting,
  title={Detecting multiple replicating signals using adaptive filtering procedures},
  author={Wang, Jingshu and Gui, Lin and Su, Weijie J and Sabatti, Chiara and Owen, Art B},
  journal={The Annals of Statistics},
  volume={50},
  number={4},
  pages={1890--1909},
  year={2022},
  publisher={Institute of Mathematical Statistics}
}

@article{liang2022powerful,
  title={Powerful Partial Conjunction Hypothesis Testing via Conditioning},
  author={Liang, Biyonka and Zhang, Lu and Janson, Lucas},
  journal={arXiv preprint arXiv:2212.11304},
  year={2022}
}

@article{bogomolov2023testing,
  title={Testing partial conjunction hypotheses under dependency, with applications to meta-analysis},
  author={Bogomolov, Marina},
  journal={Electronic Journal of Statistics},
  volume={17},
  number={1},
  pages={102--155},
  year={2023},
  publisher={The Institute of Mathematical Statistics and the Bernoulli Society}
}

@article{wang2022false,
  title={False discovery rate control with e-values},
  author={Wang, Ruodu and Ramdas, Aaditya},
  journal={Journal of the Royal Statistical Society Series B: Statistical Methodology},
  volume={84},
  number={3},
  pages={822--852},
  year={2022},
  publisher={Oxford University Press}
}

@article{vovk2021values,
  title={E-values: Calibration, combination and applications},
  author={Vovk, Vladimir and Wang, Ruodu},
  journal={The Annals of Statistics},
  volume={49},
  number={3},
  pages={1736--1754},
  year={2021},
  publisher={Institute of Mathematical Statistics}
}

@article{sun2024testing,
  title={Testing a Large Number of Composite Null Hypotheses Using Conditionally Symmetric Multidimensional Gaussian Mixtures in Genome-Wide Studies},
  author={Sun, Ryan and McCaw, Zachary R and Lin, Xihong},
  journal={Journal of the American Statistical Association},
  pages={1--13},
  year={2024},
  publisher={Taylor \& Francis}
}

@book{mackinnon2012introduction,
  title={Introduction to Statistical Mediation Analysis},
  author={MacKinnon, David},
  year={2012},
  publisher={Routledge}
}

@article{benjamini2008screening,
  title={Screening for partial conjunction hypotheses},
  author={Benjamini, Yoav and Heller, Ruth},
  journal={Biometrics},
  volume={64},
  number={4},
  pages={1215--1222},
  year={2008},
  publisher={Oxford University Press}
}

@article{visscher2012five,
  title={Five years of GWAS discovery},
  author={Visscher, Peter M and Brown, Matthew A and McCarthy, Mark I and Yang, Jian},
  journal={The American Journal of Human Genetics},
  volume={90},
  number={1},
  pages={7--24},
  year={2012},
  publisher={Elsevier}
}

@book{walters2000introduction,
  title={An introduction to ergodic theory},
  author={Walters, Peter},
  volume={79},
  year={2000},
  publisher={Springer Science \& Business Media}
}

@article{gerlach2010kolmogrov,
  author    = {Rainer Gerlach and Walter J. A. Pfeiffer},
  title     = {The Kolmogorov–Riesz compactness theorem},
  journal   = {Expositiones Mathematicae},
  volume    = {28},
  number    = {1},
  pages     = {61--65},
  year      = {2010},
  publisher = {Elsevier}
}

@article{ray2020powerful,
  title={A powerful method for pleiotropic analysis under composite null hypothesis identifies novel shared loci between type 2 diabetes and prostate cancer},
  author={Ray, Debashree and Chatterjee, Nilanjan},
  journal={PLoS genetics},
  volume={16},
  number={12},
  pages={e1009218},
  year={2020},
  publisher={Public Library of Science San Francisco, CA USA}
}

@article{gleason2020primo,
  title={Primo: integration of multiple GWAS and omics QTL summary statistics for elucidation of molecular mechanisms of trait-associated SNPs and detection of pleiotropy in complex traits},
  author={Gleason, Kevin J and Yang, Fan and Pierce, Brandon L and He, Xin and Chen, Lin S},
  journal={Genome biology},
  volume={21},
  number={1},
  pages={236},
  year={2020},
  publisher={Springer}
}

@article{mary2022querying,
  title={Querying multiple sets of P-values through composed hypothesis testing},
  author={Mary-Huard, Tristan and Das, Sarmistha and Mukhopadhyay, Indranil and Robin, St{\'e}phane},
  journal={Bioinformatics},
  volume={38},
  number={1},
  pages={141--148},
  year={2022},
  publisher={Oxford University Press}
}

@article{ding2023amdp,
  title={Amdp: an adaptive detection procedure for false discovery rate control in high-dimensional mediation analysis},
  author={Ding, Jiarong and Zhu, Xuehu},
  journal={Advances in Neural Information Processing Systems},
  volume={36},
  pages={65906--65935},
  year={2023}
}

@article{wang2019replicability,
  title={Replicability analysis in genome-wide association studies via Cartesian hidden Markov models},
  author={Wang, Pengfei and Zhu, Wensheng},
  journal={BMC bioinformatics},
  volume={20},
  number={1},
  pages={146},
  year={2019},
  publisher={Springer}
}

@article{li2012evaluating,
  title={Evaluating the effective numbers of independent tests and significant p-value thresholds in commercial genotyping arrays and public imputation reference datasets},
  author={Li, Miao-Xin and Yeung, Juilian MY and Cherny, Stacey S and Sham, Pak C},
  journal={Human genetics},
  volume={131},
  number={5},
  pages={747--756},
  year={2012},
  publisher={Springer}
}

@article{pe2008estimation,
  title={Estimation of the multiple testing burden for genomewide association studies of nearly all common variants},
  author={Pe'er, Itsik and Yelensky, Roman and Altshuler, David and Daly, Mark J},
  journal={Genetic Epidemiology: The Official Publication of the International Genetic Epidemiology Society},
  volume={32},
  number={4},
  pages={381--385},
  year={2008},
  publisher={Wiley Online Library}
}

@article{de2025large,
  title={Large-scale composite hypothesis testing procedure for omics data analyses},
  author={De Walsche, Anna{\"\i}g and Gauthier, Franck and Boissot, Nathalie and Charcosset, Alain and Mary-Huard, Tristan},
  journal={NAR Genomics and Bioinformatics},
  volume={7},
  number={3},
  pages={lqaf118},
  year={2025},
  publisher={Oxford University Press}
}

@article{dudbridge2006note,
  title={A note on permutation tests in multistage association scans},
  author={Dudbridge, Frank},
  journal={The American Journal of Human Genetics},
  volume={78},
  number={6},
  pages={1094--1095},
  year={2006},
  publisher={Elsevier}
}
\newpage
\appendix

\section{Estimation Process}
\subsection{Estimation}\label{subsec_est}
% The estimation procedure 
Let $\phi = (\pi, A, f_1, f_2)$ denote the collection of unknown parameters and density functions, with the true parameter denoted by $\phi^* = (\pi^*, A^*, f_1^*, f_2^*)$. The likelihood function for the observed paired $p$-values $(y_{1j}, y_{2j})_{j=1}^m$ is given by 
\begin{equation*}
    p_m\left((y_{1j}, y_{2j})_{j=1}^m; \phi\right) = \sum_{\boldsymbol s}\left\{\pi_{s_1}(\phi)f^{(s_1)}(y_{11}, y_{21}; \phi) \prod_{j=2}^m a_{s_{j-1}, s_j}(\phi) f^{(s_j)}(y_{1j}, y_{2j}; \phi)\right\},
\end{equation*}
where the summation is over all possible latent state sequences $\boldsymbol s = (s_1, \dots, s_m),$ and $f^{(s_j)}(y_{1j}, y_{2j}; \phi)$ is the joint density of $(y_{1j}, y_{2j})$ given latent state $s_j$ under parameter $\phi.$
The maximum-likelihood estimator of $\phi^*$ is defined as
\begin{align}
    \label{eq_MLE_definition}
    \widehat{\phi}_m = \underset{\phi \in \Phi}{\arg\max} \: p_m\left((y_{1j}, y_{2j})_{j=1}^m; \phi\right),
\end{align}
where $\Phi$ denotes the parameter space. 

To solve the maximum likelihood problem in (\ref{eq_MLE_definition}), we employ the expectation-maximization (EM) algorithm \citep{dempster1977maximum}, utilizing the forward-backward procedure \citep{baum1970maximization} for efficient computation. Define the forward probability $\alpha_j(s_j) = \mathbb P_{\phi^*}((y_{1t},y_{2t})_{t=1}^j, s_j)$ and the backward probability $\beta_j(s_j) =\mathbb P_{\phi^*}((y_{1t},y_{2t})_{t=j+1}^m\mid s_j)$, initialized by $\alpha_1(s_1) = \pi_{s_1}f^{(s_1)}(y_{11},y_{21})$ and $\beta_m(s_m) = 1.$ By the Markov property, these quantities can be computed recursively: 
\begin{align*}
    \alpha_{j+1}(s_{j+1}) =& \sum_{s_j=0}^3\alpha_j(s_j)a_{s_js_{j+1}}f^{(s_{j+1})}(y_{1,j+1},y_{2,j+1}), \quad \text{ and }
    \\
    \beta_j(s_{j})=&\sum_{s_{j+1}=0}^3\beta_{j+1}(s_{j+1})f^{(s_{j+1})}(y_{1,j+1},y_{2,j+1}) a_{s_js_{j+1}}.
\end{align*}
We define the posterior single-state and bi-state probabilities as $\gamma_j(s_j)=\mathbb P_{\phi^*}(s_j\mid(y_{1j}, y_{2j})_{j=1}^m)$ and $\xi_{j}(s_j, s_{j+1})=\mathbb P_{\phi^*}(s_j,s_{j+1}\mid(y_{1j}, y_{2j})_{j=1}^m)$. They satisfy the marginalization condition $\gamma_j(s_j) = \sum_{s_{j+1}=0}^3 \xi_j (s_j, s_{j+1})$ and can be computed by
\begin{align}\label{eq_gamma}
\gamma_j(s_j)&=\frac{\alpha_j(s_j)\beta_j(s_j)}{\sum_{s_j'=0}^3\alpha_j(s_j')\beta_j(s_j')},\quad\text{ and}
\\
\label{eq_xi}
\xi_{j}(s_j,s_{j+1})
&= \frac{\alpha_{j}(s_{j})\beta_{j+1}(s_{j+1})a_{s_{j}s_{j+1}}f^{(s_{j+1})}(y_{1,j+1},y_{2,j+1})} {\sum_{s_j'=0}^3\sum_{s_{j+1}'=0}^3\alpha_{j}(s_{j}')\beta_{j+1}(s_{j+1}')a_{s_{j}'s_{j+1}'}f^{(s_{j+1}')}(y_{1,j+1},y_{2,j+1})}.
\end{align}
The complete-data likelihood for $(y_{1j}, y_{2j}, s_j)_{j=1}^m$ 
is given by
\begin{align*}
L\left(\phi;(y_{1j}, y_{2j}, s_j)_{j=1}^m\right)
=\pi_{s_1}\prod_{j=2}^ma_{s_{j-1}s_j}\cdot\prod_{j=1}^mf^{(s_j)}(y_{1j},y_{2j}).
\end{align*} 
With an appropriate initialization $\phi^{(0)} = (\pi^{(0)}, A^{(0)}, f_1^{(0)}, f_2^{(0)})$, the EM algorithm proceeds by iteratively implementing the E-step and M-step, as described below. 

\textbf{E-step:} Given the current parameter estimate $\phi^{(t)} = (\pi^{(t)}, A^{(t)}, f_1^{(t)}, f_2^{(t)})$, compute the forward and backward probabilities $(\alpha_j^{(t)}(s_j)$ and $\beta_{j}^{(t)}(s_j))$, and use them to calculate the posterior single-state probabilities $\gamma_j^{(t)}(s_j)$ and posterior bi-state probabilities $\xi_j^{(t)}(s_j,s_{j+1})$ via Equations (\ref{eq_gamma}) and (\ref{eq_xi}). The conditional expectation of the complete-data log-likelihood function is 
\begin{align*}
&D\left(\phi\mid \phi^{(t)}\right)\\
=&\sum_{\boldsymbol s}\mathbb P_{\phi^{(t)}}\left(\boldsymbol s\mid (y_{1j}, y_{2j})_{j=1}^m\right)\log L\left(\phi;(y_{1j}, y_{2j})_{j=1}^m, \boldsymbol s\right)\\
=& \sum_{\boldsymbol s}\left[\mathbb P_{\phi^{(t)}}\left(\boldsymbol s\mid (y_{1j}, y_{2j})_{j=1}^m\right)\left\{\log(\pi_{s_1})+\sum_{j=2}^m\log (a_{s_{j-1}s_j})+\sum_{j=1}^m \log f^{(s_j)}(y_{1j},y_{2j})\right\}\right].
\end{align*}

\textbf{M-step:} Given the posterior probabilities computed in the E-step, we update the parameter estimates by 
\begin{align*}
\phi^{(t+1)}&=\underset{\pi, A,f_1,f_2}{\arg\max}\:D\left(\pi,A, f_1, f_2\mid \phi^{(t)}\right).
\end{align*}
Using Lagrange multipliers to enforce the normalization constraints, the updates for the initial distribution $\pi^{(t+1)} = (\pi_0^{(t+1)}, \pi_1^{(t+1)}, \pi_2^{(t+1)}, \pi_3^{(t+1)})$ and the transition matrix $A^{(t+1)}=(a_{k\ell}^{(t+1)})_{k,\ell=0,1,2,3}$ are given by 
\begin{align*}
    \pi_k^{(t+1)}=&\gamma_1^{(t)}(k) \quad \text{ for } k = 0,1,2,3, \quad\text{ and}
    \\
    a_{k\ell}^{(t+1)} =& \frac{\sum_{j=2}^m\xi_{j-1}^{(t)}(k,\ell)}
    {\sum_{j=2}^m\sum_{\ell'=0}^3 \xi_{j-1}^{(t)}(k,\ell')} \quad \text{ for }k,\ell=0,1,2,3.
\end{align*}
To update the non-null density functions of $f_1$ and $f_2,$ we solve the following weighted maximum likelihood problems under a monotonicity constraint. Specifically, 
\begin{align}
f_1^{(t+1)} =& \underset{f_1\in \mathcal H}{\arg\max}\sum_{j=1}^m\left\{\left(\gamma_j^{(t)}(2)+\gamma_j^{(t)}(3)\right)\log f_1(y_{1j})\right\}, \quad\text{and}
\label{eq_update_f1}
\\
f_2^{(t+1)} =& \underset{f_2\in \mathcal H}{\arg\max}\sum_{j=1}^m\left\{\left(\gamma_j^{(t)}(1)+\gamma_j^{(t)}(3)\right)\log f_2(y_{2j})\right\}, \label{eq_update_f2}  
\end{align}
where $\mathcal H$ is the class of non-increasing density functions supported on the interval $[0, 1],$ subject to the regularity condition 
$$\lim_{\delta\to0^+}\sup_{f\in\mathcal H}\int_0^\delta f(y){\rm d}y = 0,$$ which ensures the absence of point mass near zero and guarantees integrability. We iterate between the E-step and M-step until convergence of the observed data log-likelihood or until parameter changes fall below a prespecified threshold.

Next, we provide the details for solving the optimization problem in (\ref{eq_update_f1}) using the pool-adjacent-violators algorithm (PAVA; \citealp{robertson1988order,cao2022optimal}). The PAVA-based update for the non-null densities detailed below follows the estimation strategy of STAREG \citep{li2024stareg}; here we adapt it to the four-state HMM required for the composite null hypothesis. Let $0 = y_{1(0)} \leq y_{1(1)} \leq \cdots \leq y_{1(m)} \leq y_{1(m+1)} = 1$ denote the ordered $p$-values from study 1, and define the weight $$\Gamma_j^{(t)} = \gamma_j^{(t)}(2) + \gamma_j^{(t)}(3), \quad \mbox{for}\quad j=1, \ldots, m.$$ Since the objective in (\ref{eq_update_f1}) depends only on the values of $f_1$ evaluated at these ordered points and $f_1$ is constrained to be non-increasing, the solution must be piecewise-constant. Without loss of generality, assume the solution $f_1^{(\dagger)}$ satisfies $$f_1^{(\dagger)}(y) = f_1^{(\dagger)}(y_{1(j)}) \quad \mbox{for} \quad y \in (y_{1(j-1)}, y_{1(j)}],  j = 1, \ldots, m+1$$ and $f_1^{(\dagger)}(1) = 0$. Since $f_1^{(\dagger)}$ is a density function, it must satisfy $$\int_0^1 f_1^{(\dagger)}(y){\rm d}y = \sum_{j=1}^m f_1^{(\dagger)}(y_{1(j)})(y_{1(j)} - y_{1(j-1)}) = 1.$$ Therefore, we only need to estimate the values of $f_1$ at the jump points $y_{1(j)}$ for $j = 1, \ldots, m$. 
Let $z_j = f_1(y_{1(j)})$ and define the feasible set $\mathcal{Q} = \{\boldsymbol{z}= (z_1, \ldots, z_m)\in\mathbb R^m: z_1\geq\cdots \geq z_m \}$. The goal is to solve the following constrained optimization: 
\begin{align*}
    \widehat{\boldsymbol{z}} = \underset{\boldsymbol{z}\in \mathcal{Q}}{\arg\max}\sum_{j=1}^m \left\{\Gamma_{(j)}^{(t)}\log z_j\right\},\quad\text{ subject to } \sum_{j=1}^m \{(y_{1(j)} - y_{1(j-1)}) z_j\} = 1. 
\end{align*} 
To solve this, we use the method of Lagrange multiplier. The Lagrangian is given by
\begin{align*}
    L(\boldsymbol{z}, \zeta) =& \sum_{j=1}^m \left\{\Gamma_{(j)}^{(t)}\log z_j \right\}+ \zeta \left[\sum_{j=1}^m \{(y_{1(j)} - y_{1(j-1)}) z_j \}- 1\right].
\end{align*}
Taking derivatives with respect to $\zeta$ and $z_j,$ we have
\begin{align*}
    \tilde \zeta = - \sum_{j=1}^m \Gamma_{(j)}^{(t)},\quad
    \tilde z_j =\frac{\Gamma_{(j)}^{(t)}}{\sum_{k=1}^m \Gamma_{k}^{(t)} }\cdot \frac{1}{y_{1(j)} - y_{1(j-1)}}\quad \text{ for } j = 1,\ldots,m.
\end{align*}
Plugging $\tilde{\zeta}$ into the Lagrangian, the constrained maximization reduces to the following monotonic regression problem: 
\begin{align*}
    \widehat{\boldsymbol{z}} =& \underset{\boldsymbol{z}\in\mathcal Q}{\arg\min} \left\{-L(\boldsymbol{z}, \tilde{\zeta})\right\}\\
    =& \underset{\boldsymbol{z}\in\mathcal Q}{\arg\min} \sum_{j=1}^m \left(\Gamma_{(j)}^{(t)}\left[-\log z_j - \frac{-\{\sum_{k=1}^m \Gamma_{(k)}^{(t)}\}\{y_{1(j)} - y_{1(j-1)}\}}{\Gamma^{(t)}_{(j)}}z_j\right]\right).
\end{align*}
Let $u_j = -1/z_j,$ and $\boldsymbol{u} = (u_1, \ldots, u_m)$. 
This is equivalent to the following weighted least squares isotonic regression problem:
\begin{align*}
    \widehat{\boldsymbol{u}} =& \underset{\boldsymbol{u}\in\mathcal Q}{\arg\min} \sum_{j=1}^m \left(\Gamma_{(j)}^{(t)}\left[u_j - \frac{-\{\sum_{k=1}^m \Gamma_{(k)}^{(t)}\}\{y_{1(j)} - y_{1(j-1)}\}}{\Gamma_{(j)}^{(t)}}\right]^2\right).
\end{align*}
The solution has a closed-form max-min representation:
\begin{align*}
    \widehat{u}_j = \max_{b\geq j}\min_{a\leq j} \frac{-\left\{\sum_{k=1}^m \Gamma_{(k)}^{(t)}\right\}\sum_{k=a}^b\{y_{1(k)} - y_{1(k-1)}\}}{\sum_{k=a}^b\Gamma_{(k)}^{(t)}},
\end{align*}
which can be efficiently computed via the pool-adjacent-violators algorithm (PAVA) \citep{barlow1972isotonic}. 
According to Theorem 3.1 of \cite{barlow1972isotonic}, the update of (\ref{eq_update_f1}) is given by $$f_1^{(t+1)}(y_{1(j)}) = \widehat{z}_j = -1/\widehat{u}_j\quad\text{ for } j=1,\ldots,m.$$ The update for $f_2^{(t+1)}$ in (\ref{eq_update_f2}) proceeds in exactly the same way by replacing $y_{1j}$ with $y_{2j}$ and $\Gamma_j^{(t)}$ with $\gamma_j^{(t)}(1) + \gamma_j^{(t)}(3).$ We omit the details. 

\paragraph{Estimation details.}
The EM algorithm was implemented with the following hyperparameters:
\begin{itemize}
  \item maximum number of iterations: \texttt{maxIter} $= 200$;
  \item convergence tolerance: \texttt{tol} $= 10^{-3}$, based on the relative change in the observed-data log-likelihood;
  \item input $p$-value floor: before model fitting, we replaced any zero or nonpositive $p$-value by $10^{-15}$ to avoid numerical instability caused by exact zeros;
  \item non-null density floor: during the EM algorithm, estimated non-null density values were truncated below at $10^{-15}$ to prevent undefined log-likelihood terms and unstable posterior-probability calculations.
\end{itemize}

We used a deterministic initialization scheme. For a given pair of studies with $p$-value vectors $(y_a,y_b)$, the marginal null proportions were first estimated separately for the two studies and truncated above at $0.999$. The initial four-state stationary probabilities were then set to
\[
\pi^{(0)} =
\bigl(\widehat\pi_{0,a}\widehat\pi_{0,b},\;
\widehat\pi_{0,a}(1-\widehat\pi_{0,b}),\;
(1-\widehat\pi_{0,a})\widehat\pi_{0,b},\;
(1-\widehat\pi_{0,a})(1-\widehat\pi_{0,b})\bigr),
\]
where $\widehat\pi_{0,a}$ and $\widehat\pi_{0,b}$ denote the estimated marginal null proportions for studies $a$ and $b$, respectively.
The transition matrix was initialized as
\[
A^{(0)} =
\begin{pmatrix}
0.90 & 0.04 & 0.04 & 0.02 \\
0.28 & 0.30 & 0.14 & 0.28 \\
0.28 & 0.14 & 0.30 & 0.28 \\
0.14 & 0.28 & 0.28 & 0.30
\end{pmatrix}.
\]
The initial non-null emission densities were initialized by decreasing functions based on the observed $p$-values, with $f_1(y_a)=1-y_a$ and $f_2(y_b)=1-y_b$, before applying the PAVA-based M-step updates.

In the two-study simulations, one pairwise EM fit was performed for each simulated replicate. In the $n$-study simulations, \texttt{CoHiM} fits all $\binom{n}{2}$ pairwise HMMs, with each pairwise fit using the same EM hyperparameters described above.

\subsection{Oracle procedure}\label{subsec_oracle_test}
Consider the oracle setting where the true parameter $\phi^* = (\pi^*, A^*, f_1^*, f_2^*)$ is known. Define the replicability Local Index of Significance (rLIS) for the $j$th hypothesis as the posterior probability that the hypothesis is not replicable, i.e., $s_j$ belongs to the non-replicable configuration set $\{0,1,2\}$, given all observed $p$-value pairs:
$$
{\rm rLIS}_j =  \mathbb P_{\phi^*}\left(s_j\in\{0,1,2\}\mid (y_{1j'}, y_{2j'})_{j'=1}^m\right)\quad \text{ for }j = 1, \ldots, m.
$$
Let $I(B)$ denote the indicator function for an event $B$, i.e., $I(B) = 1$ if $B$ is true and $0$ otherwise.
For a rejection threshold $\lambda$, we reject $H_{0j}$ if ${\rm rLIS}_j \leq \lambda$. The total number of rejections is
$$
R(\lambda) = \sum_{j=1}^{m}I({\rm rLIS}_j\le \lambda).
$$
The number of false rejections is
$$
V(\lambda) = \sum_{j=1}^{m}I\left({\rm rLIS}_j \le \lambda, s_j \in \{0, 1, 2 \}\right).
$$
By the law of total expectation,
\begin{align}
    \mathbb E\{V(\lambda)\} =& \mathbb E\left\{\sum_{j=1}^{m}I({\rm rLIS}_j \le \lambda, s_j \in \{0, 1, 2 \})\right\}\notag\\
    =& \mathbb E\left[\mathbb E\left\{\sum_{j=1}^{m}I({\rm rLIS}_j \le \lambda, s_j \in \{0, 1, 2 \})\mid (y_{1j}, y_{2j})_{j=1}^m\right\}\right]\notag\\
    =& \mathbb E\left\{\sum_{j=1}^{m}I({\rm rLIS}_j \le \lambda){\rm rLIS}_j\right\}. 
    \label{eq_EV}
\end{align}
To control the FDR at a pre-specified level $q,$
we define FDR and false discovery proportion (FDP) as
$$
{\rm FDR}(\lambda) = \mathbb E\left[\text{FDP}(\lambda)\right],\quad \text{FDP}(\lambda) =\frac{V(\lambda)}{R(\lambda)\vee 1}= \frac{\sum_{j=1}^{m}I( {\rm rLIS}_j \le \lambda, s_j \in \{0,1,2\})}{\left\{\sum_{j=1}^{m}I({\rm rLIS}_j \le \lambda )\right\}\vee 1}.
$$
Using (\ref{eq_EV}), we approximate the FDP by
$$
{\rm FDP}(\lambda) \approx \frac{\sum_{j=1}^{m}I({\rm rLIS}_j \le \lambda){\rm rLIS}_j}{\left\{\sum_{j=1}^{m}I({\rm rLIS}_j \le \lambda )\right\}\vee 1}.
$$
To control the FDR at the target level $q$ using the oracle posterior quantities, we choose the largest threshold $\lambda$ such that the estimated FDP does not exceed $q:$ 
\begin{align}
\begin{aligned}
    &\lambda_{\rm OR} = \sup\left\{\lambda \geq 0: \frac{\sum_{j=1}^{m}I({\rm rLIS}_j \le \lambda){\rm rLIS}_j}{\left\{\sum_{j=1}^{m}I({\rm rLIS}_j \le \lambda )\right\}\vee 1} \leq q\right\},\\
    &\text{and reject } H_{0j} \text{ if } {\rm rLIS}_j \leq \lambda_{\rm OR} \quad \text{ for }j =1,\ldots,m.
\end{aligned} \label{eq_oracle_test_1}
\end{align}

Let ${\rm rLIS}_{(1)} \leq \cdots \leq {\rm rLIS}_{(m)}$ be the ordered ${\rm rLIS}$ values and $H_{0(1)}, \ldots, H_{0(m)}$ be the corresponding hypotheses. Suppose that $\lambda_{\rm OR}$ yields $R$ rejections, i.e., ${\rm rLIS}_{(R)}\leq \lambda_{\rm OR} < {\rm rLIS}_{(R+1)}$. The rejection criterion (\ref{eq_oracle_test_1}) is equivalent to the following step-up procedure:
\begin{equation}
    \label{eq_oracle_test_2}
    \begin{aligned}
        &\text{Let } R = \max\left\{r: \frac{1}{r}\sum_{j=1}^r {\rm rLIS}_{(j)}\le q\right\};\\
        &\text{then reject all $H_{0(j)}$}\quad \text{ for }j=1,\dots,  R.
    \end{aligned}
    \notag
\end{equation}

\subsection{Data-driven procedure}\label{sec_numeric_test}
With the maximum likelihood estimator $\widehat{\phi}_m = (\widehat{\pi}, \widehat{A}, \widehat{f}_1, \widehat{f}_2)$ obtained from the EM algorithm in Section~\ref{subsec_est}, we compute the estimated forward and backward probabilities as follows:
\begin{align}
    \widehat{\alpha}_1(s_1) =& \widehat{\pi}_{s_1}\widehat{f}^{(s_1)}(y_{11}, y_{21}),
    \quad \quad \widehat{\beta}_m(s_m) = 1,
    \label{eq_alpha_beta_initial}
    \\
    \widehat{\alpha}_{j+1}(s_{j+1}) =& \sum_{s_j = 0}^3 \widehat{\alpha}_j(s_j) \widehat{a}_{s_j, s_{j+1}} \widehat{f}^{(s_{j+1})}(y_{1,j+1}, y_{2,j+1})\quad \text{ and }
    \label{eq_alpha_update}
    \\
    \widehat{\beta}_j(s_j) =& \sum_{s_{j+1} = 0}^3 \widehat{\beta}_{j+1}(s_{j+1}) \widehat{a}_{s_j, s_{j+1}} \widehat{f}^{(s_{j+1})}(y_{1,j+1}, y_{2,j+1}).
    \label{eq_beta_update}
\end{align}
Using these quantities, the estimated replicability Local Index of Significance is 
\begin{align}
    \widehat{\mathrm{rLIS}}_j =& \mathbb P_{\widehat{\phi}_m} \left(s_j\in\{0, 1, 2\}\mid (y_{1j'}, y_{2j'})_{j'=1}^m\right)= \frac{\sum_{s_j=0}^2 \widehat{\alpha}_j(s_j)\widehat{\beta}_j(s_j)}{\sum_{s_j=0}^3 \widehat{\alpha}_j(s_j)\widehat{\beta}_j(s_j)}.\label{eq_test stats}
\end{align}
To implement the data-driven step-up procedure, we first order the estimated replicability Local Index of Significance values $\widehat{\mathrm{rLIS}}_{(1)}\leq \cdots \leq \widehat{\mathrm{rLIS}}_{(m)}$ with the corresponding replicability null hypotheses denoted by $H_{0(1)}, \ldots, H_{0(m)}$. Given a target FDR level $q\in (0, 1)$, we have the step-up procedure
\begin{align}
    \begin{aligned}
        &\widehat{R} = \max\left\{r: \frac{1}{r}\sum_{j=1}^r \widehat{\mathrm{rLIS}}_{(j)} \leq q\right\},\\
        & \text{and reject } H_{0(j)} \quad \text{ for }j = 1,\ldots,\widehat R.
    \end{aligned}
\end{align}

\section{Proof of main results}
\subsection{Proof of Proposition \ref{prop_compact}}
\begin{proof}
    Since the spaces of the transition matrix $A$ and the stationary probability $\pi$ are bounded and closed with finite dimensions, they are compact. We just need to show the non-increasing density function space $\mathcal H$ with the constraint $\lim_{\delta\to0^+}\sup_{f\in\mathcal H}\int_0^\delta f(y){\rm d}y = 0$ is compact under the Hellinger distance $d_H(\cdot, \cdot)$.

    First, recall the definition of the Hellinger distance between two densities $g_1,g_2$ on $[0,1]$:
    \[
      d_H(g_1,g_2)
      = \Bigl(\tfrac12 \int_0^1 \bigl(g_1^{1/2}(y) - g_2^{1/2}(y)\bigr)^2 \,dy\Bigr)^{1/2}.
    \]
    This immediately gives
    \[
      d_H(g_1,g_2)
      = \tfrac1{\sqrt2} \bigl\|g_1^{1/2} - g_2^{1/2}\bigr\|_2,
    \]
    where $\|\cdot\|_2$ denotes the $L_2$ metric. Therefore, up to the constant factor $1/\sqrt2$, the Hellinger distance is exactly the $L^2$ distance on the space of square-root densities. Hence, compactness in one metric implies compactness in the other.

    Denote $\mathcal H^{1/2} = \{h: h^2\in\mathcal H\}$. Then $h$ satisfies $\|h\|_2^2 = \int_0^1h(y)^2{\rm d}y = 1$ as well as $\lim_{\delta\to0^+}\sup_{h\in\mathcal H^{1/2}}\int_0^\delta h(y)^2{\rm d}y = 0$. Thus $\mathcal H^{1/2} \subseteq L^2[0,1]$ and we just need to show that $\mathcal H^{1/2}$ is compact with respect to the $L_2$ norm.

    For any $\epsilon > 0$, there exists $\delta = \delta(\epsilon) > 0$ such that $\sup_{h\in\mathcal H^{1/2}}\int_0^\delta h(y)^2{\rm d}y < \epsilon$. Extend $h$ by zero outside of $[0, 1]$ and let $\zeta_\delta h(y) = h(y + \delta)$ as the $\delta$-shift of $h$. Then for any $h \in \mathcal H^{1/2}$,
    \begin{align*}
        \|\zeta_\delta h - h\|_2^2 =& \int_{-\infty}^\infty\left\{h(y+\delta) - h(y)\right\}^2{\rm d}y
        \\
        =& \int_0^{1-\delta}\left\{h(y+\delta) - h(y)\right\}^2{\rm d}y + \int_{-\delta}^0 h(y+\delta)^2{\rm d}y + \int_{1-\delta}^1 h(y)^2{\rm d}y 
        \\
        \leq& \int_0^{1-\delta}\{h(y+\delta)^2+h(y)^2\}{\rm d}y -2\int_0^{1-\delta}\{h(y+\delta)h(y)\}{\rm d}y + 2\int_0^\delta h(y)^2{\rm d}y 
        \\
        \leq& \int_0^12h(y)^2{\rm d}y -2\int_0^{1-\delta}h(y)^2{\rm d}y + 2\int_0^\delta h(y)^2{\rm d}y 
        \\
        =& 4\int_0^\delta h^2(y){\rm d}y < 4\epsilon,
    \end{align*}
    which means that $\mathcal H^{1/2}$ is equicontinuous. Additionally, since $h(y) = 0$ for $y \notin [0, 1]$, we have $\lim_{r\to\infty}\int_{|y|>r}h(y)^2{\rm d}y = 0$. In other words, $\mathcal H^{1/2}$ is equitight. By Fr\'echet–Kolmogorov theorem \citep{gerlach2010kolmogrov}, $\mathcal H^{1/2}$ is relative compact by the equicontinuity and equitightness. 

    To show the compactness of $\mathcal H^{1/2}$, we just need to show it is closed. For any $h_n\in\mathcal H^{1/2}$ satisfying $\|h_n - h\|_2 \to 0$ as $n\to\infty$ for some $h\in L^2[0,1]$, our goal is to show $h\in\mathcal H^{1/2}$. First, 
    \begin{align*}
        \|h\|_2^2 = \lim_{n\to\infty}\|h_n\|_2^2 = 1.
    \end{align*}
    Next, we show that $h$ is also non-increasing. For any $\varepsilon > 0$, denote 
    \begin{align}
        E_n(\varepsilon) = \{y: |h_n(y) - h(y)| > \varepsilon\}.
    \end{align}
    Denote $\mu(\cdot)$ as the Lebesgue measure.
    Thus 
    \begin{align*}
        \varepsilon \mu\{E_n(\varepsilon)\}^{1/2} =& \left(\int_{E_n(\varepsilon)} \varepsilon^{2} {\rm d}y\right)^{1/2}\\
        \leq & \left(\int_{E_n(\varepsilon)} |h_n(y) - h(y)|^2 {\rm d}y\right)^{1/2}\\
        \leq & \left(\int_{0}^1 |h_n(y) - h(y)|^2 {\rm d}y\right)^{1/2}\\
        = & \|h_n - h\|_2 \to 0 \text{ as } n\to \infty,
    \end{align*}
    which implies that $h_n$ converges to $h$ in measure, or equivalently, for any $\varepsilon > 0$, 
    \begin{align*}
        \lim_{n\to\infty}\mu\{E_n(\varepsilon)\} = 0.
    \end{align*}
    By the theorem of Riesz \citep{riesz1928convergence}, there exists a subsequence $\{h_{n_k}\}$ of $\{h_n\}$, such that $h_{n_k} \rightarrow h$ almost everywhere. Since $h_{n_k}$ are non-increasing, we could conclude that $h$ is also non-increasing. Finally, by the triangle inequality, we have
    \begin{align*}
        \left\{\int_0^\delta h^2(y){\rm d}y\right\}^{1/2} \leq & \left\{\int_0^\delta h_n(y)^2{\rm d}y\right\}^{1/2} + \left\{\int_0^\delta \{h(y) - h_n(y)\}^2{\rm d}y\right\}^{1/2}
        \\
        \leq & \left\{\int_0^\delta h_n(y)^2{\rm d}y\right\}^{1/2} + \|h - h_n\|_2 \to 0
    \end{align*}
    as $\delta\to0^+$ and $n\to\infty$. Therefore, we have $h\in\mathcal H^{1/2}$ and thus $\mathcal H^{1/2}$ is closed and compact. Consequently, we know $\mathcal H$ is compact with respect to the Hellinger distance.
\end{proof}
\subsection{Proof of Theorem \ref{thm_consistency}}
    For any $\phi\in\Phi$ with $d(\phi, \phi^*)$, define the conditional distribution of $(y_{1j}, y_{2j})_{j=1}^m$ given $s_1 = k$ for $ k = 0, 1, 2, 3$ as 
    \begin{align*}
        \ell_m(k; \phi) :
        =& f^{(k)}(y_{11}, y_{21}; \phi) \sum_{s_2}\dots\sum_{s_m} a_{k, s_2}(\phi)f^{(s_2)}(y_{12}, y_{22}; \phi) \prod_{j = 3}^m a_{s_{j-1}, s_j}(\phi)f^{(s_j)}(y_{1j}, y_{2j}; \phi),
    \end{align*}
    where $s_j$ denotes the hidden state of the $j$th gene for $j = 1,\ldots,m$. 
    Denote the largest $\ell_m(k;\phi)$ for $k = 0, 1, 2, 3$ as
    $$
    q_m(\phi) = 
    \max_{k = 0,1,2,3} \ell_m(k; \phi).
    $$
    Then the likelihood function $p_m(\phi) = p_m\left((y_{1j}, y_{2j})_{j=1}^m; \phi\right)$ satisfies 
    \begin{align}
        p_m(\phi) =& \sum_{k=0,1,2,3} \pi_k(\phi) \ell_m(k; \phi)
        \leq  q_m(\phi),
        \label{eq_p_m_over_q_m_upper_bound}
    \end{align}
    where $\pi_k(\phi) = P_\phi(s_j = k)$ for $j = 1, \ldots, m$, $k = 0, 1, 2, 3$, and it satisfies $\sum_{k=0}^3\pi_k(\phi) = 1$.
    
    In addition, assume $q_m(\phi) = \ell_m(k_0; \phi)$ for some $k_0 \in \{0, 1, 2, 3\}$. Then
    \begin{align}
        p_m(\phi) =& \sum_{k=0,1,2,3} \pi_k(\phi) \ell_m(k; \phi)
        \geq  \pi_{k_0}(\phi) \ell_m(k_0; \phi)
        \geq \varepsilon_0 q_m(\phi),
        \label{eq_p_m_over_q_m_lower_bound}
    \end{align}
    where (\ref{eq_p_m_over_q_m_lower_bound}) holds due to (C2): $\pi_k(\phi) \geq \varepsilon_0$ for $k = 0,1,2,3$.
    
    Therefore, combining (\ref{eq_p_m_over_q_m_upper_bound}) and (\ref{eq_p_m_over_q_m_lower_bound}) and taking the logarithm, we have
    \begin{equation}
        \label{eq_similar_limit}
        \log\left(\varepsilon_0\right) \leq \log \frac{p_m(\phi)}{q_m(\phi)} \leq 0.
    \end{equation}

    Dividing (\ref{eq_similar_limit}) by $m$, we have
    \begin{align}
        \label{eq_similar_limit_2}
        \frac{1}{m}\log\left(\varepsilon_0\right) \leq \frac{1}{m}\log p_m(\phi) - \frac{1}{m}\log q_m(\phi) \leq 0.
    \end{align}
    Letting $m\rightarrow \infty$, the lower bound of inequality (\ref{eq_similar_limit_2}) tends to $0$. Hence $m^{-1}\log q_m( \phi)$ and  $m^{-1}\log p_m( \phi)$ converges to the same limit in probability. Taking the expectation on all terms of inequality (\ref{eq_similar_limit_2}), we know $m^{-1}\mathbb E_{\phi^*}\log q_m( \phi)$ has the same limit as $m^{-1}\mathbb E_{\phi^*}\log p_m( \phi)$. By Theorem 2 in \cite{leroux1992maximum}, there exists some $H(\phi^*, \phi) < \infty$ satisfying
    \begin{align*}
        \lim_{m\rightarrow \infty}\frac{1}{m} \mathbb E_{\phi^*} \{\log p_m(\phi)\}  &= H(\phi^*, \phi), \text{ and }\\
        \lim_{m\rightarrow \infty}
        \frac{1}{m} \log p_m(\phi)  &= H(\phi^*, \phi) \text{ almost surely under } \phi^*.
    \end{align*} 
    We also have 
    \begin{align*}
        \lim_{m\rightarrow \infty}\frac{1}{m} \mathbb E_{\phi^*} \{\log q_m(\phi)\}  &= H(\phi^*, \phi), \text{ and }\\
        \lim_{m\rightarrow \infty}
        \frac{1}{m} \log q_m(\phi)  &= H(\phi^*, \phi) \text{ almost surely under } \phi^*.
    \end{align*}

    Replacing $\phi$ by $\phi^*$, we get the limit $H(\phi^*, \phi^*)$. Lemma 6 in \cite{leroux1992maximum} gives that $H(\phi^*, \phi) < H(\phi^*, \phi^*)$ for $\phi \neq \phi^*$.
    Letting $\varepsilon = \{H(\phi^*, \phi^*) - H(\phi^*, \phi)\}/2$, there exists $m_\varepsilon$ such that,
    \begin{align}
        & \frac{1}{m_\varepsilon} \mathbb E_{\phi^*}\{\log q_{m_\varepsilon}(\phi)\} < H(\phi^*, \phi) + \varepsilon =  H(\phi^*, \phi^*) - \varepsilon.
        \label{eq_neighbour_sup_bound}
    \end{align}

    Denote $O_{\phi, r} = \{\phi' \in \Phi: d(\phi', \phi) < r\}$ as a ball centered at $\phi$ with radius $r > 0$, where $d(\phi', \phi)$ is the distance between $\phi'$ and $\phi$ defined in (\ref{eq_distance_of_parameters}). 
    $\mathbb E_{\phi^*} [\{\log (\sup_{\phi' \in O_{\phi,r}} q_{m_\varepsilon}(\phi'))\}^+] < \infty$ by (C4).
    Therefore, $\mathbb E_{\phi^*} [\{\log (\sup_{\phi' \in O_{\phi, r}} q_{m_\varepsilon}(\phi'))\}^+]$ is a bounded monotone increasing function of $r$. Since $f_1(\phi), f_2(\phi)$ are continuous functions of $\phi$, $p_m(\phi)$ and $q_m(\phi)$ are also continuous. By the Monotone Convergence Theorem and the continuity of $q_{m_\varepsilon}(\phi)$, we have
    \begin{align*}
        \frac{1}{m_\varepsilon}\mathbb E_{\phi^*} \left\{\log \left(\sup_{\phi' \in O_{\phi, r}} q_{m_\varepsilon}(\phi')\right)\right\} \rightarrow \frac{1}{m_\varepsilon}\mathbb E_{\phi^*} \left\{\log  q_{m_\varepsilon}(\phi)\right\}\text{ as } 
        r \rightarrow 0.
    \end{align*}
    Then there exists $r_0>0$, such that 
    \begin{align}
        \frac{1}{m_\varepsilon}\mathbb E_{\phi^*} \left\{\log \left(\sup_{\phi' \in O_{\phi, r_0}} q_{m_\varepsilon}(\phi')\right)\right\} <&\frac{1}{m_\varepsilon} \mathbb E_{\phi^*}\{\log q_{m_\varepsilon}(\phi)\}+ \varepsilon/2\notag\\
        <& H(\phi^*, \phi^*) - \varepsilon /2,\label{eq_q_at_n_epsion_and_entropy}
    \end{align}
    where the second inequality holds due to (\ref{eq_neighbour_sup_bound}).

    Noting that $p_m(\phi)$ and $q_m(\phi)$ are continuous with respect to $\phi\in O_{\phi, r}$ for $r<\delta_0$ and $k = 0,1,2,3$. Thus we can extend (\ref{eq_p_m_over_q_m_upper_bound}) and (\ref{eq_p_m_over_q_m_lower_bound}) as follows: 
    \begin{align*}
        \varepsilon_0\sup_{\phi'\in O_{\phi, r}} q_m(\phi') \leq \sup_{\phi'\in O_{\phi, r}}p_m(\phi') \leq & \sup_{\phi'\in O_{\phi, r}} q_m(\phi').
    \end{align*}
    Taking the logarithm, we have
    $$
    \frac{1}{m}\log\left(\varepsilon_0\right) \leq \frac{1}{m}\log \left\{\sup_{\phi'\in O_{\phi, r}}p_m(\phi')\right\} - \frac{1}{m}\log\left\{ \sup_{\phi'\in O_{\phi, r}} q_m(\phi')\right\} \leq 0.
    $$
    Thus, $m^{-1}\log \{\sup_{\phi'\in O_{\phi,r}}p_m(\phi')\}$ and $m^{-1}\log \{\sup_{\phi'\in O_{\phi,r}}q_m(\phi')\}$ converge to the same limit in probability. Define
    \begin{equation*}
        J(\phi^*, \phi; O_{\phi, r}) = \lim_{m\rightarrow \infty} \frac{1}{m}\mathbb E_{\phi^*}\left\{\log \left(\sup_{\phi'\in O_{\phi, r}}q_m(\phi')\right)\right\}.
    \end{equation*}
    In addition, we have
    \begin{align}
        \frac{1}{m}\log \left\{\sup_{\phi'\in O_{\phi, r}}q_m(\phi')\right\} &\rightarrow J(\phi^*, \phi; O_{\phi, r})\text{ in probability, and }\notag\\
        \frac{1}{m}\log \left\{\sup_{\phi'\in O_{\phi, r}}p_m(\phi')\right\} &\rightarrow J(\phi^*, \phi; O_{\phi, r})\text{ in probability}.
        \label{eq_neighbour_limit_of_p_m}
    \end{align}

    By the construction of $q_m(\phi) = q_m((y_{1j}, y_{2j})_{j=1}^m; \phi)$, Lemma 3 of \cite{leroux1992maximum} shows that  $\log q_m((y_{1j}, y_{2j})_{j=1}^m; \phi)$ is subadditive, which means for any sequence $(y_{1j}, y_{2j})_{j=1}^m$,
    \begin{align*}
        \log q_{s+t}((y_{1j}, y_{2j})_{j=1}^{s+t}; \phi) \leq \log q_s((y_{1j}, y_{2j})_{j=1}^s; \phi) + \log q_t((y_{1j}, y_{2j})_{j=s+1}^{s+t}; \phi).
    \end{align*}
    By the property of subadditive processes \citep{fekete1923verteilung},
    \begin{equation*}
        J(\phi^*, \phi; O_{\phi, r}) = \inf_{m} \frac{1}{m}\mathbb E_{\phi^*}\left\{\log \left(\sup_{\phi'\in O_{\phi, r}}q_m(\phi')\right)\right\}
    \end{equation*}
    which implies that
    \begin{equation}
        \label{eq_neighbour_entropy_and_q_at_n_epsilon}
        J(\phi^*, \phi; O_{\phi, r}) \leq \frac{1}{m_\varepsilon}\mathbb E_{\phi^*}\left\{\log \left(\sup_{\phi'\in O_{\phi, r}}q_{m_\varepsilon}(\phi')\right)\right\}.
    \end{equation}
    Consequently, by (\ref{eq_neighbour_limit_of_p_m}), (\ref{eq_neighbour_entropy_and_q_at_n_epsilon}) and (\ref{eq_q_at_n_epsion_and_entropy}), we have as $m\rightarrow \infty$,
    \begin{align}
        \frac{1}{m}\log \left\{\sup_{\phi'\in O_{\phi, r}}p_m(\phi')\right\}\rightarrow & J(\phi^*, \phi; O_{\phi, r}) \text{ in probability, and }\notag\\
         J(\phi^*, \phi; O_{\phi, r}) \leq& \frac{1}{m_\varepsilon}\mathbb E_{\phi^*}\left\{\log \left(\sup_{\phi'\in O_{\phi, r}}q_{m_\varepsilon}(\phi')\right)\right\}\notag\\
         <& H(\phi^*, \phi^*) - \varepsilon/2.
         \label{eq_limit_log_likelihood_in_neighborhood}
    \end{align}

    Next, we use (\ref{eq_limit_log_likelihood_in_neighborhood}) to show the consistency of $\widehat{\phi}_m$. Let $C$ be any closed subset of $\Phi$, not containing $\phi^*$. Since $\Phi$ is compact, $C$ is also compact and is covered by the union of finite open sets $\bigcup_{h=1}^d O_{\phi_h, r}$, where $\{\phi_1, \ldots, \phi_d\}$ is a finite set in $C$.

    Therefore,
    \begin{align*}
        & \sup_{\phi \in C}\left\{\log p_m(\phi) - \log p_m(\phi^*)\right\}\\
        \leq&\max_{1\leq h \leq d} \left[m\left\{\frac{1}{m}\log \left(\sup_{\phi \in O_{\phi_h, r}} p_m(\phi)\right) - \frac{1}{m} \log p_m (\phi^*)\right\}\right]\\
         \rightarrow & -\infty \text{ in probability},
    \end{align*}
    where the limit in the last line holds due to  (\ref{eq_limit_log_likelihood_in_neighborhood}) and 
    that $m^{-1}\log p_m(\phi^*) \to H(\phi^*, \phi^*)$ almost surely as $m\to\infty$ by Birkhoff's ergodic theorem \citep{birkhoff1931proof}.
    Since $\widehat{\phi}_m$ is a maximum likelihood estimator, $\log p_m (\widehat{\phi}_m)\geq \log p_m (\phi^*)$. Therefore, $\widehat{\phi}_m$ cannot be in $C$.

    In other words, for any open set $O_{\phi, r}\subseteq \Phi$ containing $\phi^*$, $\widehat{\phi}_m$ must be in $O_{\phi, r}$ for large $m$. Letting $r \rightarrow 0$, we conclude that $\widehat{\phi}_m\rightarrow\phi^*$ in probability.

\subsection{Proof of Theorem \ref{thm_data_driven_FDR_control} }

    First, we introduce some notations used in the proof. Consider an infinite hidden Markov model with hidden states $\{S_j\}_{j=-\infty}^{\infty}$ and $p$-values $(y_{1j}, y_{2j})_{j=-\infty}^{\infty}$. Denote the following test statistics
    \begin{align*}
        T_j =& \mathbb P_{\phi^*}(s_j\in\{0, 1, 2\}\mid (y_{1j}, y_{2j})_{j=1}^m),\\
        \widehat{T}_j =& \mathbb P_{\widehat{\phi}_m}(s_j\in\{0, 1, 2\}\mid (y_{1j}, y_{2j})_{j=1}^m),\\
        T_j^\infty =& \mathbb P_{\phi^*}(s_j\in\{0, 1, 2\}\mid (y_{1j}, y_{2j})_{j=-\infty}^{\infty}),\\
        \widehat{T}_j^\infty =& \mathbb P_{\widehat\phi_m}(s_j\in\{0, 1, 2\}\mid (y_{1j}, y_{2j})_{j=-\infty}^{\infty})\quad \text{ for }j=1,\ldots,m.
    \end{align*}
    For any test statistics $\xi_j \in \{T_j, \widehat T_j, T_j^\infty, \widehat T_j^\infty\}$ corresponding to the null hypothesis $H_{0j}$, consider the testing procedure based on ordered $\xi_{(1)} \leq \cdots \leq \xi_{(m)}$ with corresponding null hypotheses $H_{0(1)}, \ldots, H_{0(m)}$. We have the number of rejections given by 
    \begin{align}
        R_0 = \max\left\{r: \frac{1}{r}\sum_{j=1}^r \xi_{(j)} \leq q\right\}.
        \label{eq_R_0}
    \end{align}
    We reject $H_{0(j)}$ for $j = 1,\ldots,R_0$. 
    An equivalent algorithm is 
    \begin{align}
        \lambda_0 = \sup\left\{\lambda\in(0,1): \frac{\sum_{j=1}^m \xi_jI(\xi_j\leq \lambda)}{\left\{\sum_{j=1}^m I(\xi_j \leq \lambda)\right\} \vee 1} \leq q\right\}.
        \label{eq_lambda_0}
    \end{align}
    The rejection threshold can be written as $\lambda_0 = \xi_{(R_0)}$. The total number of false rejections is $V_0 = \sum_{j=1}^m I(\xi_j \leq \lambda_0 \text{ and } s_j\in\{0, 1, 2\})$. Replacing $\xi_j$ by $T_j, \widehat{T}_j, T_j^\infty$ and $\widehat{T}_j^\infty$, the number of rejections and number of false rejections are denoted by $(R, V)$, $(\widehat{R}, \widehat{V})$, $(R^\infty, V^\infty)$ and $(\widehat{R}^\infty, \widehat{V}^\infty)$. Moreover, we define the corresponding rejection thresholds as $\widehat{\lambda}_{\mathrm{OR}}, \widehat{\lambda}_{\mathrm{rLIS}}, \widehat{\lambda}_{\mathrm{OR}}^{\infty}, \widehat{\lambda}_{\mathrm{rLIS}}^{\infty}$.    
    
    Next, we consider the distribution of $T_j^{\infty}$. Since $\{S_j\}_{j=-\infty}^{\infty}$ is stationary, irreducible, and aperiodic, the two-sided generalization of Theorem 6.1.3 in \cite{durrett2019probability} implies that $\{T_j^\infty\}$ is ergodic. Therefore,  $T_j^\infty$ are identically distributed. 
    Denote the cumulative distribution function of $T_j^\infty$ as
    \begin{align*}
        \mathbb P_{\phi^*}(T_j^\infty \leq t) =  G^\infty(t).
    \end{align*}
    Denote the conditional cumulative distribution function of $T_j^{\infty}$ given $s_j = k$ as
    \begin{align*}
        \mathbb P_{\phi^*}(T_j^\infty \leq t\mid s_j = k) = G_k^\infty(t) \quad\text{ for }k = 0, 1, 2, 3.
    \end{align*} 
    Thus for $\phi^* = (\pi^*, A^*, f_1^*, f_2^*),$
    $$
    G^\infty(t) = \pi_0^* G_0^\infty(t) + \pi_1^* G_1^\infty(t) + \pi_2^* G_2^\infty(t) + \pi_3^* G_3^\infty(t).
    $$
    Let 
    \begin{align}
        \alpha_* = \inf\{0\leq t\leq 1: G^\infty(t) = 1\}.
        \label{eq_alpha_*}
    \end{align}
    
    By the forward-backward algorithm \citep{baum1970maximization}, 
    \begin{align*}
        T_j^\infty = \frac{\sum_{s_j=0}^2 \alpha_j(s_j)\beta_j(s_j)}{\sum_{s_j=0}^3 \alpha_j(s_j)\beta_j(s_j)},
    \end{align*}
    where $\alpha_j(s_j) = \mathbb P_{\phi^*}((y_{1t}, y_{2t})_{t=-\infty}^j, s_j)$ and $\beta_j(s_j) = \mathbb P_{\phi^*}((y_{1t}, y_{2t})_{t=j+1}^\infty \mid s_j)$. $\alpha_j(\cdot)$ and $\beta_j(\cdot)$ can be derived recursively by $\alpha_{j+1}(s_{j+1}) = \sum_{s_j=0}^3 \alpha_j(s_j) a_{s_j,s_{j+1}} f^{(s_j+1)} (y_{1,j+1}, y_{2,j+1})$ and $\beta_j(s_j) = \sum_{s_{j+1}=0}^3 \beta_{j+1}(s_{j+1}) f^{(s_{j+1})} (y_{1,j+1}, y_{2,j+1})$. 
    Since the joint distribution of $(y_{1j}, y_{2j})_{j=-\infty}^\infty$ is continuous, and $T_j^\infty$ is a continuous map from $(y_{1j}, y_{2j})_{j=-\infty}^\infty$ to (0, 1), the probability density function of $T_j^\infty$ is positive and continuous on $(0, \alpha_*)$. 
    It suffices to show that $G^\infty$ is strictly increasing in $(0, \alpha_*)$, which is needed in the proof of Lemma \ref{lemma_consistency_when_lambda_less_than_alpha}. For some threshold $\lambda > 0$, define the number of rejections and false rejections as
    \begin{align*}
        R_\lambda^\infty =& \sum_{j=1}^m I(T_j^{\infty} \leq \lambda), \\
        V_\lambda^\infty =& \sum_{j=1}^m I(T_j^{\infty} \leq \lambda, s_j\in\{0, 1, 2\}).
    \end{align*}
    Thus, we have the expectations 
    \begin{align*}
        \mathbb E (R_\lambda^\infty) =& m G^\infty(\lambda),\\
        \mathbb E (V_\lambda^\infty) =& m(\pi_0 G_0^\infty(\lambda) + \pi_1 G_1^\infty(\lambda) + \pi_2 G_2^\infty(\lambda)).
    \end{align*} 
    Therefore, the marginal FDR  is 
    $$
    Q_{\text{OR}}^\infty(\lambda) = \mathbb E (V_\lambda^\infty)/\mathbb E (R_\lambda^\infty) = (\pi_0 G_0^\infty(\lambda) + \pi_1 G_1^\infty(\lambda) + \pi_2 G_2^\infty(\lambda))/G^\infty(\lambda).
    $$
    Theorem 1 of \cite{sun2009large} implies that $Q_{\text{OR}}^\infty(\lambda)$ is increasing in $\lambda$. Define the threshold based on the marginal FDR as 
    $$\lambda_{\text{OR}}^\infty = \sup\{\lambda: Q_{\text{OR}}^\infty(\lambda) \leq q\}.$$ 
    Since $G^{\infty}(t) = 1 $ is equivalent to the statement that $ G_s^\infty(t) = 1$ for $s = 0, 1, 2, 3$, we have
    \begin{align*}
        Q_{\mathrm{OR}}^\infty (\alpha_*) = \pi_0+\pi_1+\pi_2 > q
    \end{align*}
    under (C2) with $\pi_3 < 1-q$. Without loss of generality, we assume $\lambda_{\mathrm{OR}}^\infty < \alpha_*.$

    With the notations above, we will prove Theorem \ref{thm_data_driven_FDR_control} as follows. In Step 1, we show that the total number of rejections $R$ and $\widehat R$ approach infinity almost surely. In Step 2, we show that $\mathbb E|R/\widehat R - 1|  \rightarrow 0$ and $\mathbb E|V/\widehat V - 1| \rightarrow 0$ as $m\to \infty$. Finally, we show the asymptotic FDR control in Step 3.

    \textbf{Step 1.} Asymptotic behavior of rejection numbers. 

        Recall that $\widehat{\lambda}_{\mathrm{OR}}^{\infty}$ and $\widehat{\lambda}_{\mathrm{rLIS}}^{\infty}$ are the corresponding rejection threshold given by $\{T_j^\infty\}_{j=1}^m$ and $\{\widehat{T}_j^\infty\}_{j=1}^m$. First, we show $\widehat{\lambda}_{\mathrm{OR}}^{\infty} \rightarrow \lambda_{\mathrm{OR}}^{\infty}$ and $\widehat{\lambda}_{\mathrm{rLIS}}^{\infty} \rightarrow \lambda_{\mathrm{OR}}^{\infty}$ in probability by Lemma \ref{lemma_consistency_of_thresholds}.
        \begin{lemma}
            \label{lemma_consistency_of_thresholds}
            Assume (C1)-(C4) hold.
            $\widehat{\lambda}_{\mathrm{OR}}^{\infty} \rightarrow \lambda_{\mathrm{OR}}^{\infty}$ and $\widehat{\lambda}_{\mathrm{rLIS}}^{\infty} \rightarrow \lambda_{\mathrm{OR}}^{\infty}$ in probability.
        \end{lemma}
        
    We next show $\widehat{R} \rightarrow \infty$ almost surely.
        For simplicity, denote $(y_{1j}, y_{2j})_{j=j_1}^{j_2}$ as $y_{j_1}^{j_2}$ for any $j_1 < j_2$.
        By (C2), $\varepsilon_0 \leq a_{lk}(\phi) \leq 1$ for all $l, k$ and (C5), for any $k,k'=0,1,2,3$, $f^{(k')}(y_{1, j+1}, y_{2,j+1}; \phi)/f^{(k)}(y_{1, j+1}, y_{2,j+1}; \phi) \leq \rho_0(y_{1,j+1}, y_{2,j+1})$. Then for any states $k,k'$ and $l = 0,1,2,3$, 
        \begin{align*}
            &\frac{\mathbb P_{\phi}(S_{j+1}=k'\mid S_j = l, y_1^m)}{\mathbb P_{\phi}(S_{j+1}=k\mid S_j = l, y_1^m)}\\
            =& \frac{\mathbb P_{\phi}(S_{j+1}=k', S_j = l, y_1^m)}{\mathbb P_{\phi}(S_{j+1}=k, S_j = l, y_1^m)}\\
            =& \frac{\mathbb P_{\phi}(S_{j+1}=k', y_1^m\mid S_j = l)}{\mathbb P_{\phi}(S_{j+1}=k, y_1^m\mid S_j = l)}\\
            =& \frac{\sum_{k_0=0}^3 \mathbb P_{\phi}(S_{j+1}=k', S_{j+2} = k_0, y_1^m\mid S_j = l)}{\sum_{k_0=0}^3 \mathbb P_{\phi}(S_{j+1}=k, S_{j+2} = k_0, y_1^m\mid S_j = l)}\\
            =& \frac{\mathbb P_{\phi}(y_1^j\mid S_{j} = l;\phi)\sum_{k_0=0}^3 a_{lk'}(\phi)a_{k'k_0}(\phi)f^{(k')}(y_{1, j+1}, y_{2,j+1};\phi)\mathbb P_{\phi}(y_{j+2}^m\mid S_{j+2} = k_0;\phi)}{\mathbb P_{\phi}(y_1^j\mid S_{j} = l;\phi)\sum_{k_0=0}^3 a_{lk}(\phi)a_{kk_0}(\phi)f^{(k)}(y_{1, j+1}, y_{2,j+1};\phi)\mathbb P_{\phi}(y_{j+2}^m\mid S_{j+2} = k_0;\phi)}\\
            =& \frac{f^{(k')}(y_{1, j+1}, y_{2,j+1};\phi)\sum_{k_0=0}^3 a_{lk'}(\phi)a_{k'k_0}(\phi)\mathbb P_{\phi}(y_{j+2}^m\mid S_{j+2} = k_0;\phi)}{f^{(k)}(y_{1, j+1}, y_{2,j+1};\phi)\sum_{k_0=0}^3 a_{lk}(\phi)a_{kk_0}(\phi)\mathbb P_{\phi}(y_{j+2}^m\mid S_{j+2} = k_0;\phi)}\\
            \leq &  \frac{f^{(k')}(y_{1, j+1}, y_{2,j+1};\phi)\sum_{k_0=0}^3 \mathbb P_{\phi}(y_{j+2}^m\mid S_{j+2} = k_0;\phi)}{f^{(k)}(y_{1, j+1}, y_{2,j+1};\phi)\sum_{k_0=0}^3 \varepsilon_0^2 \mathbb P_{\phi}(y_{j+2}^m\mid S_{j+2} = k_0;\phi)}\\
            =& \varepsilon_0^{-2}\frac{f^{(k')}(y_{1, j+1}, y_{2,j+1};\phi)}{ f^{(k)}(y_{1, j+1}, y_{2,j+1};\phi)}\\
            \leq& \varepsilon_0^{-2} \rho_0(y_{1,j+1}, y_{2,j+1}).
        \end{align*}

        Let $\tau_0(y_1, y_2) = (1+3\varepsilon_0^{-2}\rho_0(y_1, y_2))^{-1}$.
        Since $\sum_{k'=0}^3\mathbb P_{\phi}(S_{j+1}=k'\mid S_j = l, y_1^m) = 1$, we conclude that for all $k,l = 0,1,2,3$,
        \begin{align}
            \mathbb P_{\phi}(S_{j+1}=k\mid S_j = l, y_1^m) =& \frac{\mathbb P_{\phi}(S_{j+1}=k\mid S_j = l, y_1^m)}{\sum_{k'=0}^3\mathbb P_{\phi}(S_{j+1}=k'\mid S_j = l, y_1^m)}\notag\\ 
            =&\frac{1}{1+\sum_{k'\neq k}\frac{\mathbb P_{\phi}(S_{j+1}=k'\mid S_j = l, y_1^m)}{\mathbb P_{\phi}(S_{j+1}=k\mid S_j = l, y_1^m)}}\notag\\
            \geq& \{1+ 3\varepsilon_0^{-2}\rho_0(y_{1,j+1}, y_{2, j+1})\}^{-1}.\notag
        \end{align}
        Define 
        \begin{align*}
            \tau_0(y_{1,j+1}, y_{2, j+1}) = \{1+ 3\varepsilon_0^{-2}\rho_0(y_{1,j+1}, y_{2, j+1})\}^{-1}.
        \end{align*}
        Then 
        \begin{align}
            \mathbb P_{\phi}(S_{j+1}=k\mid S_j = l, y_1^m) \geq \tau_0(y_{1,j+1}, y_{2, j+1}).
            \label{eq_condition_of_lemma_A_3}
        \end{align}
        Then we apply Lemma \ref{lemma_R_to_infty} below to show $R\rightarrow \infty$ and $\widehat{R}\rightarrow \infty$ almost surely as $m\rightarrow\infty$. 
        \begin{lemma}
            \label{lemma_R_to_infty}
            If (\ref{eq_condition_of_lemma_A_3}) and (C1)-(C3) hold, then $R/m \geq G^\infty(q/2)>0$ and $\widehat{R}/m \geq G^\infty(q/2)>0$ almost surely. 
        \end{lemma}

        \textbf{Step 2.} Convergence of $R/\widehat R$ and $V/\widehat V$ in expectation.

        In this step, we show $\mathbb E|R/\widehat R - 1| \rightarrow 0$ and $\mathbb E|V/\widehat V - 1| \rightarrow 0$ by Lemma \ref{lemma_consistency_when_lambda_less_than_alpha}.
    \begin{lemma}\label{lemma_consistency_when_lambda_less_than_alpha}
        If $0< \lambda_{\text{OR}}^\infty < \alpha_*$, then $\mathbb E|R/\widehat{R}  - 1| \rightarrow 0$ and $E\mathbb |V/\widehat{V} - 1| \rightarrow 0$ as $m\to\infty$.
    \end{lemma}
    When $\lambda_{\rm OR} \geq \alpha_*$, $R/m$ tends to $1$ as $m\to\infty$, which means all the null hypotheses will be rejected. This is not a feasible case.
    
    \textbf{Step 3.} Asymptotic FDR control.
    
    We have
    \begin{align*}
        \frac{\widehat V}{\widehat R\vee 1} - \frac{V}{R\vee 1} \leq & \frac{\widehat V }{\widehat R \vee 1}\left(1 - \frac{V}{\widehat V \vee 1}\right) + \frac{V}{R \vee 1}\left(\frac{R\vee 1}{\widehat R \vee 1} - 1\right).
    \end{align*}
    Since $\widehat V/(\widehat R\vee 1) \leq 1$, we have
    \begin{align*}
        0 \leq \bigg|\mathbb E\left\{\frac{\widehat V}{\widehat R\vee 1}\left(1 - \frac{V}{\widehat V \vee 1}\right)\right\}\bigg|\leq \mathbb E\bigg|\frac{\widehat V}{\widehat R\vee 1}\left(1 - \frac{V}{\widehat V \vee 1}\right)\bigg| \leq \mathbb E\bigg|1 - \frac{V}{\widehat V \vee 1}\bigg|.
    \end{align*}
    By Lemma \ref{lemma_consistency_when_lambda_less_than_alpha}, we have 
    \begin{align*}
        \mathbb E\left\{\frac{\widehat V}{\widehat R\vee 1}\left(1 - \frac{V}{\widehat V \vee 1}\right)\right\} \rightarrow 0 \text{ as } m\to\infty.
    \end{align*}
    Similarly, we also have
    \begin{align*}
        \mathbb E\left\{\frac{V}{ R\vee 1}\left(\frac{R}{\widehat R \vee 1} - 1\right)\right\} \rightarrow 0 \text{ as } m\to\infty.
    \end{align*}
    Therefore,
    \begin{align*}
        {\rm FDR} - {\rm FDR}_{\rm OR} = \mathbb E\left(\frac{\widehat V}{\widehat R\vee 1}\right) - \mathbb E\left(\frac{V}{R\vee 1}\right) \to 0 \text{ as } m\rightarrow \infty.
    \end{align*}
    Since ${\rm FDR}_{\rm OR} \leq q$ by Theorem \ref{thm_oracle_FDR_control}, we know ${\rm FDR}$ is asymptotically controlled.
    
\subsection{Proof of Theorem \ref{thm_eBH_procedure}}
\label{subsec_eBH_proof}
    \begin{proof}
        We prove this theorem via 3 steps. We first give the convergence of the pairwise FDR level $\widehat q_m$. Then we show the convergence of the pairwise rLIS quantities. Finally, we establish the asymptotic performance of the e-values and the eBH procedure.

        Without loss of generality, consider the case that $\widehat R(\widehat q_m) > 0$. Since the discovery set for the eBH procedure is the subset of the intersection of pairwise discovery sets with pairwise FDR level $\widehat q_m$, in this case, we have $\widehat R^{k\ell}(\widehat q_m) > 0$ for any $1 \leq k<\ell \leq n$.
        
        \paragraph{Step 1. Convergence of $\widehat q_m$.} First,
        we show that there exists a constant $q_*\in[q_-, q]$ such that $\widehat q_m \to q_*$ almost surely as $m\to\infty$.
        
        For simplicity, denote $\widehat T_j^{k\ell} = \widehat{\rm rLIS}_j^{k\ell}$ for studies $k$ and $\ell$. Without loss of generality, we use $\widehat{\lambda}_{\rm rLIS}^{k\ell}(\widehat q_m), \widehat\alpha_j^{k\ell} $ and $ \widehat\beta_j^{k\ell}$ to denote the rejection threshold from procedure (\ref{eq_rej_procedure}) with pairwise FDR level $\widehat q_m$, the forward probability and the backward probability. Denote $\widehat R^{k\ell}(\widehat q_m)$ as the number of rejections for pair $k<\ell$ with pairwise FDR $\widehat q_m$ and $\widehat R(\widehat q_m)$ as the number of rejections of the following eBH procedure.
        Note that 
    \begin{align}
        \widehat e_{j}^{k\ell}(\widehat q_m)=& 
        \frac{I(\widehat T_j^{k\ell}\leq \widehat \lambda_{\rm rLIS}^{k\ell}(\widehat q_m))}{m^{-1}\sum_{j'=1}^m I(\widehat T_{j'}^{k\ell}\leq \widehat \lambda_{\rm rLIS}^{k\ell}(\widehat q_m))\widehat T_{j'}^{k\ell}}. 
        \label{eq_eval_supp}
    \end{align}

    By the step-up procedure, the non-zero e-values satisfy
    \begin{align}
         \label{eq_qm_solution}
        \frac{m/\widehat R^{k\ell}(\widehat q_m)}{(1/\widehat R^{k\ell}(\widehat q_m))\sum_{j=1}^{\widehat R^{k\ell}(\widehat q_m)}\widehat T_{(j)}^{k\ell} } \geq \frac{m/\widehat R^{k\ell}(\widehat q_m)}{\widehat q_m}.
    \end{align}
    Note that $\widehat q_m$ is the largest value satisfying
    \begin{align}
        \label{eq_choice_of_qm}
        \min_{k<\ell}\left\{\frac{(\widehat\pi_0^{k\ell} + \widehat\pi_1^{k\ell} + \widehat\pi_2^{k\ell})\widehat R(\widehat q_m)}{\widehat R^{k\ell}(\widehat q_m)}\right\}
        \geq\frac{\widehat q_m}{q},
    \end{align}
    which means
    \begin{align*}
        \frac{m/\widehat R(\widehat q_m)}{q} 
        \leq& \min_{k<\ell}\left\{\frac{(\widehat\pi_0^{k\ell} + \widehat\pi_1^{k\ell} + \widehat\pi_2^{k\ell})m/\widehat R^{k\ell}(\widehat q_m)}{\widehat q_m}\right\}
        \\
        \leq& \min_{k<\ell}\left\{\frac{(\widehat\pi_0^{k\ell} + \widehat\pi_1^{k\ell} + \widehat\pi_2^{k\ell})m/\widehat R^{k\ell}(\widehat q_m)}{(1/\widehat R^{k\ell}(\widehat q_m))\sum_{j=1}^{\widehat R^{k\ell}(\widehat q_m)}\widehat T_{(j)}^{k\ell} }\right\}.
    \end{align*}
    Thus $m/(\widehat R(\widehat q_m)\cdot q) \leq \min_{k<\ell}\{(\widehat\pi_0^{k\ell} + \widehat\pi_1^{k\ell} + \widehat\pi_2^{k\ell})e_{(1)}^{k\ell}(\widehat q_m)\} = \widehat e_{(1)}(\widehat q_m)$. By the eBH procedure, we have $\widehat R(\widehat q_m) > 0$. If no solution exists, then for any $\widehat q_m\in[q_-, q]$, the eBH procedure returns an empty discovery set. In this case, let $\widehat q_m = q$.

    By replacing $q$ by $q_-$ in (C2), since $\widehat q_m \geq q_-$, Lemma \ref{lemma_consistency_of_thresholds} and Lemma \ref{lemma_R_to_infty} show that $\widehat R^{k\ell}(\widehat q_m)/m$ is lower bounded by some constant. Thus following the derivation of (\ref{eq_average_of_T_hat_as}), we have $|1/\widehat R^{k\ell}(\widehat q_m)\cdot\sum_{j=1}^{\widehat R^{k\ell}(\widehat q_m)} \widehat T_{(j)}^{k\ell} - \widehat q_m| \to 0$ almost surely as $m\to\infty$. Moreover, by Theorem \ref{thm_consistency}, $\widehat\pi_0^{k\ell},\widehat\pi_1^{k\ell} $and $\widehat\pi_2^{k\ell}$ are consistent estimators of $ \pi_0^{k\ell}, \pi_1^{k\ell}$ and $\pi_2^{k\ell}$, respectively, and their summation is bounded by $1$.  Thus, the difference between the LHS and RHS of (\ref{eq_qm_solution}) satisfies
    \begin{align*}
        \left|\frac{(\widehat\pi_0^{k\ell} + \widehat\pi_1^{k\ell} + \widehat\pi_2^{k\ell})m/\widehat R^{k\ell}(\widehat q_m)}{(1/\widehat R^{k\ell}(\widehat q_m))\sum_{j=1}^{\widehat R^{k\ell}(\widehat q_m)}\widehat T_{(j)}^{k\ell} } - \frac{(\widehat\pi_0^{k\ell} + \widehat\pi_1^{k\ell} + \widehat\pi_2^{k\ell})m/\widehat R^{k\ell}(\widehat q_m)}{\widehat q_m}\right| \to 0 \text{ as } m\to\infty.
    \end{align*}
    Then $\widehat R^{k\ell}(\widehat q_m)$ is maximized when $\widehat R(\widehat q_m)$ is positive. Note that $\widehat R^{k\ell}(\widehat q_m)$ and $\widehat R(\widehat q_m)$ increase as $\widehat q_m$ increases. Therefore, $\widehat R(\widehat q_m)$ is maximized by choosing $\widehat q_m$ as the largest value satisfying (\ref{eq_choice_of_qm}).
    
    By Lemma \ref{lemma_consistency_when_lambda_less_than_alpha}, for any $t\in(q_-, q)$, $\mathbb E|\widehat R^{k\ell}(t)/R^{k\ell}(t) - 1|\to 0$ and $\mathbb E|\widehat R(t)/R(t) - 1|\to0$ as $m\to\infty$. Thus for any $1\leq k < \ell \leq n$, we have $\widehat R^{k\ell}(t) = R^{k\ell}(t) (1 + \varepsilon^{k\ell})$, where $\mathbb E(\varepsilon^{k\ell}) \to 0$ as $m\to\infty$. Then for any pair $1 \leq k < \ell \leq n$, we have $\widehat R^{k\ell}(t) \leq \max_{k<\ell}\{R^{k\ell}(t)\}(1 + \varepsilon^{k\ell})$, which means that $\max_{k<\ell}\widehat R^{k\ell}(t) \leq \max_{k<\ell}\{R^{k\ell}(t)\}(1 + \max_{k<\ell}\varepsilon^{k\ell})$. Since $n$ is fixed, we have $\mathbb E|\max_{k<\ell}\varepsilon^{k\ell}| \to 0$ as $m\to\infty$. Additionally, we also have $\max_{k<\ell}\widehat R^{k\ell}(t) \geq \max_{k<\ell}\{R^{k\ell}(t)\}(1 + \min_{k<\ell}\varepsilon^{k\ell})$ with $\mathbb E|\min_{k<\ell}\varepsilon^{k\ell}| \to 0$ as $m\to\infty$.
    Therefore, we have 
    \begin{align*}
        \mathbb E\left|\frac{\max_{k<\ell}\widehat R^{k\ell}(t)}{\max_{k<\ell} R^{k\ell}(t)} - 1\right| \to 0 \text{ as } m\to\infty.
    \end{align*}
    Note that
    \begin{align*}
        &\mathbb E\left|\frac{\widehat{R}(t)}{\max_{k<\ell}\widehat R^{k\ell}(t)} - \frac{R(t)}{\max_{k<\ell}R^{k\ell}(t)}\right| 
        \\
        \leq & \mathbb E\left|\frac{\widehat{R}(t)}{\max_{k<\ell}\widehat R^{k\ell}(t)} - \frac{\widehat R(t)}{\max_{k<\ell} R^{k\ell}(t)}\right| +\mathbb E \left|\frac{\widehat R(t)}{\max_{k<\ell} R^{k\ell}(t)} - \frac{R(t)}{\max_{k<\ell}R^{k\ell}(t)}\right|
        \\
        =& \mathbb E\left\{\frac{\widehat R(t)}{\max_{k<\ell}\widehat R^{k\ell}(t)}\cdot\left|1 - \frac{\max_{k<\ell}\widehat R^{k\ell}(t)}{\max_{k<\ell} R^{k\ell}(t)}\right|\right\} + \mathbb E\left\{\frac{ R(t)}{\max_{k<\ell}R^{k\ell}(t)}\cdot\left|\frac{\widehat R(t)}{R(t)} - 1\right|\right\}
        \\
        \to& 0 \quad \text{ as } m\to \infty.
    \end{align*}
    Then uniformly over $t\in[q_-, q]$, $\widehat R(t)/\max_{k<\ell}\widehat R^{k\ell}(t)$ converges to the almost-everywhere continuous function $R(t)/\max_{k<\ell}R^{k\ell}(t)$ almost surely as $m\to\infty$. Furthermore, we also have $\min_{k<\ell}\{(\widehat\pi_0^{k\ell} + \widehat\pi_1^{k\ell} + \widehat\pi_2^{k\ell})\widehat R(t)/\widehat R^{k\ell}(t)\}$ converges to almost-everywhere continuous function $\min_{k<\ell}\{(\pi_0^{k\ell} + \pi_1^{k\ell} + \pi_2^{k\ell})R(t)/R^{k\ell}(t)\}$ uniformly over $t\in[q_-, q]$ almost surely as $m\to\infty$.
    Let $\widehat f_m(t) =\min_{k<\ell}\{(\widehat\pi_0^{k\ell} + \widehat\pi_1^{k\ell} + \widehat\pi_2^{k\ell})\widehat R(t)/\widehat R^{k\ell}(t)\} - t / q$ and $f_*(t) = \mathbb E[\min_{k<\ell}\{(\pi_0^{k\ell} + \pi_1^{k\ell} + \pi_2^{k\ell})R(t)/R^{k\ell}(t)\}] - t / q$. Then $|\widehat f_m(t) - f_*(t)| \to 0$ almost surely uniformly over $t\in[q_-, q]$. Denote $\widehat q_m$ as the largest value in $[q_-, q]$ satisfying $\widehat f_m(t) \geq 0$ and $q_*$ as the largest value in $[q_-, q]$ satisfying $f_*(t)\geq 0$. We first consider the case that $q_*\in [q_-, q]$ exists. The uniform convergence implies $\mathbb E|\widehat f_m(q_*) - f_*(q_*)|\to 0$, which means $\widehat f_m(\widehat q_m) \geq 0$ and thus $\widehat q_m \geq q_*$ almost surely. If $q_* = q$, then $\widehat q_m\to q$ as $m\to\infty.$ If $q_*<q$, then $f_*(q_*) \geq 0$ and for any $\epsilon \in (0, q-q_*]$, $f_*(q_*+\epsilon) < 0$. The uniform convergence implies $\mathbb E|\widehat f_m(q_*+\epsilon) - f_*(q_*+\epsilon)|\to 0$, which means $\widehat f_m(q_*+\epsilon) < 0$ for any $\epsilon\in(0, q-q_*]$ and thus $\widehat q_m \leq q_*$ almost surely. Thus, we know that $\widehat q_m \to q_*$ almost surely. 

    Next, we consider the case that $f_*(t)<0$ for any $t\in[q_-, q]$, then $q_*$ does not exist. In this case $\mathbb E|\widehat f_m(t) - f_*(t)| \to 0$ as $m\to\infty$, which means that when $m$ is large enough, (\ref{eq_choice_of_qm}) is not satisfied for any choice of $\widehat q_m$ and there is no rejection as discussed above.
    
     \paragraph{Step 2.  Consistency of pairwise rLIS quantities.} By Lemma \ref{lemma_consistency_of_thresholds}, $\widehat\lambda_{\rm rLIS}^{k\ell}(q_*) \to \lambda_{\rm OR}^{\infty,k\ell}(q_*)$ in probability as $m\to \infty$. By the ergodic stationary distribution of $T_j^{\infty, k\ell}$, we have $\widehat\alpha_j^{k\ell}(s_j)\widehat\beta_j^{k\ell}(s_j) \to \alpha_j^{k\ell}(s_j)\beta_j^{k\ell}(s_j)$ almost surely for $s_j = 0,1,2,3$ as $m\to\infty$. By the continuous mapping theorem, we have $\widehat T_j^{k\ell}\to T_j^{\infty,k\ell}$ almost surely. Noting that $T_j^{\infty,k\ell}$ has positive and continuous density function on its support $(0, \alpha_*^{k\ell})$, we have $\mathbb P(T_j^{\infty,k\ell} = \lambda_{\rm OR}^{\infty,k\ell}(q_*)) = 0.$ Therefore, we have $I(\widehat T_j^{k\ell}\leq \widehat \lambda_{\rm rLIS}^{k\ell}(q_*)) - I(T_j^{\infty, k\ell}\leq \lambda_{\rm OR}^{\infty,k\ell}(q_*)) \to 0$ in probability and $I(\widehat T_j^{k\ell}\leq \widehat \lambda_{\rm rLIS}^{k\ell}(q_*))\widehat T_j^{k\ell} - I(T_j^{\infty, k\ell}\leq \lambda_{\rm OR}^{\infty,k\ell}(q_*))T_j^{\infty, k\ell}\to 0$ in probability. 
    Considering the composite global null space $\mathcal H_0 = \{j: \prod_{i=1}^n \theta_{ij} = 0\}$ and the pairwise null space $\mathcal H_0^{k\ell} = \{j: \theta_{kj}\theta_{\ell j} = 0\}$, we have
    \begin{align}
        &\mathbb E\left\{I(\widehat T_j^{k\ell}\leq \widehat \lambda_{\rm rLIS}^{k\ell}(q_*))\mid j\in\mathcal H_0^{k\ell}\right\} \notag
        \\
        \to&\mathbb E\left\{I(T_j^{\infty,k\ell}\leq \lambda_{\rm OR}^{\infty,k\ell}(q_*))\mid j\in\mathcal H_0^{k\ell}\right\} \notag
        \\
        =& \mathbb P(j\in\mathcal H_0^{k\ell})^{-1}\int I(T_{j}^{\infty, k\ell}\leq \lambda_{\rm OR}^{\infty, k\ell}(q_*))\cdot \left\{\sum_{s_j=0}^2\alpha_j^{k\ell}(s_j)\beta_j^{k\ell}(s_j)\right\} \prod_{j'=1}^m({\rm d}y_{kj'}{\rm d}y_{\ell j'})%\quad\text{ almost surely.} 
        \label{eq_evalue_numerator_limit}
    \end{align}
    almost surely. Additionally, Birkhoff's ergodic theorem \citep{birkhoff1931proof} shows that the denominator in (\ref{eq_eval_supp}) satisfies
    \begin{align}
        &m^{-1}\sum_{j'=1}^m I(\widehat T_{j'}^{k\ell}\leq \widehat \lambda_{\rm rLIS}^{k\ell}(q_*))\widehat T_{j'}^{k\ell}\notag
        \\
        \to& \mathbb E\{I(T_j^{\infty,k\ell} \leq \lambda_{OR}^{\infty,k\ell}(q_*))T_j^{\infty,k\ell}\}\notag
        \\
        =&\int I(T_{j}^{\infty,k\ell}\leq\lambda_{\rm OR}^{\infty,k\ell}(q_*))\cdot\frac{\sum_{s_j=0}^2\alpha_j^{k\ell}(s_j)\beta_j^{k\ell}(s_j)}{\sum_{s_j=0}^3\alpha_j^{k\ell}(s_j)\beta_j^{k\ell}(s_j)}\left\{\sum_{s_j=0}^3\alpha_j^{k\ell}(s_j)\beta_j^{k\ell}(s_j)\right\} \prod_{j'=1}^m({\rm d}y_{kj'}{\rm d}y_{\ell j'})\notag
        \\
        =& \int I(T_{j}^{\infty, k\ell}\leq \lambda_{\rm OR}^{\infty,k\ell}(q_*))\cdot \left\{\sum_{s_j=0}^2\alpha_j^{k\ell}(s_j)\beta_j^{k\ell}(s_j)\right\} \prod_{j'=1}^m({\rm d}y_{kj'}{\rm d}y_{\ell j'})
        \quad\text{ almost surely.}\label{eq_evalue_denominator_limit}
    \end{align}
    Furthermore, by (\ref{eq_average_of_T_hat_as}), $1/\widehat R^{k\ell(q_*)}\sum_{j=1}^{\widehat R^{k\ell}(q_*)}\widehat T_{(j)}^{k\ell}\to q_*$ almost surely and by Lemma \ref{lemma_R_to_infty}, $\widehat R^{k\ell}(q_*)/m\geq G^{\infty}(q_*/2) > 0$ almost surely. Noting that $\sum_{j=1}^{\widehat R^{k\ell}(q_*)}\widehat T_{(j)}^{k\ell} = \sum_{j=1}^m I(\widehat T_j^{k\ell}\leq \widehat\lambda_{\rm rLIS}^{k\ell}(q_*))\widehat T_j^{k\ell}$, we have
    \begin{align*}
        &m^{-1}\sum_{j=1}^m I(\widehat T_j^{k\ell}\leq \widehat\lambda_{\rm rLIS}^{k\ell}(q_*))\widehat T_j^{k\ell}\\
        =& \frac{\widehat R^{k\ell}(q_*)}{m}\cdot \frac{1}{\widehat R^{k\ell}(q_*)}\sum_{j=1}^{\widehat R^{k\ell}(q_*)}\widehat T_{(j)}^{k\ell}\\
        \geq& G^\infty(q_*/2)\cdot q_* \text{ almost surely.} 
    \end{align*}
    This implies that $\{\widehat e_j^{k\ell}(q_*)\}_{j=1}^m$ is upper bounded by $1/ \{G^\infty(q_*/2)q_*\}$ almost surely, and then it is uniformly integrable. Therefore, combining (\ref{eq_evalue_numerator_limit}) and (\ref{eq_evalue_denominator_limit}), we have 
    \begin{align}
        \mathbb E[\widehat e_j^{k\ell}(q_*)\mid j\in\mathcal H_0^{k\ell}] \to 1/\mathbb P(j\in\mathcal H_0^{k\ell}) \text{ as }m\to\infty.
        \label{eq_conditional_E_evalue}
    \end{align} 
    
    \paragraph{Step 3. The constructed e-values satisfy the conditions of the eBH procedure.} Next,
    we show that $m^{-1}\sum_{j\in\mathcal H_0}\mathbb E[\widehat e_j] \leq 1$ almost surely.
    For each $ j \in \mathcal H_0 $, denote 
\begin{align*}
\mathcal R(j) = \{ (k, \ell) : k < \ell, j\in\mathcal H_0^{k\ell}\}.
\end{align*}
Choose a unique pair $(k_j, \ell_j)$ as follows: Randomly choose a pair $(k, \ell)\in\mathcal R(j)$ with probability $1/|\mathcal R(j)|$.
% Set $(k_j, \ell_j) = (k, \ell)$ with a success probability
% \begin{align*}
% \mathbb P((k_j, \ell_j) = (k, \ell)) = \mathbb P(j\in\mathcal H_0^{k\ell}\mid j\in\mathcal H_0).
% \end{align*}
% If it fails, repeat the process of choosing $(k,\ell)$ and setting $(k, \ell)$ until success.

For any $ k < \ell $, define 
\begin{align*}
\mathcal A_{k\ell} = \{ j \in \mathcal H_0 : (k_j, \ell_j) = (k, \ell) \}.
\end{align*}
Thus, $ \{ \mathcal A_{k\ell} \}_{k < \ell} $ is a partition of $\mathcal H_0$.
Furthermore, for any $(\omega_1,\omega_2)\in\{(0,0), (0,1),(1,0)\}$, we have
\begin{align}
&\mathbb P((\theta_{kj}, \theta_{\ell j}) = (\omega_1, \omega_2))\mid j\in \mathcal A_{k\ell})\notag
\\
=&\frac{\mathbb P((\theta_{kj}, \theta_{\ell j}) = (\omega_1,\omega_2), j\in\mathcal A_{k\ell})}{\mathbb P(j\in \mathcal A_{k\ell})}.\label{eq_theta_Akl}
\end{align}
The numerator in (\ref{eq_theta_Akl}) is 
\begin{align*}
    &\mathbb P((\theta_{kj}, \theta_{\ell j}) = (\omega_1,\omega_2), j\in\mathcal A_{k\ell}) \\
    =& \mathbb P(j\in\mathcal A_{k\ell}\mid (\theta_{kj}, \theta_{\ell j}) = (\omega_1,\omega_2))\cdot \mathbb P((\theta_{kj}, \theta_{\ell j}) = (\omega_1,\omega_2))
    \\
    =& \mathbb P(j\in\mathcal A_{k\ell}\mid j\in\mathcal H_0^{k\ell})\cdot \mathbb P((\theta_{kj}, \theta_{\ell j}) = (\omega_1,\omega_2)).
\end{align*}
The denominator in (\ref{eq_theta_Akl}) is 
\begin{align*}
    \mathbb P(j\in \mathcal A_{k\ell}) =& \mathbb P(j\in \mathcal A_{k\ell}\mid j\in\mathcal H_0^{k\ell}) \cdot\mathbb P(j\in\mathcal H_0^{k\ell}).
\end{align*}
Therefore, we have
\begin{align*}
    \mathbb P((\theta_{kj}, \theta_{\ell j}) = (\omega_1,\omega_2)\mid j\in\mathcal A_{k\ell}) =& \frac{\mathbb P((\theta_{kj}, \theta_{\ell j}) = (\omega_1,\omega_2))}{\mathbb P(j\in\mathcal H_0^{k\ell})} 
    \\
    =& \mathbb P((\theta_{kj}, \theta_{\ell j}) = (\omega_1,\omega_2)\mid j\in\mathcal H_0^{k\ell}),
\end{align*}
which means that $(\theta_{kj}, \theta_{\ell j})$ has the same conditional distribution for $j\in\mathcal A_{k\ell}$ and $j\in\mathcal H_0^{k\ell}$. By Birkhoff's ergodic theorem \citep{birkhoff1931proof}, we have
\begin{align*}
&\left| \frac{1}{|\mathcal A_{k\ell}|} \sum_{j \in \mathcal A_{k\ell}} \mathbb E(\widehat e_j^{k\ell}(q_*)) - \mathbb E(\widehat e_j^{k \ell}(q_*) \mid j \in \mathcal H_0^{k \ell}) \right| \to 0 \text{ almost surely as } m\to\infty.
\end{align*}
Thus by (\ref{eq_conditional_E_evalue}), we have $|\mathcal A_{k\ell}|^{-1} \sum_{j \in \mathcal A_{k\ell}} \mathbb E(\widehat e_j^{k\ell}(q_*)) \to 1/\mathbb P(j\in\mathcal H_0^{k\ell})$ almost surely. Noting that $\mathcal A_{k\ell} $ is a partition of $\mathcal H_0$, we have $\sum_{k<\ell}|A_{k\ell}|/m \to \mathbb P(j\in\mathcal H_0)$ almost surely. Moreover, $\widehat\pi_0^{k\ell}, \widehat \pi_1^{k\ell}$ and $\widehat\pi_2^{k\ell}$ are consistent estimators of $\pi_0^{k\ell}, \pi_1^{k\ell}$ and $\pi_2^{k\ell}$, respectively, and their summation is bounded by $1$. By Lemma \ref{lemma_consistency_of_thresholds} and Lemma \ref{lemma_R_to_infty}, the non-zero e-values $\widehat e_j^{k\ell}(q_*)$ are upper bounded by some constant. Therefore, by the dominated convergence theorem,
\begin{align*}
\frac{1}{m} \sum_{j \in \mathcal H_0} \mathbb E(\widehat e_j(q_*)) 
=& \frac{1}{m} \sum_{k < \ell} |\mathcal A_{k\ell}| \left( \frac{1}{|\mathcal A_{k\ell}|} \sum_{j \in \mathcal A_{k\ell}} \mathbb E(\widehat e_j(q_*))\right)
\\
\leq &\frac{1}{m} \sum_{k < \ell} |\mathcal A_{k\ell}| \left( \frac{1}{|\mathcal A_{k\ell}|} \sum_{j \in \mathcal A_{k\ell}} \mathbb E((\widehat\pi^{k\ell}_0+\widehat\pi^{k\ell}_1+\widehat\pi^{k\ell}_2)\widehat e_j^{k \ell}(q_*)) \right)
\\
\leq &\frac{1}{m} \sum_{k < \ell} |\mathcal A_{k\ell}|\cdot 1
\\
= & \frac{|\mathcal H_0|}{m}\leq 1, \quad \text{ as } m \to \infty,
\end{align*}
which is the asymptotic version of the condition in Theorem 2 of \cite{wang2022false}.
    
    Then eBH procedure based on $(\widehat e_j(q_*))_{j=1}^m$ controls the FDR by the asymptotic version of Theorem 2 in \cite{wang2022false}.
\end{proof}

\section{Proof of lemmas}
\subsection{Proof of Lemma \ref{lemma_consistency_of_thresholds}}
\begin{proof}
    Recall that
        \begin{align*}
            \widehat{\lambda}_{\mathrm{OR}}^{\infty} =& \sup\left\{t: \widehat{Q}_{\mathrm{OR}}^{\infty}(t) \leq q\right\},\\
            \lambda_{\mathrm{OR}}^{\infty} =& \sup\left\{t: Q_{\mathrm{OR}}^{\infty}(t) \leq q\right\},
        \end{align*}
        where
        \begin{align}
            \widehat{Q}_{\mathrm{OR}}^{\infty}(t) =& \frac{\sum_{j=1}^m I(T_j^{\infty} \leq t) T_j^{\infty}}{\sum_{j=1}^m I(T_j^{\infty} \leq t)},\label{eq_Q_hat_OR}\\
            Q_{\mathrm{OR}}^{\infty}(t) =& \frac{\pi_0 G_0^{\infty}(t) + \pi_1 G_1^{\infty}(t) + \pi_2 G_2^{\infty}(t)}{G^{\infty}(t)}.\notag
        \end{align}
        Since $T_1^\infty = \mathbb P(s_1\in\{0, 1, 2\}\mid (y_{1j}, y_{2j})_{j=-\infty}^\infty)$, it is a function of $(y_{1j}, y_{2j})_{j=-\infty}^\infty$. Thus
        \begin{align*}
            \mathbb E\left\{I(T_1^{\infty} \leq t)T_1^{\infty}\right\} =& \mathbb E\left\{I(T_1^{\infty} \leq t) \mathbb E\left[I(s_1\in\{0,1,2\})\mid (y_{1j}, y_{2j})_{j=-\infty}^{\infty}\right]\right\}\\
            =& \mathbb E\left\{\mathbb E[I(T_1^{\infty} \leq t, s_1\in\{0, 1, 2\})\mid (y_{1j}, y_{2j})_{j=-\infty}^{\infty}]\right\}\\
            =& \mathbb P\left(T_1^{\infty} \leq t, s_1 \in \{0, 1, 2\}\right)\\
            =& \pi_0 G_0^{\infty}(t) + \pi_1 G_1^{\infty}(t) + \pi_2 G_2^{\infty}(t),\text{ and }\\
            \mathbb E\left\{I(T_1^{\infty} \leq t)\right\} =& G^{\infty}(t).
        \end{align*}
        Birkhoff's ergodic theorem \citep{birkhoff1931proof} gives 
        \begin{align}
            \frac{1}{m}\sum_{j=1}^m I(T_j^{\infty} \leq t) T_j^{\infty}\rightarrow& \pi_0 G_0^{\infty}(t) + \pi_1 G_1^{\infty}(t) + \pi_2 G_2^{\infty}(t) \text{ almost surely}\quad (0\leq t \leq 1),\notag\\
            \frac{1}{m}\sum_{j=1}^m I(T_j^{\infty} \leq t) \rightarrow & G^{\infty}(t) \text{ almost surely}\quad (0\leq t \leq 1).\notag
       \end{align}
        Consequently, 
        \begin{align}
            \widehat{Q}_{\mathrm{OR}}^{\infty}(t) \rightarrow& Q_{\mathrm{OR}}^{\infty}(t) \text{ almost surely for any }t \text{ such that } G^\infty(t)>0.\label{eq_convergence_of_Q_hat}
        \end{align}
        In addition, $G^\infty(\lambda^\infty_{\rm OR}) > 0$. Therefore,
        \begin{align*}
            \mathbb P\bigg(\lim_{m\rightarrow\infty}\widehat{Q}_{\mathrm{OR}}^{\infty}(\lambda_{\mathrm{OR}}^{\infty}) \leq q\bigg)=& \mathbb P\bigg(\lim_{m\rightarrow\infty}\big(\widehat{Q}_{\mathrm{OR}}^{\infty}(\lambda_{\mathrm{OR}}^{\infty}) - Q_{\mathrm{OR}}^{\infty}(\lambda_{\mathrm{OR}}^{\infty})\big) + Q_{\mathrm{OR}}^{\infty}(\lambda_{\mathrm{OR}}^{\infty}) \leq q \bigg) = 1,
        \end{align*}
        which implies that $\mathbb P(\lim_{m\rightarrow\infty}\widehat{\lambda}_{\mathrm{OR}}^{\infty} \geq \lambda_{\mathrm{OR}}^{\infty}) = 1$, or equivalently, 
        \begin{align}
            \widehat{\lambda}_{\mathrm{OR}}^{\infty} \geq \lambda_{\mathrm{OR}}^{\infty}\text{ almost surely.}
            \label{eq_lambda_hat_greater_than_lambda_OR}
        \end{align}

        By  construction, $\widehat{Q}_{\mathrm{OR}}^{\infty}(t)$ is an increasing step function with jump at $T_{(j)}^{\infty}$. For $T_{(j)}^{\infty} \leq t < T_{(j+1)}^{\infty}$, construct the lower bound of $\widehat Q_{\rm OR}^\infty(t)$ as
        \begin{align*}
            \widehat{L}_{\mathrm{OR}}^{\infty}(t) =& \frac{T_{(j+1)}^{\infty} - t}{T_{(j+1)}^{\infty} - T_{(j)}^{\infty}} \widehat{Q}_{\mathrm{OR}}^{\infty} (T_{(j-1)}^{\infty}) + \frac{t - T_{(j)}^{\infty}}{T_{(j+1)}^{\infty} - T_{(j)}^{\infty}} \widehat{Q}_{\mathrm{OR}}^{\infty} (T_{(j)}^{\infty}).
        \end{align*}
        Then $\widehat{L}_{\mathrm{OR}}^{\infty}(t)$ is strictly increasing in $t$. We also have
        \begin{align*}
            0 \leq \widehat{Q}_{\mathrm{OR}}^{\infty}(t) - \widehat{L}_{\mathrm{OR}}^{\infty}(t) \leq& \widehat{Q}_{\mathrm{OR}}^{\infty}(T_{(j)}^{\infty}) - \widehat{Q}_{\mathrm{OR}}^{\infty}(T_{(j-1)}^{\infty})\\
            =& \frac{(j-1)T_{(j)}^{\infty} - \sum_{k=1}^{j-1}T_{(k)}^{\infty}}{j(j-1)}\\
            \leq & \frac{1}{j}\\
            =& \frac{1}{R^{\infty}(t)},
        \end{align*}
        where $R^\infty (t) = \sum_{k=1}^m 1(T_k^\infty \leq t)$ denotes the number of rejections yielded by threshold $t$, satisfying $R^\infty (t) = j$ if $T_{(j)}^{\infty} \leq t < T_{(j+1)}^{\infty}$.
        By Birkhoff's ergodic theorem \citep{birkhoff1931proof}, $R^\infty(t) / m \to G^\infty(t)$ almost surely as $m \to \infty$. Then we have
        \begin{align*}
            \widehat{Q}_{\mathrm{OR}}^{\infty}(t) - \widehat{L}_{\mathrm{OR}}^{\infty}(t) \rightarrow 0\text{ almost surely}\quad (0\leq t \leq 1).
        \end{align*}
        By (\ref{eq_convergence_of_Q_hat}), $\widehat{L}_{\mathrm{OR}}^{\infty}(t) \rightarrow Q_{\mathrm{OR}}^{\infty}(t)$ almost surely for $0\leq t\leq 1$.

        Denote
        \begin{align*}
            \widehat{\lambda}_{\mathrm{L,OR}}^{\infty} = \sup\{t\in (0, 1): \widehat{L}_{\mathrm{OR}}^{\infty}(t) \leq q\}.
        \end{align*}
        As $\widehat{Q}_{\mathrm{OR}}^{\infty}(t) \geq \widehat{L}_{\mathrm{OR}}^{\infty}(t)$ with probability $1$, we have 
        \begin{align}
            \widehat{\lambda}_{\mathrm{OR}}^{\infty} \leq \widehat{\lambda}_{\mathrm{L,OR}}^{\infty} \text{ with probability }1.
            \label{eq_lambda_OR_less_than_lambda_L_OR}
        \end{align}
        By (\ref{eq_lambda_hat_greater_than_lambda_OR}), we also have $\widehat\lambda_{\rm L,OR}^\infty \geq \lambda_{\rm OR}^\infty$ almost surely.

        We claim that $\widehat{\lambda}_{\mathrm{L,OR}}^{\infty} \rightarrow \lambda_{\mathrm{OR}}^{\infty}$ in probability. If not, there exist $\varepsilon_2 > 0$ and $\eta_0 >0$ such that for any $M>0$, there exists $m_1 \geq M$ satisfying
        \begin{align*}
            \mathbb P(K_{m_1}^1) \geq 2\eta_0,
        \end{align*}
        where $K_{m_1}^1$ denotes the event that $\widehat{\lambda}_{\mathrm{L,OR}}^{\infty} - \lambda_{\mathrm{OR}}^{\infty} > \varepsilon_2$. 

        Let \begin{align*}
            2\delta_1 = Q_{\mathrm{OR}}^{\infty}(\lambda_{\mathrm{OR}}^{\infty} + \varepsilon_2) - q > 0.
        \end{align*}
        Since $\widehat{L}_{\mathrm{OR}}^{\infty}(t) \rightarrow Q_{\mathrm{OR}}^{\infty}(t)$ in probability for any $t\in[0, 1]$, there exists $M_2>0$, such that for any $m_2 \geq M_2$,
        \begin{align*}
            \mathbb P(K_{m_2}^2) \geq 1-\eta_0,
        \end{align*}
        where $K_{m_2}^2$ denotes the event that 
        $|\widehat{L}_{\mathrm{OR}}^{\infty}(\lambda_{\mathrm{OR}}^{\infty} + \varepsilon_2) - Q_{\mathrm{OR}}^{\infty}(\lambda_{\mathrm{OR}}^{\infty} + \varepsilon_2)| < \delta_1$.

        Without loss of generality, assume $m_1 = m_2 = m$. Letting $K_m = K_m^1 \bigcap K_m^2$, we have
        \begin{align*}
            \mathbb P(K_m) =& 1- \mathbb P((K_m^1)^c \cup (K_m^2)^c)\\
            \geq & 1- \{(1-2\eta_0) + \eta_0\}\\
            =& \eta_0.
        \end{align*}
        Thus $K_m$ has positive probability.

        Additionally, $\widehat{L}_{\mathrm{OR}}^{\infty}(t)$ is strictly increasing  over $t$ with probability $1$. On $K_m$, we have
        \begin{align*}
            q =& \widehat{L}_{\mathrm{OR}}^{\infty}(\widehat{\lambda}_{\mathrm{L,OR}}^{\infty})\\
            >& \widehat{L}_{\mathrm{OR}}^{\infty}(\lambda_{\mathrm{OR}}^{\infty} + \varepsilon_2)\\
            >& Q_{\mathrm{OR}}^{\infty} (\lambda_{\mathrm{OR}}^{\infty} + \varepsilon_2) - \delta_1\\
            =& q + \delta_1,
        \end{align*}
        which is a contradiction. Thus we must have $\widehat{\lambda}_{\mathrm{L,OR}}^{\infty} \rightarrow \lambda_{\mathrm{OR}}^{\infty}$ in probability. Furthermore, by (\ref{eq_lambda_hat_greater_than_lambda_OR}) and (\ref{eq_lambda_OR_less_than_lambda_L_OR}),
        \begin{align*}
            \widehat{\lambda}_{\mathrm{OR}}^{\infty} \rightarrow \lambda_{\mathrm{OR}}^{\infty}\text{ in probability as } m\rightarrow\infty.
        \end{align*}
        $\widehat \lambda_{\rm rLIS}^\infty \to \lambda_{\rm OR}^\infty$ in probability can be shown in the same way and we omit the details.
\end{proof}

\subsection{Proof of Lemma \ref{lemma_R_to_infty}}
\begin{proof}
    Define $M_d^+(j,\phi) = \max_{k,l=0,1,2,3}\mathbb P_{\phi}(S_j = k\mid y_1^m, S_{j-d} = l)$. Similarly, define $M_d^-(j,\phi) = \min_{k,l=0,1,2,3}\mathbb P_{\phi}(S_j = k\mid y_1^m, S_{j-d} = l)$. We first show that 
    \begin{align}
        \label{eq_M_d_plus_M_d_minus_upper_bound}
        |M_d^+(j,\phi) - M_d^-(j,\phi)| \leq \prod_{i=j-d+1}^{j-1} \{1-2\tau_0(y_{1i}, y_{2i})\}.
    \end{align}

    Since $\varepsilon_0 \leq 1/4$ and $\rho_0(y_{1i}, y_{2i}) \geq 1$ with probability $1$ by definition, we have $\tau_0(y_{1i}, y_{2i}) = \{1 + \varepsilon_0^{-2}\rho_0(y_{1i}, y_{2i})\}^{-1} \leq 1/17$ with probability $1$. Thus $1-2\tau_0(y_{1i}, y_{2i}) > 0$ with probability $1$ for any $i = 1,\ldots,m.$
    We have
    \begin{align*}
        \mathbb P_{\phi}(S_j = k\mid y_1^m, S_{j-d} = l) =& \sum_{k'=0}^3 \mathbb P_{\phi}(S_j = k, S_{j-d+1} = k'\mid y_1^m, S_{j-d} = l)\\
        =& \sum_{k'=0}^3 \mathbb P_{\phi}(S_j = k\mid y_1^m, S_{j-d+1} = k')\mathbb P_{\phi}(S_{j-d+1} = k'\mid y_1^m, S_{j-d} = l).
    \end{align*}
    Since $\mathbb P_{\phi}(S_{j-d+1}=k'\mid y_1^m, S_{j-d}=l) \geq \tau_0(y_{1,j-d+1}, y_{2,j-d+1})$, we have
    \begin{align*}
        M_d^+(j, \phi) \leq& \{1-\tau_0(y_{1,j-d+1}, y_{2,j-d+1})\} M_{d-1}^+(j, \phi) + \tau_0(y_{1,j-d+1}, y_{2,j-d+1}) M_{d-1}^-(j, \phi),
    \end{align*}
    and similarly,
    \begin{align*}
        M_d^-(j, \phi) \geq& \{1-\tau_0(y_{1,j-d+1}, y_{2,j-d+1})\} M_{d-1}^-(j, \phi) + \tau_0(y_{1,j-d+1}, y_{2,j-d+1}) M_{d-1}^+(j, \phi).
    \end{align*}
    Therefore,
    \begin{align*}
        M_d^+(j, \phi) - M_d^-(j, \phi) \leq& \{1-2\tau_0(y_{1,j-d+1}, y_{2,j-d+1})\} \{M_{d-1}^+(j, \phi) - M_{d-1}^-(j, \phi)\}\\
        \leq& \prod_{i=j-d+1}^{j-1} \{1-2\tau_0(y_{1i}, y_{2i})\} \{M_1^+(j,\phi) - M_1^-(j, \phi)\}.
    \end{align*}
    Since  $M_{1}^+(j, \phi) - M_{1}^-(j, \phi) \leq 1$, we know (\ref{eq_M_d_plus_M_d_minus_upper_bound}) is true.
    We have the similar definitions $N_d^+(j,\phi) = \max_{k,l=0,1,2,3}\mathbb P_{\phi}(S_j = k\mid y_1^m, S_{j+d} = l)$ and $N_d^{-}(j,\phi) = \min_{k,l=0,1,2,3}\mathbb P_{\phi}(S_j = k\mid y_1^m, S_{j+d} = l)$. We also have 
    \begin{align}
        \label{eq_N_d_plus_N_d_minus_upper_bound}
        |N_d^{+}(j,\phi) - N_d^{-}(j,\phi)| \leq \prod_{i=j+1}^{j+d-1} \{1-2\tau_0(y_{1i}, y_{2i})\}.
    \end{align}

    We move to the second step. Let $L<m/2$. For any $j$, let $L_1 = 1 \vee (j-L)$ and $L_2 = m \wedge (j+L)$. We claim that when $L_1 > 1$ and $L_2 <m$,
    \begin{align}
        \label{eq_difference_between_m_and_all_observations}
        &|\mathbb P_{\phi} (S_j\in\{0, 1, 2\}\mid y_1^m) - \mathbb P_{\phi} (S_j\in\{0, 1, 2\}\mid y_{-\infty}^\infty)|\notag\\ <& 3\prod_{i=L_1+1}^{j-1}\exp\{-2\tau_0(y_{1i}, y_{2i})\} + 3\prod_{i=j+1}^{L_2-1}\exp\{-2\tau_0(y_{1i}, y_{2i})\}.
    \end{align}
    We have 
    \begin{align*}
        &|\mathbb P_{\phi}(S_j\in\{0,1,2\}\mid y_1^m) - \mathbb P_{\phi}(S_j\in\{0,1,2\}\mid y_{-\infty}^\infty)|\\
        \leq &|\mathbb P_{\phi}(S_j\in\{0,1,2\}\mid y_1^m) - \mathbb P_{\phi}(S_j\in\{0,1,2\}\mid y_{-\infty}^m)| \\
        &+ |\mathbb P_{\phi}(S_j\in\{0,1,2\}\mid y_{-\infty}^m) - \mathbb P_{\phi}(S_j\in\{0,1,2\}\mid y_{-\infty}^\infty)|.
    \end{align*}
    We just need to show
    \begin{align*}
        |\mathbb P_{\phi}(S_j\in\{0,1,2\}\mid y_1^m) - \mathbb P_{\phi}(S_j\in\{0,1,2\}\mid y_{-\infty}^m)| 
        \leq& 3\prod_{i=L_1+1}^{j-1}\exp\{-2\tau_0(y_{1i}, y_{2i})\}
    \end{align*}
    and 
    \begin{align*}
        |\mathbb P_{\phi}(S_j\in\{0,1,2\}\mid y_{-\infty}^m) - \mathbb P_{\phi}(S_j\in\{0,1,2\}\mid y_{-\infty}^\infty)| \leq&  3\prod_{i=j+1}^{L_2-1}\exp\{-2\tau_0(y_{1i}, y_{2i})\}.
    \end{align*}
    We have for $k = 0,1,2$,
    \begin{align*}
        &|\mathbb P_{\phi}(S_j = k\mid y_1^m) - \mathbb P_{\phi}(S_j = k\mid y_{-\infty}^m)| \\
        =& \bigg|\sum_{l=0}^3 \mathbb P_{\phi}(S_j  = k\mid  S_{j-L}=l, y_1^m)\mathbb P_{\phi}(S_{j-L}=l \mid y_1^m)\\
        &- \sum_{l'=0}^3  \mathbb P_{\phi}(S_j  = k\mid  S_{j-L}=l', y_1^m)\mathbb P_{\phi}(S_{j-L}=l' \mid y_{-\infty}^m)\bigg|\\
        \leq & \max_{l,l'=0,1,2,3} |\mathbb P_{\phi}(S_j  = k\mid  S_{j-L}=l, y_1^m) -   \mathbb P_{\phi}(S_j  = k\mid  S_{j-L}=l', y_1^m)|\\
        \leq & M_L^+(j, \phi) - M_L^-(j, \phi)\\
        \leq & \prod_{i=L_1+1}^{j-1}\{1-2\tau_0(y_{1i}, y_{2i})\}\\
        \leq & \prod_{i=L_1+1}^{j-1}\exp\{-2\tau_0(y_{1i}, y_{2i})\}.
    \end{align*}
    Then 
    \begin{align*}
        &|\mathbb P_\phi(S_j\in\{0,1,2\}\mid y_1^m) - \mathbb P_\phi(S_j\in\{0,1,2\}\mid y_{-\infty}^m)|\\
        \leq& \sum_{k=0}^2 |\mathbb P_\phi(S_j = k\mid y_1^m) - \mathbb P_\phi(S_j = k \mid y_{-\infty}^m)| \\
        \leq& 3\prod_{i=L_1+1}^{j-1}\exp\{-2\tau_0(y_{1i}, y_{2i})\}.
    \end{align*}
    Similarly, we also have
    \begin{align*}
        |\mathbb P_\phi(S_j\in\{0,1,2\}\mid y_{-\infty}^m) - \mathbb P_\phi(S_j\in\{0,1,2\}\mid y_{-\infty}^\infty)| \leq&  3\prod_{i=j+1}^{L_2-1}\exp\{-2\tau_0(y_{1i}, y_{2i})\}.
    \end{align*}
    Therefore, (\ref{eq_difference_between_m_and_all_observations}) is true.
    Then we consider the expectations. 
    \begin{align*}
        &\mathbb E_{\phi^*}|\mathbb P_\phi (S_j\in\{0, 1, 2\}\mid y_1^m) - \mathbb P_\phi (S_j\in\{0, 1, 2\}\mid y_{-\infty}^\infty)|\\
        \leq& \mathbb E_{\phi^*} \left[3\prod_{i=L_1+1}^{j-1}\exp\{-2\tau_0(y_{1i}, y_{2i})\} + 3\prod_{i=j+1}^{L_2-1}\exp\{-2\tau_0(y_{1i}, y_{2i})\}\right]\\
        =& \mathbb E_{\phi^*}\left\{\mathbb E_{\phi^*} \left[3\prod_{i=L_1+1}^{j-1}\exp\{-2\tau_0(y_{1i}, y_{2i})\} + 3\prod_{i=j+1}^{L_2-1}\exp\{-2\tau_0(y_{1i}, y_{2i})\}\bigg|S_1,\ldots, S_m\right]\right\}\\
        =& \mathbb E_{\phi^*}\left\{3\prod_{i=L_1+1}^{j-1}\mathbb E_{\phi^*} \left[\exp\{-2\tau_0(y_{1i}, y_{2i})\}\mid S_i\right] + 3\prod_{i=j+1}^{L_2-1}\mathbb E_{\phi^*} \left[\exp\{-2\tau_0(y_{1i}, y_{2i})\}\mid S_i\right]\right\}.
    \end{align*}
    By (C5) and the construction of $\tau_0(Y_{1j}, Y_{2j})$, we have $\mathbb P_{\phi^*}(\tau_0(Y_{1j}, Y_{2j}) > 0\mid S_j = k) = 1$ for $k = 0,1,2,3$. Let 
    \begin{align}
        \beta_0 = \max_{k=0,1,2,3} \mathbb E_{\phi^*} \left[\exp\{-2\tau_0(y_{11}, y_{21})\}\mid S_1 = k\right],
        \label{eq_beta_0}
    \end{align}
    then we have $\beta_0 < 1$. Therefore, for some $C_0 > 0$, 
    \begin{align}
        \mathbb E_{\phi^*}|\mathbb P_\phi (S_j\in\{0, 1, 2\}\mid y_1^m) - \mathbb P_\phi (S_j\in\{0, 1, 2\}\mid y_{-\infty}^\infty)| \leq C_0 \beta_0^L.
        \label{eq_upper_bound_of_Ed}
    \end{align}

    By L\'{e}vy's upward theorem \citep{williams1991probability}, $\mathbb P_\phi(S_0 \in\{0,1,2\}\mid (y_{1j}, y_{2j})_{-N}^N) \rightarrow T_0^\infty$ almost surely as $N\rightarrow \infty$. 
    Next, we show that $G^\infty(q/2) = \mathbb P_\phi (T_0^\infty \leq q/2) > 0$. Note that
    \begin{align*}
        &\mathbb P_{\phi} (T_0^\infty \leq q/2)\\ =& \mathbb P_\phi \left(\frac{\mathbb P_\phi(S_j\in\{0,1,2\}\mid (y_{1j}, y_{2j})_{j=-\infty}^\infty)}{\mathbb P_\phi(S_j = 3\mid (y_{1j}, y_{2j})_{j=-\infty}^\infty)} \leq \frac{q/2}{1-q/2}\right)\\
        =& \mathbb P_\phi \left(\frac{\sum_{k=0}^2 \pi_k\mathbb P_\phi((y_{1j}, y_{2j})_{j=-\infty}^\infty\mid S_j = k)}{\pi_3\mathbb P_\phi((y_{1j}, y_{2j})_{j=-\infty}^\infty \mid S_j = 3)} \leq \frac{q/2}{1-q/2}\right)\\
        =&  \mathbb P_\phi \bigg(\sum_{k=0}^2 \frac{\sum_{l_1,l_2=0,1,2,3}\mathbb P_\phi((y_{1j}, y_{2j})_{j=-\infty}^{-1}\mid S_{-1} = l_1)a_{l_1k}a_{kl_2} \mathbb P_\phi((y_{1j}, y_{2j})_1^{\infty}\mid S_{1} = l_2)}{\sum_{l_1,l_2=0,1,2,3}\mathbb P_\phi((y_{1j}, y_{2j})_{j=-\infty}^{-1}\mid S_{-1} = l_1)a_{l_13}a_{3l_2} \mathbb P_\phi((y_{1j}, y_{2j})_1^{\infty}\mid S_{1} = l_2)}\\
        & \quad\cdot \frac{\pi_k f^{(k)}(y_{10}, y_{20})}{\pi_3f^{(3)}(y_{10}, y_{20})}\leq \frac{q/2}{1-q/2}\bigg).
    \end{align*}
    By (C2), $\varepsilon_0\leq a_{kl} \leq 1$ for $k,l=0,1,2,3$. Then we have
    \begin{align*}
        \varepsilon_0^2 \leq \frac{\sum_{l_1,l_2=0,1,2,3}\mathbb P_\phi((y_{1j}, y_{2j})_{j=-\infty}^{-1}\mid S_{-1} = l_1)a_{l_1k}a_{kl_2} \mathbb P_\phi((y_{1j}, y_{2j})_1^{\infty}\mid S_{1} = l_2)}{\sum_{l_1,l_2=0,1,2,3}\mathbb P_\phi((y_{1j}, y_{2j})_{j=-\infty}^{-1}\mid S_{-1} = l_1)a_{l_13}a_{3l_2} \mathbb P_\phi((y_{1j}, y_{2j})_1^{\infty}\mid S_{1} = l_2)} \leq \varepsilon_0^{-2}.   
    \end{align*}
    Consequently,
    \begin{align*}
        \mathbb P_{\phi} (T_0^\infty \leq q/2) \geq& \mathbb P_\phi \left\{\varepsilon_0^{-2}\left(\frac{\pi_0}{\pi_3} \frac{1}{f_1(y_{10})f_2(y_{20})} + \frac{\pi_1}{\pi_3} \frac{1}{f_2(y_{20})} + \frac{\pi_2}{\pi_3} \frac{1}{f_1(y_{10})}\right) \leq \frac{q/2}{1-q/2}\right\}\\
        =& \mathbb P_\phi \left\{\varepsilon_0^2 \frac{q/2}{1-q/2} f_1(y_{10})f_2(y_{20}) - \frac{\pi_1}{\pi_3}f_1(y_{10}) - \frac{\pi_2}{\pi_3}f_2(y_{20}) - \frac{\pi_0}{\pi_3} \geq 0\right\}.
    \end{align*}
    By (C2), $\lim_{y\rightarrow 0}f_1(y) > c$ and $\lim_{y\rightarrow 0} f_2(y)  > c$.
    Moreover, two roots of the quadratic equation
    \begin{align*}
        g_2(x) = \varepsilon_0^2 \frac{q/2}{1-q/2} x^2 - \frac{\pi_1 + \pi_2}{\pi_3}x  - \frac{\pi_0}{\pi_3} = 0
    \end{align*}
    are
    \begin{align*}
        x_0^{(1)} =& \frac{\pi_1+\pi_2 - \{(\pi_1 + \pi_2)^2 + 4 \varepsilon_0^2\pi_0\pi_3q/(2-q)\}^{1/2}}{2\varepsilon_0^2\pi_3q/(2-q)}<0 \text{ and }\\
        x_0^{(2)} =&\frac{\pi_1+\pi_2 + \{(\pi_1 + \pi_2)^2 + 4 \varepsilon_0^2\pi_0\pi_3q/(2-q)\}^{1/2}}{2\varepsilon_0^2\pi_3q/(2-q)}>0.
    \end{align*}
    By (C2), $\varepsilon_0 \leq \pi_k \leq 1-3\varepsilon_0$ for $k = 0,1,2,3$. Thus
    \begin{align*}
        x_0^{(2)} \leq \frac{1-2\varepsilon_0 + \{(1-2\varepsilon_0)^2 + 4\varepsilon_0^3(1-3\varepsilon_0)q/(2-q)\}^{1/2}}{2\varepsilon_0^3q/(2-q)} = c,
    \end{align*}
    where $c$ is defined in (C2). Thus $g_2(c) > 0$ and
    \begin{align*}
        \varepsilon_0^2 \frac{q/2}{1-q/2} c - \frac{\pi_1}{\pi_3} >& \frac{1}{c}\left(\frac{\pi_2}{\pi_3}c + \frac{\pi_0}{\pi_3}\right) > 0, \text{ and }\\
        \varepsilon_0^2 \frac{q/2}{1-q/2} c - \frac{\pi_2}{\pi_3} >& \frac{1}{c}\left(\frac{\pi_1}{\pi_3}c + \frac{\pi_0}{\pi_3}\right) > 0.
    \end{align*}
    By (C2), $\lim_{x_1\to 0}f_1(x_1) > c$ and $\lim_{x_2\to 0}f_2(x_2) > c$. Since $f_1, f_2$ are continuous, there exist $u_1, u_2 \in (0, 1)$ such that $f_1(x_1) > c$ and $f_2(x_2) > c$ whenever $0<x_1<u_1$ and $0<x_2<u_2$. Consequently, for $0<x_1<u_1$ and $0<x_2<u_2$
    \begin{align*}
        \varepsilon_0^2 \frac{q/2}{1-q/2} f_2(x_2)- \frac{\pi_1}{\pi_3} > \varepsilon_0^2 \frac{q/2}{1-q/2} c- \frac{\pi_1}{\pi_3} > 0.
    \end{align*}
    Therefore,
    \begin{align*}
        &\varepsilon_0^2 \frac{q/2}{1-q/2} f_1(x_1)f_2(x_2) - \frac{\pi_1}{\pi_3}f_1(x_1) - \frac{\pi_2}{\pi_3}f_2(x_2) - \frac{\pi_0}{\pi_3}\\
        =& \left\{\varepsilon_0^2\frac{q/2}{1-q/2}f_2(x_2) - \frac{\pi_1}{\pi_3}\right\}f_1(x_1) - \frac{\pi_2}{\pi_3} f_2(x_2) - \frac{\pi_0}{\pi_3}\\
        \geq & \left\{\varepsilon_0^2\frac{q/2}{1-q/2}f_2(x_2) - \frac{\pi_1}{\pi_3}\right\}c - \frac{\pi_2}{\pi_3} f_2(x_2) - \frac{\pi_0}{\pi_3}\\
        =&  \left\{\varepsilon_0^2\frac{q/2}{1-q/2}c - \frac{\pi_2}{\pi_3}\right\}f_2(x_2) - \frac{\pi_1}{\pi_3}c - \frac{\pi_0}{\pi_3}\\
        \geq& \varepsilon_0^2 \frac{q/2}{1-q/2} c^2 - \frac{\pi_1 + \pi_2}{\pi_3}c  - \frac{\pi_0}{\pi_3} = g_2(c) > 0.
    \end{align*}
    Therefore, we have
    \begin{align*}
        \mathbb P_{\phi} (T_0^\infty \leq q/2) \geq& \mathbb P_\phi \left\{Y_{10} \in (0, u_1), Y_{20} \in (0, u_2)\right\} > 0.
    \end{align*}
    and thus we can conclude that $\mathbb P_{\phi} (T_0^\infty \leq q/2) > 0$.

    Finally, we show that $R/m\geq G^\infty(q/2)$ and $\widehat{R}/m \geq G^\infty(q/2)$ almost surely as $m\to\infty$. We consider the case that not all hypotheses are rejected. Recall (\ref{eq_lambda_0}). The threshold $\widehat\lambda_{\rm OR}$ satisfies $\widehat\lambda_{\rm OR} \geq q$ with probability $1$. 
    It suffices to show that $m^{-1}\sum_{j=1}^m I(T_j \leq q) \geq G^\infty(q/2)$ almost surely as $m\to\infty$. Take $L_m = m^\kappa$, with $\kappa \in (0, 1)$. $L_m$ satisfies $L_m < m/2$ when $m$ is large enough. For any $j$ satisfying $L_m+1 < j < m-L_m-1$, by (\ref{eq_difference_between_m_and_all_observations}), we have
    \begin{align*}
        |T_j - T_j^\infty| =& |\mathbb P_{\phi^*}(s_j\in\{0,1,2\}\mid y_1^m) - \mathbb P_{\phi^*}(s_j\in\{0,1,2\}\mid y_{-\infty}^\infty)|
        \\
        <& 3 \prod_{i=j-L_m+1}^{j-1} \exp\{-2\tau_0(y_{1i},y_{2i})\} + 3 \prod_{i=j+1}^{j+L_m-1} \exp\{-2\tau_0(y_{1i},y_{2i})\} 
    \end{align*}
    with probability $1$. 
    
    Define $d_j((y_{1i}, y_{2i})_1^m) = 3 \prod_{i=j-L_m+1}^{j-1} \exp\{-2\tau_0(y_{1i},y_{2i})\} + 3 \prod_{i=j+1}^{j+L_m-1} \exp\{-2\tau_0(y_{1i},y_{2i})\}$. Then $d_j((y_{1i}, y_{2i})_1^m)$ is ergodic. Thus Birkhoff's ergodic theorem \citep{birkhoff1931proof} gives that 
    \begin{align*}
        \frac{1}{m-2L_m-1} \sum_{j=L_m+1}^{m-L_m-1} I(d_j > q/2)\rightarrow \mathbb P_{\phi^*} (d_1 > q/2)\text{ in probability}. 
    \end{align*}
    Moreover, $\mathbb E_{\phi^*}[d_j] < C_0\beta_0^{L_m}$ by the construction of $\beta_0$ in (\ref{eq_beta_0}). Then Markov's inequality gives 
    \begin{align*}
        \mathbb P_{\phi^*} (d_1 > q/2) \leq 
        \frac{\mathbb E_{\phi^*}[d_j]}{q/2} \rightarrow 0 \text{ as } L_m = m^\kappa \rightarrow \infty.
    \end{align*}
    Thus
    \begin{align*}
        \frac{1}{m}\sum_{j=1}^m I\left(|T_j - T_j^\infty| > q/2\right) \leq \frac{2L_m+1}{m} + \frac{1}{m}\sum_{j=L_m+1}^{m-L_m-1} I(d_j > q/2) \rightarrow 0\text{ in probability}
    \end{align*}
    as $m \to \infty$.
    We use the property that $I(T_j \leq q) + I(|T_j - T_j^\infty| > q/2) \geq I(T_j^\infty \leq q/2)$. Then
    \begin{align*}
        \frac{1}{m}\sum_{j=1}^mI(T_j \leq q) + \frac{1}{m}\sum_{j=1}^mI(|T_j - T_j^\infty| > q/2) \geq \frac{1}{m}\sum_{j=1}^m I(T_j^\infty \leq q/2).
    \end{align*}
    By Birkhoff's ergodic theorem, we have
    \begin{align*}
        \frac{1}{m}\sum_{j=1}^m I(T_j^\infty \leq q/2) \rightarrow G^{\infty}(q/2) \text{ almost surely as } m\to\infty.
    \end{align*}
    Then we have $m^{-1}\sum_{j=1}^mI(T_j \leq q) \geq G^\infty(q/2)$ almost surely. We have shown that $G^\infty(t) > 0$ for $0<t <1$. Therefore, $m^{-1}\sum_{j=1}^mI(T_j \leq q) \geq G^\infty(q/2) > 0$ almost surely, which means $R/m \geq G^\infty(q/2)$ almost surely. We can use a similar argument to show that $\widehat{R}/m \geq G^\infty(q/2)$ almost surely. The details are omitted.
\end{proof}

\subsection{Proof of Lemma \ref{lemma_consistency_when_lambda_less_than_alpha}}
        \begin{proof}
        We prove the lemma via 2 steps. We first show that the difference between the averages of rejected true and estimated test statistics is small in expectation as $m\to\infty$. Then we use a contradiction argument to show the result of the lemma.

        \paragraph{Step 1. The difference between the averages of rejected true and estimated test statistics is small in expectation as $m\to\infty$.} Note that
        $R \rightarrow \infty$ almost surely as $m\to\infty$ as shown in Lemma \ref{lemma_R_to_infty}. The rejection criteria in (\ref{eq_oracle_test_2}) implies that 
        $$
        \frac{1}{R} \sum_{j=1}^{R} T_{(j)} \leq q <\frac{1}{R+1} \sum_{j=1}^{R+1} T_{(j)}.$$
        Note that as $m \rightarrow \infty$,
        \begin{align*}
            \mathbb E\left|\frac{1}{R} \sum_{j=1}^{R} T_{(j)}-\frac{1}{R+1} \sum_{j=1}^{R+1} T_{(j)}\right|=\mathbb E\left|\frac{\sum_{j=1}^{R}\left(T_{(j)}-T_{(R+1)}\right)}{R(R+1)}\right| \leq \mathbb E\left|\frac{1}{R+1}\right| \rightarrow 0.
        \end{align*}
        Since 
        \begin{align*}
            0\leq&\left|\frac{1}{R} \sum_{j=1}^{R} T_{(j)}-q\right| \leq \left|\frac{1}{R} \sum_{j=1}^{R} T_{(j)}-\frac{1}{R+1} \sum_{j=1}^{R+1} T_{(j)}\right|,
        \end{align*}
        we have
        \begin{align}
            \mathbb E\left|\frac{1}{R} \sum_{j=1}^{R} T_{(j)}-q\right| \rightarrow& 0 \text{ as }m\rightarrow\infty.\label{eq_average_of_T}
        \end{align}
        We can use the same approach to show
        \begin{align}
            \mathbb E\left|\frac{1}{\widehat R}\sum_{j=1}^{\widehat R} \widehat T_{(j)}-q\right| \rightarrow& 0 \text{ as }m\rightarrow\infty.\label{eq_average_of_T_hat}
        \end{align}
        Moreover, since $1/\widehat R\to 0$ almost surely as $m\to\infty$ as Lemma \ref{lemma_R_to_infty} shows, we can also show that 
        \begin{align}
            \label{eq_average_of_T_hat_as}
            \frac{1}{\widehat R}\sum_{j=1}^{\widehat R} \widehat T_{(j)} \to q \text{ almost surely as }m\to\infty.
        \end{align}
        Combining (\ref{eq_average_of_T}) and (\ref{eq_average_of_T_hat}), we have
        \begin{equation}
            \mathbb E\left|\frac{1}{\widehat{R}} \sum_{j=1}^{\widehat{R}} \widehat{T}_{(j)}-\frac{1}{R} \sum_{j=1}^{R} T_{(j)}\right| \to 0 \text{ as }m\rightarrow\infty.
            \label{eq_consistency_of_avg_T}
        \end{equation}
        
        \paragraph{Step 2. Contradiction argument to show the result of the lemma.} We    
        finish the proof by contradiction. Assume that $\lim_{m\rightarrow\infty}\mathbb E|R / \widehat{R} - 1| = 0$ does not hold, where $R$ is the total number of rejections by (\ref{eq_oracle_test_2}) when the total number of hypotheses is $m$ and $\widehat R$ is the total number of rejections by (\ref{eq_rej_procedure}) when the total number of hypotheses is $m$. Then there is $\varepsilon_1>0$ such that, for any $M>0$, there exists some $m\geq M$ satisfying $\mathbb E|R / \widehat{R}-1|>\varepsilon_1$. Since 
        \begin{align*}
            \mathbb E|R/\widehat R - 1| =& \mathbb E\{(1-R/\widehat R)I(\widehat R > R)\} + \mathbb E\{(R/\widehat R - 1)I(R > \widehat R)\},
        \end{align*}
        $\mathbb E|R / \widehat{R}-1|>\varepsilon_1$ implies that either (i) $\mathbb E\{(1-R/\widehat R)I(\widehat R > R)\} > \varepsilon_1/2$, or (ii) $\mathbb E\{(R/\widehat R - 1)I(R > \widehat R)\} > \varepsilon_1/2$.

        \paragraph{Step 2.1. Contradiction argument for case (i).} We
        first consider the case that (i) is true. Then $\mathbb E\{(1-R/\widehat R)I(\widehat R > R)\} > \varepsilon_1/2$ and therefore the event $E_1 = \{\widehat R > R\}$ has positive probability. On the event $E_1$, we have
        \begin{align}
        &\left|\frac{1}{\widehat{R}} \sum_{j=1}^{\widehat{R}} \widehat{T}_{(j)}-\frac{1}{R} \sum_{j=1}^{R} T_{(j)}\right|\notag\\
        =&\left|\frac{1}{\widehat{R}} \sum_{j=1}^{\widehat{R}} \widehat{T}_{(j)}-\frac{1}{\widehat{R}} \sum_{j=1}^{\widehat{R}} T_{(j)}+\frac{1}{\widehat{R}} \sum_{j=1}^{\widehat{R}} T_{(j)}-\frac{1}{R} \sum_{j=1}^{R} T_{(j)}\right| \notag\\
        =&\left|\frac{1}{\widehat{R}} \sum_{j=1}^{\widehat{R}} \left(\widehat{T}_{(j)}-T_{(j)}\right)+\frac{1}{\widehat{R}} \left(\sum_{j=1}^{R}+\sum_{j=R+1}^{\widehat{R}}\right) T_{(j)}-\frac{1}{R} \sum_{j=1}^{R} T_{(j)}\right| \notag\\
        =&\left|\frac{1}{\widehat{R}} \sum_{j=1}^{\widehat{R}} \left(\widehat{T}_{(j)}-T_{(j)}\right) + \frac{1}{\widehat{R}} \sum_{j=R+1}^{\widehat{R}} T_{(j)} -\left(1-\frac{R}{\widehat{R}}\right)\frac{1}{R} \sum_{j=1}^{R} T_{(j)} \right|\notag \\
        \geq&\left|\frac{1}{\widehat{R}} \sum_{j=R+1}^{\widehat{R}} T_{(j)} -\left(1-\frac{R}{\widehat{R}}\right)\frac{1}{R} \sum_{j=1}^{R} T_{(j)}\right| - \left|  \frac{1}{\widehat{R}} \sum_{j=1}^{\widehat{R}} \left(\widehat{T}_{(j)}-T_{(j)}\right)\right| \label{eq_T_ineq_1}\\
        \geq& \left|1-\frac{R}{\widehat{R}}\right|\left|T_{\left(R+1\right)}-\frac{1}{R} \sum_{j=1}^{R} T_{(j)}\right|-\left|\frac{1}{\widehat{R}} \sum_{j=1}^{\widehat{R}} \widehat{T}_{(j)}-\frac{1}{\widehat{R}} \sum_{j=1}^{\widehat{R}} T_{(j)}\right|,\label{eq_T_ineq_2}
        \end{align}
        where (\ref{eq_T_ineq_1}) holds due to triangle inequality $|a+b| \geq |b| - |a|$ and (\ref{eq_T_ineq_2}) holds because
        \begin{align*}
            &\frac{1}{\widehat{R}} \sum_{j=R+1}^{\widehat{R}} T_{(j)} -\left(1-\frac{R}{\widehat{R}}\right)\frac{1}{R} \sum_{j=1}^{R} T_{(j)}\\
            \geq& \frac{1}{\widehat{R}} \sum_{j=R+1}^{\widehat{R}} T_{(R+1)} - \left(1-\frac{R}{\widehat{R}}\right)\frac{1}{R} \sum_{j=1}^{R} T_{(j)}\\
            =& \frac{\widehat{R} - R}{\widehat{R}}T_{(R+1)} - \left(1-\frac{R}{\widehat{R}}\right)\frac{1}{R} \sum_{j=1}^{R} T_{(j)}\\
            =& \left(1-\frac{R}{\widehat{R}}\right) \left(T_{(R+1)} - \frac{1}{R} \sum_{j=1}^{R} T_{(j)}\right)
            \geq 0.
        \end{align*}

        In Step 2.1, our goal is to show that the right-hand side of (\ref{eq_T_ineq_2}) is positive with probability $1$ and thus contradicts with (\ref{eq_consistency_of_avg_T}). We first achieve Goal 1 that $|T_{(R+1)} - R^{-1}\sum_{j=1}^R T_{(j)}|$ in (\ref{eq_T_ineq_2}) is positive with probability $1$. Then we achieve Goal 2 that $\left|\frac{1}{\widehat{R}} \sum_{j=1}^{\widehat{R}} \widehat{T}_{(j)}-\frac{1}{\widehat{R}} \sum_{j=1}^{\widehat{R}} T_{(j)}\right|$ in (\ref{eq_T_ineq_2}) converges to $0$ in probability as $m\to\infty$.

            \paragraph{Goal 1:} 
            We show that $|T_{(R+1)} - R^{-1}\sum_{j=1}^R T_{(j)}|$ in (\ref{eq_T_ineq_2}) is positive with probability $1$.
         Since the event $\{R^{-1}\sum_{j=1}^R T_{(j)} \leq q\}$ has probability $1$, it suffices to show that
        \begin{align}
            T_{(R+1)} \geq& \lambda_{\mathrm{OR}}^{\infty} \text{ with probability approaching $1$.}\label{eq_T_R_plus_1_lower_bound}\\
            \lambda_{\rm OR}^\infty >&  q. \label{eq_lambda_OR_greater_than_q}
        \end{align}
        Since $\widehat{\lambda}_{\mathrm{OR}} > q$ with probability $1$, for $0<\gamma<1$ and $\widehat Q_{\rm OR}$ defined in (\ref{eq_Q_hat_OR}), we have
        \begin{align}
            \widehat{Q}_{\mathrm{OR}}(\widehat{\lambda}_{\mathrm{OR}}) =& \frac{1}{R}\sum_{j=1}^m T_j I(T_j \leq \widehat\lambda_{\rm OR})\notag\\
            =& \frac{1}{R}\sum_{j=1}^m T_j I(T_j \leq \gamma q) + \frac{1}{R}\sum_{j=1}^m T_j I(\gamma q < T_j \leq  \widehat{\lambda}_{\mathrm{OR}})\notag\\
            \leq& \gamma q\frac{\sum_{j=1}^m  I(T_j \leq \gamma q)}{\sum_{j=1}^m  I(T_j \leq  \widehat{\lambda}_{\mathrm{OR}})} + \widehat{\lambda}_{\mathrm{OR}}\frac{\sum_{j=1}^m  I(\gamma q < T_j \leq  \widehat{\lambda}_{\mathrm{OR}})}{\sum_{j=1}^m  I(T_j \leq  \widehat{\lambda}_{\mathrm{OR}})}. \label{eq_Q_hat_upper_bound}
        \end{align}
        
        Lemma \ref{lemma_consistency_of_thresholds} shows that $\widehat{\lambda}_{\mathrm{OR}}^\infty \rightarrow \lambda_{\mathrm{OR}}^\infty$ in probability, and the construction of $T_j$ and $T_j^\infty$ gives that $\widehat \lambda_{\rm OR} - \widehat \lambda_{\rm OR}^\infty \rightarrow 0$ in probability. Therefore, 
        \begin{align}
            \widehat\lambda_{\rm OR} \to \lambda_{\rm OR}^\infty\text{ in probability as } m\to\infty.
            \label{eq_lambda_hat_OR_to_lambda_OR_inf}
        \end{align}
        Similarly, we also have
        \begin{align}
            \widehat\lambda_{\rm rLIS} \to \lambda_{\rm OR}^\infty\text{ in probability as } m\to\infty.
            \label{eq_lambda_hat_rLIS_to_lambda_OR_inf}
        \end{align}
        Combining (\ref{eq_lambda_hat_OR_to_lambda_OR_inf}) and (\ref{eq_lambda_hat_rLIS_to_lambda_OR_inf}), we have
        \begin{align}
            \widehat \lambda_{\rm OR} - \widehat\lambda_{\rm rLIS} \to 0 \text{ in probability as } m\to\infty.
            \label{eq_lambda_hat_OR_minus_lambda_hat_rLIS_to_0}
        \end{align}
        By the rejection criteria (\ref{eq_oracle_test_1}) and (\ref{eq_oracle_test_2}), $T_{(R+1)} > \widehat\lambda_{\rm OR}$ with probability $1$. Thus (\ref{eq_T_R_plus_1_lower_bound}) holds. Moreover, for any $\epsilon > 0$, $\mathbb P(|\widehat\lambda_{\rm OR} - \lambda_{\rm OR}^\infty| > \epsilon) \rightarrow 0$ as $m\rightarrow\infty$. Then on the event $|\widehat\lambda_{\rm OR} - \lambda_{\rm OR}^\infty| \leq \epsilon$,
        \begin{align*}
            \frac{1}{m}\sum_{j=1}^m I(T_j \leq \widehat\lambda_{\rm OR}) \leq& \frac{1}{m}\sum_{j=1}^m I(T_j \leq \lambda_{\rm OR}^\infty + \epsilon),\\
            \frac{1}{m}\sum_{j=1}^m I(T_j \leq \widehat\lambda_{\rm OR}) \geq& \frac{1}{m}\sum_{j=1}^m I(T_j \leq \lambda_{\rm OR}^\infty - \epsilon).
        \end{align*}
        Birkhoff's ergodic theorem \citep{birkhoff1931proof} gives that as $m\rightarrow\infty$,
        \begin{align}
            \frac{1}{m}\sum_{j=1}^m I(T_j \leq \lambda_{\rm OR}^\infty + \epsilon) \rightarrow&  G^\infty(\lambda_{\rm OR}^\infty + \epsilon) \text{ almost surely, and } \notag\\
            \frac{1}{m}\sum_{j=1}^m I(T_j \leq \lambda_{\rm OR}^\infty - \epsilon) \rightarrow& G^\infty(\lambda_{\rm OR}^\infty - \epsilon) \text{ almost surely.}\notag
        \end{align}
        When $\epsilon$ tends to $0$, the continuity of $G^\infty$ gives that as $m \rightarrow \infty$,
        \begin{align}
            \frac{1}{m}\sum_{j=1}^m I(T_j \leq \widehat\lambda_{\rm OR}) \rightarrow G^\infty(\lambda_{\rm OR}^\infty) \text{ almost surely.} \label{eq_Birkhoff1}
        \end{align} 
        Similarly,
        \begin{align}
            \frac{1}{m}\sum_{j=1}^m I(\widehat T_j \leq \widehat\lambda_{\rm rLIS}) \rightarrow G^\infty(\lambda_{\rm OR}^\infty) \text{ almost surely.} \label{eq_limit_of_R_hat_over_m}
        \end{align}
        Moreover, by Birkhoff's ergodic theorem \citep{birkhoff1931proof}, we have
        \begin{align}
            \frac{1}{m}\sum_{j=1}^m I(T_j \leq \gamma q) \rightarrow&  G^\infty(\gamma q)\text{ almost surely.}\label{eq_Birkhoff2}
        \end{align}
        
        Combining (\ref{eq_Q_hat_upper_bound}), (\ref{eq_Birkhoff1}) and (\ref{eq_Birkhoff2}), we have
        \begin{align*}
            \widehat{Q}_{\mathrm{OR}}(\widehat{\lambda}_{\mathrm{OR}}) \leq &\gamma q\frac{G^\infty(\gamma q)}{G^\infty(\lambda_{\mathrm{OR}}^\infty)} + \lambda_{\mathrm{OR}}^\infty\frac{G^\infty(\lambda_{\mathrm{OR}}^\infty) - G^\infty(\gamma q)}{G^\infty(\lambda_{\mathrm{OR}}^\infty)},
        \end{align*}
        with probability approaching $1$.
        Recall that $\widehat{Q}_{\mathrm{OR}}(\widehat{\lambda}_{\mathrm{OR}}) = R^{-1} \sum_{j=1}^{R} T_{(j)} = q + o_p(1)$. Thus for any $\gamma \in (0, 1)$,
        \begin{align*}
            q \leq& \gamma q\frac{G^\infty(\gamma q)}{G^\infty(\lambda_{\mathrm{OR}}^\infty)} + \lambda_{\mathrm{OR}}^\infty\frac{G^\infty(\lambda_{\mathrm{OR}}^\infty) - G^\infty(\gamma q)}{G^\infty(\lambda_{\mathrm{OR}}^\infty)} + o(1).
        \end{align*}
        Equivlently, we have
        \begin{align*}
            \lambda_{\mathrm{OR}}^\infty \geq& \sup_{\gamma \in (0, 1)}\left\{q + \frac{q(1-\gamma)G^\infty(\gamma q)}{G^\infty(\lambda_{\mathrm{OR}}^\infty) - G^\infty(\gamma q)}\right\} > q.
        \end{align*}
        Since $G^\infty$ is strictly increasing in $(0, \alpha_*)$ and $\alpha_* > \lambda_{\rm OR}^\infty$, we have $G^\infty(\gamma q) > 0$ and $G^\infty(\lambda_{\mathrm{OR}}^\infty) - G^\infty(\gamma q) > 0$. Therefore, (\ref{eq_lambda_OR_greater_than_q}) holds.

        By (\ref{eq_average_of_T}), (\ref{eq_T_R_plus_1_lower_bound}) and (\ref{eq_lambda_OR_greater_than_q}), we have
        \begin{align*}
            &\mathbb E\left(T_{(R+1)} - \frac{1}{R}\sum_{j=1}^R T_{(R)}\right) \\
            =& \mathbb E\left(T_{(R+1)} - \lambda_{\rm OR}^\infty\right) + (\lambda_{\rm OR}^\infty - q) + \mathbb E\left(q -  \frac{1}{R}\sum_{j=1}^R T_{(R)}\right)\\
            \geq& \lambda_{\rm OR}^\infty - q \text{ as } m\to\infty.
        \end{align*}
        It implies that
        \begin{align}
            T_{(R+1)} - \frac{1}{R}\sum_{j=1}^R T_{(R)} \geq \lambda_{\rm OR} - q \text{ with probability approaching } 1, \label{eq_T_R_plus_1_minus_avg_T_greater_than_lambda_minus_q}
        \end{align}
        thus Goal 1 is achieved.

            \paragraph{Goal 2:}
            We show
            \begin{align}
                \mathbb E \left| \frac{1}{\widehat{R}} \sum_{j=1}^{\widehat{R}} \widehat{T}_{(j)}-\frac{1}{\widehat{R}}  \sum_{j=1}^{\widehat{R}} T_{(j)} \right| \to 0 \text{ as } m\to\infty. \label{eq_E_avg_T_hat_minus_T}
            \end{align}

        Note that 
        \begin{align}
            \left|\frac{1}{\widehat R}\sum_{j=1}^{\widehat R} T_{(j)} - \frac{1}{\widehat R}\sum_{j=1}^{\widehat R}\widehat T_{(j)}\right| \leq & \left|\frac{1}{\widehat R}\sum_{j=1}^{\widehat R} T_{(j)} - \frac{1}{\widehat R}\sum_{j=1}^{\widehat R} T_{(j)}^\infty\right|\notag\\
            +& \left|\frac{1}{\widehat R}\sum_{j=1}^{\widehat R} T_{(j)}^\infty - \frac{1}{\widehat R}\sum_{j=1}^{\widehat R} \widehat T_{(j)}^\infty\right|\notag\\
            +& \left|\frac{1}{\widehat R}\sum_{j=1}^{\widehat R} \widehat T_{(j)}^\infty - \frac{1}{\widehat R}\sum_{j=1}^{\widehat R} \widehat T_{(j)}\right|.
            \label{eq_R_hat_avg_T_minus_T_hat_3_bounds}
        \end{align}
        Denote $S_{\widehat R} = \{j: T_j \leq T_{(\widehat R)}\}$ and $S_{\widehat R}^\infty = \{j: T_j^\infty \leq T_{(\widehat R)}^\infty\}$. Since $\sum_{j\in S_{\widehat R}} T_j \leq \sum_{j\in S_{\widehat R}^\infty} T_j$ and $\sum_{j\in S_{\widehat R}^\infty} T_j^\infty \leq \sum_{j\in S_{\widehat R}} T_j^\infty$, we have
        \begin{align*}
            \sum_{j\in S_{\widehat R}} T_j - \sum_{j\in S_{\widehat R}} T_j^\infty
            \leq
            \sum_{j\in S_{\widehat R}} T_j - \sum_{j\in S_{\widehat R}^\infty }T_j^\infty
            \leq 
            \sum_{j\in S_{\widehat R}^\infty} T_j - \sum_{j\in S_{\widehat R}^\infty }T_j^\infty.
        \end{align*}
        Therefore,
        \begin{align}
            \left|\sum_{j\in S_{\widehat R}} T_j - \sum_{j\in S_{\widehat R}^\infty }T_j^\infty\right|
            \leq 
            \left|\sum_{j\in S_{\widehat R}} T_j - \sum_{j\in S_{\widehat R}} T_j^\infty\right|
            +
            \left|\sum_{j\in S_{\widehat R}^\infty} T_j - \sum_{j\in S_{\widehat R}^\infty }T_j^\infty\right|.
            \label{eq_avg_T_minus_T_inf_upper_bound}
        \end{align}
        For $L_m = m^\kappa$ with $\kappa \in (0, 1)$, denote 
        \begin{align*}
            \tilde S_{\widehat R} =& \{j\in S_{\widehat R}: L_m+1 < j < m-L_m-1\}\text{, and}\\
            \tilde S_{\widehat R}^\infty =& \{j\in S_{\widehat R}^\infty: L_m+1 < j < m-L_m-1\}.
        \end{align*}
        By (\ref{eq_upper_bound_of_Ed}), $\mathbb E|T_j - T_j^\infty| \leq  C_0\beta_0^{L_m}$ for $\beta_0 \in (0, 1)$ if $L_m+1 < j < m-L_m-1$. Whenever $j \leq L_m+1$ or $j \geq m-L_m-1$, $|T_j - T_j^\infty| \leq 1$. Note that Lemma \ref{lemma_R_to_infty} shows that $\widehat R / m \geq G^\infty(q/2)$ almost surely as $m\to\infty$. Therefore, by (\ref{eq_avg_T_minus_T_inf_upper_bound}),
        \begin{align}
            \mathbb E\left|\frac{1}{\widehat R}\sum_{j\in S_{\widehat R}} T_j - \frac{1}{\widehat R}\sum_{j\in S_{\widehat R}^\infty }T_j^\infty\right|
            \leq&
            2\mathbb E\left(\frac{2L_m + 1}{\widehat R}\right) + \mathbb E\left|\frac{1}{\widehat R}\sum_{j\in \tilde S_{\widehat R}}(T_j - T_j^\infty)\right| + \mathbb E\left|\frac{1}{\widehat R}\sum_{j\in \tilde S_{\widehat R}^\infty}(T_j - T_j^\infty)\right|\notag
            \\
            \leq& \frac{2(2m^\kappa + 1)}{mG^\infty(q/2)} + 2 \max_{L_m+1<j<m-L_m-1}\mathbb E |T_j - T_j^\infty|
            \notag\\
            \leq& \frac{2(2m^\kappa + 1)}{mG^\infty(q/2)} + 2 C_0\beta_0^{m^\kappa} \to 0 \text{ as } m\to\infty.\label{eq_E_upper_bound_1}
        \end{align}
        Similarly, for $\widehat S_{\widehat R} = \{j: \widehat T_j \leq \widehat T_{(\widehat R)}\}$ and $\widehat S_{\widehat R}^\infty = \{j: \widehat T_j^\infty \leq \widehat T_{(\widehat R)}^\infty\}$, we have
        \begin{align}
            \mathbb E\left|\frac{1}{\widehat R}\sum_{j\in \widehat S_{\widehat R}^\infty} \widehat T_j^\infty - \frac{1}{\widehat R}\sum_{j\in \widehat S_{\widehat R} }\widehat T_j\right|
            \leq \frac{2(2m^\kappa + 1)}{mG^\infty(q/2)} + 2 C_0\beta_0^{m^\kappa} \to 0 \text{ as } m\to\infty.
            \label{eq_E_upper_bound_3}
        \end{align}
        
        Furthermore, denote $S_{\widehat R}^\infty = \{j: T_j^\infty \leq T_{(\widehat R)}^\infty\}$ and $\widehat S_{\widehat R}^\infty = \{j: \widehat T_j^\infty \leq \widehat T_{(\widehat R)}^\infty\}$. By (\ref{eq_limit_of_R_hat_over_m}), we have $\widehat R/m \to G^\infty(\lambda_{\rm OR}^\infty)$ almost surely as $m\to\infty$. By Birkhoff's ergodic theorem \citep{birkhoff1931proof}, we have
        $T_{(\widehat R)}^\infty \to \lambda_{\rm OR}^\infty$ almost surely as $m\to\infty$. 
        Note that
        \begin{align*}
            \frac{1}{\widehat R}\sum_{j\in S_{\widehat R}^\infty} T_j^\infty =& \frac{1}{\widehat R} \sum_{j\in S_{\widehat R}^\infty} T_j^\infty I(T_j^\infty \leq T_{(\widehat R)}^\infty).
        \end{align*}
        Birkhoff's ergodic theorem \citep{birkhoff1931proof} gives that
        \begin{align*}
            \mathbb E\left\{\frac{1}{\widehat R} \sum_{j\in S_{\widehat R}^\infty} T_j^\infty I(T_j^\infty \leq T_{(\widehat R)}^\infty)\right\} \to & \mathbb E\{T_1^\infty \mid T_1^\infty \leq \lambda_{\rm OR}^\infty\}
            = \frac{1}{G^\infty(\lambda_{\rm OR}^\infty)}\int_0^{\lambda_{\rm OR}^\infty} x {\rm d}G^\infty(x) \text{ as } m\to\infty.
        \end{align*}
        Similarly, we have $\widehat T_{(\widehat R)}^\infty \to \lambda_{\rm OR}^\infty$ almost surely as $m\to\infty$ by Birkhoff's ergodic theorem \citep{birkhoff1931proof}. Therefore, 
        \begin{align*}
            \mathbb E\left\{\frac{1}{\widehat R}\sum_{j\in\widehat S_{\widehat R}^\infty}\widehat T_j^\infty\right\}\to \frac{1}{G^\infty(\lambda_{\rm OR}^\infty)}\int_0^{\lambda_{\rm OR}^\infty} x {\rm d}G^\infty(x) \text{ as } m\to\infty.
        \end{align*}
        Therefore, we have
        \begin{align}
            \mathbb E\left|\frac{1}{\widehat R}\sum_{j\in S_{\widehat R}^\infty} T_j^\infty - \frac{1}{\widehat R}\sum_{j\in\widehat S_{\widehat R}^\infty}\widehat T_j^\infty\right| \to 0 \text{ as } m\to\infty. \label{eq_E_upper_bound_2}
        \end{align}
        Combining (\ref{eq_R_hat_avg_T_minus_T_hat_3_bounds}), (\ref{eq_E_upper_bound_1}), (\ref{eq_E_upper_bound_3}) and (\ref{eq_E_upper_bound_2}), we have (\ref{eq_E_avg_T_hat_minus_T}) and thus Goal 2 is achieved.

        Then by (\ref{eq_T_ineq_2}), (\ref{eq_T_R_plus_1_minus_avg_T_greater_than_lambda_minus_q}) and (\ref{eq_E_avg_T_hat_minus_T}), for any $M>0$, there exists some $m \geq M$ satisfying
        \begin{align*}
            &\mathbb E\left|\frac{1}{\widehat{R}} \sum_{j=1}^{\widehat{R}} \widehat{T}_{(j)}-\frac{1}{R} \sum_{j=1}^{R} T_{(j)}\right|\\
            \geq& \mathbb E\left\{\left|1 - \frac{R}{\widehat R}\right|\cdot I(\widehat R >R)\cdot \bigg|T_{(R+1)} - \frac{1}{R}\sum_{j=1}^R T_{(j)}\bigg|\right\} - \mathbb E\bigg|\frac{1}{\widehat R}\sum_{j=1}^{\widehat R}\widehat T_{(j)} - \frac{1}{\widehat R}\sum_{j=1}^{\widehat R} T_{(j)} \bigg|\\
            >&\frac{\varepsilon_1}{2}\left|\lambda_{\mathrm{OR}}^{\infty}- q \right|+o(1).
        \end{align*}
        This is a contradiction to (\ref{eq_consistency_of_avg_T}). Therefore, (i) does not hold.

        \paragraph{Step 2.2. Contradiction argument for case (ii).} Now
        consider the case when (ii) is true. In this case, $\mathbb E\{(R/\widehat R - 1)I(R> \widehat R)\} > \varepsilon_1/2$ and therefore the event  $E_2 = \{R/\widehat R > 1+\varepsilon_1/2\}$ has positive probability. By (\ref{eq_average_of_T_hat}) and (\ref{eq_E_avg_T_hat_minus_T}), we have 
        \begin{align}
            (1 / \widehat{R}) \sum_{j=1}^{\widehat{R}} T_{(j)} = q
            \text{  with probability approaching $1$}.
            \label{eq_R_hat_avg_T_equal_q}
        \end{align}
        Thus $T_{(\widehat R+1)} \geq q $ with probability $1$. Then we can use a similar method as (\ref{eq_T_ineq_2}) and obtain that on the event $E_2$,
        \begin{align*}
            \left|\frac{1}{R} \sum_{j=1}^{R} T_{(j)}-\frac{1}{\widehat{R}} \sum_{j=1}^{\widehat{R}} \widehat{T}_{(j)}\right|
            =&\left|\frac{1}{\widehat{R}} \sum_{j=1}^{\widehat{R}} \left(\widehat{T}_{(j)}-T_{(j)}\right) + \left(1-\frac{\widehat{R}}{R}\right)\left(\frac{1}{\widehat{R}}\sum_{j=1}^{\widehat{R}}T_{(j)}- \frac{1}{R - \widehat{R}}\sum_{j=\widehat{R}+1}^{R}T_{(j)}\right)  \right|\notag \\
            \geq&\left|1-\frac{\widehat{R}}{R}\right|\left|\frac{1}{R-\widehat{R}} \sum_{j=\widehat{R}+1}^{R} T_{(j)}-\frac{1}{\widehat{R}} \sum_{j=1}^{\widehat{R}} T_{(j)}\right|-\left|\frac{1}{\widehat{R}} \sum_{j=1}^{\widehat{R}} \widehat{T}_{(j)}-\frac{1}{\widehat{R}} \sum_{j=1}^{\widehat{R}} T_{(j)}\right| .
        \end{align*}
        In Step 2.2, our goal is to show that the right-hand side of the above inequality is positive with probability $1$ and thus contradicts with (\ref{eq_consistency_of_avg_T}). 

        Let $\eta \in \left(q,  \lambda_{\mathrm{OR}}^{\infty}\right), \mathcal S_1=\left\{j: T_{(\widehat R+1)}\leq T_{(j)} \leq \eta\right\}$ and $\mathcal S_2=\left\{j: \eta<T_{(j)}\leq T_{(R)}\right\}$. We know $|\mathcal{S}_1| + |\mathcal{S}_2| = R-\widehat{R}$, where $|\cdot |$ denotes the cardinality of a set.
        Since $T_{(\widehat R+1)} \geq q$  with probability $1$, we have
        \begin{align*}
            \frac{1}{R-\widehat{R}} \sum_{j=\widehat{R}+1}^{R} T_{(j)} 
            =&  \frac{1}{R-\widehat{R}}\left( \sum_{j \in S_1} T_{(j)} + \sum_{j \in S_2} T_{(j)}\right)\\
            \geq & \frac{1}{R-\widehat{R}}\bigg(|\mathcal{S}_1|q + |\mathcal{S}_2|(\eta - q + q) + o_p(1)\bigg) 
            \\
            =& q+\frac{\left|\mathcal{S}_2\right|}{R-\widehat{R}}(\eta-q)+o_p(1).
        \end{align*}
        We apply the ergodic theorem \citep{birkhoff1931proof} and continuity of $G^{\infty}$ to obtain 
        \begin{align*}
            \frac{1}{m}\left|\mathcal{S}_2\right|=& 
            \frac{1}{m}\sum_{j=1}^m I\left(\eta<T_{j}\leq T_{(R)}\right) \rightarrow G^{\infty}\left(\lambda_{\mathrm{OR}}^{\infty}\right)-G^{\infty}(\eta) \text{ almost surely}
        \end{align*} 
        as $m\to\infty.$
        Since $T_{(R)} \leq \lambda_{\rm OR}^\infty$ with probability $1$ and $T_{(\widehat R+1)} \geq q$ with probability, we have
        $$
        \frac{1}{m}\left(R-\widehat{R}\right) = \frac{1}{m}\sum_{j=1}^m I\left(T_{(\widehat R+1)}\leq T_j \leq T_{(R)}\right) \leq G^\infty\left(\lambda_{\mathrm{OR}}^{\infty}\right)-G^\infty(q) \text{ almost surely}$$
        as $m\to\infty.$
        Since $|\mathcal S_2|/(R - \widehat R) \leq 1$, the continuous mapping theorem gives that
        \begin{align}
            \frac{1}{R-\widehat{R}} \sum_{j=\widehat{R}+1}^{R} T_{(j)} \geq q+ \frac{G^{\infty}\left(\lambda_{\mathrm{OR}}^{\infty}\right)-G^{\infty}(\eta)}{G^\infty \left(\lambda_{\mathrm{OR}}^{\infty}\right)-G^\infty(q)}(\eta-q)\text{ almost surely}
            \label{eq_R_minus_R_hat_avg_T_greater_than_q_plus_nu_0}
        \end{align}
        as $m\to\infty.$
        Denote $\nu_0=\left[\left\{G^{\infty}\left(\lambda_{\mathrm{OR}}^{\infty}\right)-G^{\infty}(\eta)\right\} /\left\{G\left(\lambda_{\mathrm{OR}}^{\infty}\right)-G(q)\right\}\right](\eta-q)$.
        Note that $G^{\infty}(t)$, the cumulative distribution function of $T_j^{\infty}$, is strictly increasing in $t$ over the interval $\left(0, \alpha_*\right)$. It implies that $\nu_0>0$. Hence by (\ref{eq_R_hat_avg_T_equal_q}) and (\ref{eq_R_minus_R_hat_avg_T_greater_than_q_plus_nu_0}), we have
        \begin{align}
        \left|\frac{1}{R} \sum_{j=1}^{R} T_{(j)}-\frac{1}{\widehat{R}} \sum_{j=1}^{\widehat{R}} \widehat{T}_{(j)}\right| \geq \left|1 - \frac{\widehat R}{R}\right|\nu_0 - \bigg|\frac{1}{\widehat R}\sum_{j=1}^{\widehat R}\widehat T_{(j)} - \frac{1}{\widehat R}\sum_{j=1}^{\widehat R} T_{(j)}\bigg|. \label{eq_lower_bound_ii}
        \end{align}
        By (\ref{eq_E_avg_T_hat_minus_T}), we take expectations on both sides of (\ref{eq_lower_bound_ii}) to get 
        \begin{align*}
            \mathbb E\left|\frac{1}{R} \sum_{j=1}^{R} T_{(j)}-\frac{1}{\widehat{R}} \sum_{j=1}^{\widehat{R}} \widehat{T}_{(j)}\right| 
            \geq& \mathbb E\left\{\left|1 - \frac{\widehat R}{R}\right|\cdot I(E_2)\right\}\cdot \nu_0 - \mathbb E\bigg|\frac{1}{\widehat R}\sum_{j=1}^{\widehat R}\widehat T_{(j)} - \frac{1}{\widehat R}\sum_{j=1}^{\widehat R} T_{(j)}\bigg|\\
            =& \mathbb E\left\{1-\widehat R / R\mid E_2\right\}\cdot \mathbb P(E_2)\cdot \nu_0 + o(1)\\
            \geq& \frac{\nu_0\varepsilon_1/2}{1+\varepsilon_1/2} \cdot \mathbb P(E_2) + o(1) > 0.
        \end{align*}
        The result is contradictory to (\ref{eq_consistency_of_avg_T}). Therefore, (ii) does not hold either.  

        We have shown that neither (i) nor (ii) holds, which implies that $\lim_{m\rightarrow \infty}\mathbb E|R/\widehat{R} - 1| = 0$. Similarly, we can obtain that $\lim_{m\rightarrow \infty}\mathbb E|V/\widehat{V} - 1| = 0$. The details are omitted.
    \end{proof}

\section{Competing methods}\label{sec_competing_methods}
We summarize the details of some of the competing methods below.

\subsection{The PLACO method}
The PLACO method \citep{ray2020powerful} is designed to detect pleiotropic variants associated with two traits under a composite null hypothesis. Let $Z_{1j}$ and $Z_{2j}$ denote the marginal GWAS $Z$-statistics for SNP $j$ in two traits. PLACO tests
\[
H_{0j}: \beta_{1j}\beta_{2j}=0
\quad \text{versus} \quad
H_{1j}: \beta_{1j}\beta_{2j}\neq 0,
\]
where the null includes variants associated with neither trait or only one trait.

The test statistic is the product of the two $Z$-statistics,
\[
T_j = Z_{1j}Z_{2j}.
\]
Under the composite null, $T_j$ follows a mixture of product-normal distributions corresponding to the sub-null states
\[
(0,0),\quad (0,1),\quad (1,0).
\]
PLACO approximates the null distribution of $T_j$ using genome-wide summary statistics and computes a pleiotropy $p$-value for each SNP. If the two GWAS have overlapping samples or correlated traits, PLACO first decorrelates the two $Z$-statistics using an estimated correlation matrix. Finally, SNPs with PLACO $p$-values below a genome-wide significance threshold, such as $5\times 10^{-8}$, are declared pleiotropic.

\subsection{The Primo method}
The Primo method \citep{gleason2020primo} integrates multiple sets of GWAS and omics QTL summary statistics to identify joint association patterns and provide mechanistic interpretation. Suppose there are $J$ studies or traits, and let
\[
T_i=(T_{i1},\ldots,T_{iJ})
\]
denote the vector of association statistics for SNP $i$. Each SNP can belong to one of
\[
K=2^J
\]
possible association patterns. Let $Q=(q_{kj})$ be the $K\times J$ binary matrix of all possible patterns, where $q_{kj}=1$ means association with trait $j$ under pattern $k$.

Primo models the posterior probability that SNP $i$ belongs to pattern $k$ as
\[
P(a_i=k\mid T_i)
=
\frac{\pi_k D_k(T_i)}
{\sum_{b=1}^K \pi_b D_b(T_i)},
\]
where $\pi_k$ is the genome-wide proportion of SNPs in pattern $k$, and $D_k(T_i)$ is the pattern-specific multivariate density. Primo first estimates marginal null and alternative densities for each study, then estimates the pattern proportions $\pi_k$ by an EM algorithm.

For a biological query of interest, such as association with a complex trait and at least one omics trait, Primo sums the posterior probabilities over the corresponding set of patterns. If $\widehat P_i$ denotes this collapsed posterior probability, the estimated FDR at threshold $\lambda$ is
\[
\widehat{\mathrm{FDR}}(\lambda)
=
\frac{\sum_i (1-\widehat P_i)I(\widehat P_i\ge \lambda)}
{\sum_i I(\widehat P_i\ge \lambda)}.
\]
SNPs with $\widehat P_i\ge \lambda$ are selected, where $\lambda$ is chosen to control the estimated FDR. In gene regions containing known trait-associated SNPs, Primo further performs conditional association analysis to reduce spurious multi-omics associations caused by linkage disequilibrium.

\subsection{The QCH method}
The QCH method \citep{mary2022querying} provides a general framework for testing composed hypotheses using multiple sets of $p$-values. Suppose each item $i$ has $Q$ $p$-values,
\[
P_i=(P_{i1},\ldots,P_{iQ}).
\]
For each test $q$, let $Z_{iq}=0$ denote the null state and $Z_{iq}=1$ denote the alternative state. Then each item belongs to one of
\[
2^Q
\]
configurations
\[
c=(c_1,\ldots,c_Q)\in\{0,1\}^Q.
\]

QCH defines a composed null and alternative by partitioning the configuration space into two sets,
\[
C_0 \cup C_1=\{0,1\}^Q,\qquad C_0\cap C_1=\emptyset.
\]
For example, the intersection-union alternative that an item is non-null in all $Q$ studies corresponds to
\[
C_1=\{(1,\ldots,1)\}.
\]

QCH fits a joint mixture model
\[
P_i \sim \sum_{c\in\{0,1\}^Q} w_c \psi_c,
\]
where $w_c$ is the proportion of items in configuration $c$. Under a conditional independence assumption, the component density is written as
\[
\psi_c(P_i)
=
\prod_{q:c_q=0} f_0^q(P_{iq})
\prod_{q:c_q=1} f_1^q(P_{iq}),
\]
where $f_0^q$ is the null density and $f_1^q$ is the alternative density for the $q$th $p$-value set.

After estimating the marginal alternative densities and the configuration proportions by EM, QCH computes the posterior probability that item $i$ satisfies the composed alternative:
\[
\widehat s_i
=
\sum_{c\in C_1}\widehat P(Z_i=c\mid P_i).
\]
Items are ranked by $\widehat s_i$. For a threshold $t$, the estimated FDR is
\[
\widehat{\mathrm{FDR}}(t)
=
1-
\frac{1}{N(t)}
\sum_{i:\widehat s_i>t}\widehat s_i,
\qquad
N(t)=\sum_i I(\widehat s_i>t).
\]
The final rejection set is obtained by choosing the smallest threshold $t$ such that the estimated FDR is controlled at the target level.

\subsection{The Cartesian HMM method}

The Cartesian hidden Markov model (Cartesian HMM) method \citep{wang2019replicability} is designed for replicability analysis across two GWAS studies while accounting for local dependence among adjacent SNPs. Let $p_{1j}$ and $p_{2j}$ denote the $p$-values for SNP $j$ in the two studies, and let
\[
z_{ij}=\Phi^{-1}(1-p_{ij}), \qquad i=1,2,
\]
be the corresponding one-sided $z$-scores. For study $i$ and SNP $j$, let $H_{ij}=1$ indicate that SNP $j$ is associated with the phenotype in study $i$, and let $H_{ij}=0$ otherwise. The no-replicability null hypothesis is
\[
H^{\mathrm{NR}}_{0j}:
(H_{1j},H_{2j})\in\{(0,0),(1,0),(0,1)\},
\]
whereas the replicability alternative is
\[
H^{\mathrm{R}}_{1j}:
(H_{1j},H_{2j})=(1,1).
\]

To model local dependence, the Cartesian HMM assumes that the joint latent states
\[
\{(H_{1j},H_{2j})\}_{j=1}^m
\]
form a stationary, irreducible, and aperiodic four-state Markov chain with state space
\[
\{(0,0),(1,0),(0,1),(1,1)\}.
\]
The transition probabilities are given by
\[
A_{uv}
=
P\{(H_{1,j+1},H_{2,j+1})=v\mid (H_{1j},H_{2j})=u\},
\]
where $u$ and $v$ range over the four joint states.

Conditional on the latent states, the observed $z$-scores are assumed to be independent across studies and SNPs. Specifically,
\[
Z_{ij}\mid H_{ij}
\sim
(1-H_{ij})f_{i0}+H_{ij}f_{i1},
\]
where $f_{i0}$ is the null density and $f_{i1}$ is the non-null density for study $i$. In practice, the method assumes
\[
f_{10}=f_{20}=N(0,1),
\]
and models the non-null densities parametrically as
\[
f_{11}=N(\mu_1,\sigma_1^2),
\qquad
f_{21}=N(\mu_2,\sigma_2^2).
\]

The testing statistic is the replicated local index of significance, defined as the posterior probability that SNP $j$ is not replicable:
\[
\mathrm{repLIS}_j
=
P\{H^{\mathrm{NR}}_{0j}\text{ is true}\mid (z_{1k},z_{2k})_{k=1}^m\}.
\]
Equivalently,
\[
\mathrm{repLIS}_j
=
P\{(H_{1j},H_{2j})\in\{(0,0),(1,0),(0,1)\}
\mid (z_{1k},z_{2k})_{k=1}^m\}.
\]
These posterior probabilities are computed efficiently using the forward--backward algorithm.

The repLIS procedure ranks SNPs by increasing $\mathrm{repLIS}_j$. Let
\[
\mathrm{repLIS}_{(1)}\le \cdots \le \mathrm{repLIS}_{(m)}
\]
be the ordered values. For a target FDR level $\alpha$, the rejection number is chosen as
\[
\ell
=
\max\left\{
t:
\frac{1}{t}\sum_{j=1}^t \mathrm{repLIS}_{(j)}
\le \alpha
\right\}.
\]
The method then rejects $H^{\mathrm{NR}}_{0(j)}$ for $j=1,\ldots,\ell$, declaring these SNPs replicable across the two studies.

When the Cartesian HMM parameters are unknown, they are estimated by an EM algorithm. The estimated parameters include the initial state probabilities, the transition matrix of the four-state Markov chain, and the parameters of the study-specific non-null normal densities. The estimated parameters are then plugged into the forward--backward algorithm to obtain the data-driven $\widehat{\mathrm{repLIS}}_j$ values.

The Cartesian HMM method exploits the clustering of associated SNPs along the genome and can improve power relative to methods that ignore local dependence. However, it is primarily developed for two-study replicability analysis, uses parametric non-null densities, and a direct extension to many studies would require a high-dimensional Markov chain with exponentially many latent states.

\subsection{The AdaFilter procedure}
The AdaFilter procedure \citep{wang2022detecting} tests partial conjunction null hypotheses with $n$ studies: for some $1 \leq r \leq n$, the null hypothesis for the $j$th SNP is 
\begin{align*}
    H_{0j}^{r/n}: \text{ fewer than } r \text{ out of } n \text{ hypotheses for SNP $j$ are non-null.}
\end{align*}
In this paper, we study the case where $r = n$.

AdaFilter for FDR control works as follows:
\begin{itemize}
    \item {\it Step 1.} For each $j$, order the $p$-values from $n$ studies as $p_{(1)j} \leq p_{(2j)} \leq \dots \leq p_{(nj)}$ and construct the filtering and selection ``$p$-values'' as
    \begin{align*}
        F_j =& (n - r +1)p_{(r-1)j},\\
        S_j =& (n-r+1)p_{(r)j}.
    \end{align*}
    \item {\it Step 2.} Rank the selection $p$-values as $S_{(1)} \leq  \dots \leq S_{(n)}$ with corresponding null hypotheses $H_{0(1)}^{r/n}, \ldots, H_{0(n)}^{r/n}$. For each $j = 1,\ldots, m$, construct an AdaFilter adjustment number
    \begin{align*}
        M_{(j)}^{\rm AF} = \sum_{j'=1}^m I(F_{j'} \leq S_{(j)}).
    \end{align*}
    \item {\it Step 3.} Construct the AdaFilter BH adjusted $p$-value for $H_{0(j)}^{r/m}$ as
    \begin{align*}
        p_{(j)}^{\rm BH} = \min \left\{\min_{j' \geq j} \left\{S_{(j')}\frac{M_{(j')}^{\rm AF}}{j'}\right\}, 1\right\},
    \end{align*}
    and reject the null hypotheses with AdaFilter adjusted $p$-values less than $q$.
\end{itemize}

\subsection{Other methods}
For details of the competing methods including \textit{ad hoc} BH \citep{benjamini1995controlling}, MaxP, MaRR \citep{philtron2018maximum}, radjust \citep{bogomolov2018assessing}, JUMP \citep{lyu2023jump} and STAREG \citep{li2024stareg}, please refer to the Supplementary Materials of \citet{li2024stareg}.

\end{document}